\documentclass[journal,10pt,onecolumn]{IEEEtran}

\usepackage{setspace}
\usepackage{multicol}
\usepackage{array}

\usepackage{epsfig}
\usepackage{framed}
\usepackage[utf8]{inputenc}
\usepackage{amssymb}
\usepackage{amsthm}
\usepackage{amsmath}
\usepackage{amsfonts}
\usepackage{graphicx}
\usepackage{comment}
\usepackage{float}
\usepackage{subfigure}
\usepackage[ruled]{algorithm2e}
\allowdisplaybreaks

\usepackage{enumitem}
\usepackage{mathrsfs}
\usepackage{graphicx}
\usepackage{multirow}
\usepackage{color,soul}
\usepackage{color}
\usepackage{xcolor}

\usepackage{float}
\usepackage{enumitem}
\usepackage{adjustbox}
\usepackage{algorithmicx}
\usepackage{hyperref}
\usepackage{algcompatible}
\usepackage[noend]{algpseudocode}
\usepackage[outercaption]{sidecap}
\DeclareMathAlphabet{\mathpzc}{OT1}{pzc}{m}{it}
\algrenewcommand\alglinenumber[1]{\scriptsize #1:}
\definecolor{mygray}{gray}{0.6}

\newtheorem{definition}{Definition}
\newtheorem{lemma}{Lemma}
\newtheorem{theorem}{Theorem}
\newtheorem*{theorem*}{Theorem}
\newtheorem{remark}{Remark}

\newcommand{\pr}[1]{\mathrm{Pr}#1} 

\newcommand{\mc}[1]{\mathcal{#1}}
\newcommand{\ms}[1]{\mathsf{#1}}
\newcommand{\mb}[1]{\mathbf{#1}}

\newcommand{\ekem}{\mathsf{kem.Enc}}
\newcommand{\edem}{\mathsf{SE.Enc}}

\newcommand{\dkem}{\mathsf{kem.Dec}}
\newcommand{\gkem}{\mathsf{kem.Gen}}

\newcommand{\ck}{\mathsf{ckem}}
\newcommand{{\cksch}}{c\mathcal{KEM}} 
\newcommand{\ckg}{\mathsf{ckem.Gen}}
\newcommand{\cke}{\mathsf{ckem.Enc}}
\newcommand{\ckd}{\mathsf{ckem.Dec}}

\newcommand{\pk}{\mathsf{pkem}}
\newcommand{{\pksch}}{p\mathcal{KEM}} 
\newcommand{\pkg}{\mathsf{pkem.Gen}}
\newcommand{\pke}{\mathsf{pkem.Enc}}
\newcommand{\pkd}{\mathsf{pkem.Dec}}

\newcommand{\remove}[1]{}
\newcommand{\mg}{\color{magenta}}

\newcommand{\rd}{\color{red}}

\newcommand{\kem}{\mathsf{kem}}
\newcommand{{\kemsch}}{\mathcal{KEM}} 
\newcommand{\kemg}{\mathsf{kem.Gen}}
\newcommand{\keme}{\mathsf{kem.Enc}}
\newcommand{\kemd}{\mathsf{kem.Dec}}
\newcommand{\keml}{\mathsf{kem.Len}}

\newcommand{\ikeml}{\mathsf{ikem.Len}}
\newcommand{\pkeml}{\mathsf{pkem.Len}}
\newcommand{\ikem}{\mathsf{ikem}}
\newcommand{\ikemsch}{i\mathcal{KEM}}       
\newcommand{\ikemg}{\mathsf{ikem.Gen}}
\newcommand{\ikeme}{\mathsf{ikem.Enc}}
\newcommand{\ikemd}{\mathsf{ikem.Dec}}

\newcommand{\ike}{\mathsf{HE}}
\newcommand{\gike}{\mathsf{HE.Gen}}
\newcommand{\eike}{\mathsf{HE.Enc}}
\newcommand{\dike}{\mathsf{HE.Dec}}

\newcommand{\ddem}{\mathsf{SE.Dec}}

\newcommand{\dem}{\mathsf{SE}}
\newcommand{{\demsch}}{\mathcal{SE}}      
\newcommand{\ldem}{\mathsf{SE.Len}}

\newcommand{\psucc}{$P_{Succ}(k,c)$}
\newcommand{\fpsucc}{$P_{Succ}$}
\newcommand{\fdpsucc}{$P_{Succ}^{q_d}$}
\newcommand{\ksucc}{$P_{S}$}

\newcommand{\ckem}{\mathsf{Comb}}
\newcommand{\dckem}{\mathsf{Comb.Dec}}
\newcommand{\eckem}{\mathsf{Comb.Enc}}
\newcommand{\gckem}{\mathsf{Comb.Gen}}

\newcommand{\ind}{\mathrm{KIND}}
\newcommand{\pind}{\mathrm{pKIND}}

\newcommand{\eps}{\epsilon}

\newcommand{\x}{\mathbf{x}}
\newcommand{\X}{\mathbf{X}}
\newcommand{\y}{\mathbf{y}}
\newcommand{\Y}{\mathbf{Y}}
\newcommand{\z}{\mathbf{z}}
\newcommand{\Z}{\mathbf{Z}}

\newcommand{\e}{\mathbf{e}}

\providecommand{\eqref}[1]{(\ref{#1})}
\newtheorem{construction}{Construction}
\newtheorem{clam}{Claim}

\newtheorem{coro}{Corollary}
\newcommand{\msa}{\mathsf{D}}
\newcommand{\qres}{  quantum-resistant}

\begin{document}

\title{CCA-Secure Hybrid Encryption in  Correlated Randomness Model and  KEM  Combiners}

\author{Somnath~Panja,
        Setareh~Sharifian,
        Shaoquan~Jiang,
    and~Reihaneh~Safavi-Naini 
\thanks{Somnath Panja and Reihaneh~Safavi-Naini are with the University of Calgary, Canada.}
\thanks{Setareh Sharifian is with Intel Corporation.}
\thanks{Shaoquan~Jiang is with the University of Windsor, Canada.}}




\maketitle
\noindent
\begin{abstract}
A hybrid encryption (HE) system is an efficient public key encryption system for arbitrarily long messages.
An HE system consists of  a public key component called
key encapsulation mechanism (KEM), and a symmetric key component called data encapsulation mechanism (DEM). The HE encryption algorithm uses a KEM generated key k to encapsulate the message using DEM, and send the ciphertext together with the encapsulaton of k,  to the decryptor  who decapsulates k and uses it to decapsulate the message using the corresponding KEM and DEM components.  The KEM/DEM composition theorem proves that if KEM and DEM satisfy well-defined security notions, then HE will be secure with well defined security. 

We introduce HE in correlated randomness model  where the encryption and decryption algorithms have samples of correlated random variables that are partially leaked to the adversary. Security of the new KEM/DEM paradigm is  defined against computationally unbounded or polynomially bounded adversaries. We define iKEM and cKEM with respective information theoretic computational security,  and prove a composition theorem for them and a computationally secure DEM, resulting in secure HEs with proved computational  security (CPA and CCA) and without any computational assumption.
We construct two iKEMs that provably satisfy the required security notions of the composition theorem. 
The iKEMs are used to construct two efficient quantum-resistant HEs when used with an AES based DEM.
We also define and construct {\em combiners} with proved security that combine the new KEM/DEM paradigm of HE with the traditional public key based paradigm of HE.

\end{abstract}

\IEEEpeerreviewmaketitle

\section{Introduction}

A {\em hybrid Encryption (HE) system}   is a public-key encryption system with two components: a public-key  {\em  key encapsulation mechanism (KEM)}  that generates a pair $(k, c_1)$ where $k$ is a secret key and $c_1$  is the  encapsulation  of $k$  under the KEM's public-key,  and 
  an  efficient  symmetric key 
component  called  {\em  data encapsulation
mechanism (DEM)} that will use  $k$  to efficiently encrypt  an arbitrary long message $m$, and generate the ciphertext $c_2$.
  Decryption algorithm  has the private key of KEM   and takes $(c_1,c_2)$ as input. It   decapsulates   $c_1$ to find  $k$  and  uses it to decrypt  $c_2$, and recover $m$.
  This is an attractive construction that effectively provides a computationally efficient public key encryption system  for arbitrarily long messages, by using the  computationally expensive public key   KEM  once,  and 
  encrypt long messages by employing a computationally efficient  DEM that can  be constructed using efficient and standardised secure symmetric key ciphers such as AES (Advanced Encryption Standard)  in one of the known modes of operation such as  counter mode.
Cramer and Shoup     \cite{Cramer2003DesignAA} 
 defined KEM/DEM paradigm, formalized security of KEM  and DEM, 
 and 
proved a general composition theorem that   shows  that   if KEM is CCA (chosen ciphertext attack) secure, and DEM is  a one-time symmetric key encryption with CCA security, 
 then    the resulting hybrid 
 encryption system will be CCA secure (see section \ref{sec:pre} for definitions).  
This level of security is known as the gold standard of security for modern encryption systems. 
KEM/DEM paradigm has been widely studied and more refined 
 notions of security for KEM have been proposed and the  corresponding composition theorems for  HE have been proved \cite{HERRANZ20101243,kiltz2006chosen}.  
There is  a large body  of work on the construction of  KEM  \cite{kurosawa2004new,herranz2006kurosawa,Masayuki2005,kiltz2006chosen,Shacham2007} that are 
all  public key based and rely on
 computational assumptions. KEM has been  widely used   
for securing communication over the Internet including 
as part of TLS  (Transport Layer Security) \cite{SchwabePQTLS2020}. 

{\em Quantum-resistant security} of an HE system requires  \qres~ security of KEM and DEM. Shor's invention of efficient quantum algorithms for integer factorization and discrete logarithm problems \cite{Shor1994} has made KEM constructions that rely on these assumptions, and constitute all existing KEMs in practice, insecure. 
KEM has been one of the first cryptographic  primitives that has been standardized by NIST (National Institute of Standards and Technology) post-quantum cryptography standardization effort
\cite{bos2018crystals,nistpqccompet}.
%
 %
 DEM component of an HE system  uses 
symmetric  block cipher algorithms such as AES, for which the main known quantum attack is
the speed-up for secret key search
that is offered by the Grover's algorithm \cite{grover1996}.
This speed-up however can be compensated by doubling the length of the secret key and so the research on
\qres~security of
 KEM/DEM paradigm  has primarily  focused on the \qres~ security of KEM.
%
%
%

{\em  Information theoretic key agreement.}
 Our main observation is that KEM is  effectively a  one-way secret key agreement (OWSKA) algorithm, a widely studied topic in information theoretic cryptography,  
but  with a somewhat different
definition of security.

Information theoretic key agreement
was first introduced 
by Maurer~\cite{Maurer1993} and Ahlswede~\cite{Ahlswede1993} (independently) 
in what is known as the 
{\em source model}, where Alice and Bob have samples  of two correlated random variables $\X $ and $\Y $ that are distributed according to $P_{\X\Y\Z}$ and  are partially leaked to Eve through the variable $\Z $. 
The  probability distribution
$P_{\X\Y\Z}$ is public but the concrete samples $\x$, $\y$ and $\z$ are private to Alice, Bob and Eve, respectively. 
There is a long line of research on deriving fundamental results on the possibility  of secret key agreement, bounds on rate and capacity of information theoretic key agreement in this model and its variations, 
and providing constructions for    optimal (capacity achieving) systems~\cite{holenstein2005one,holenstein2006strengthening,renes2013efficient,Chou2015a}, together with  the finite  length analysis of the constructions \cite{holenstein2006strengthening,sharif2020}.


Information theoretic key agreement  has also been considered in fuzzy extractor (FE) setting \cite{eurocryptDodisRS04} where Alice and Bob, respectively, have samples $w$ and $w'$ of the same randomness source, satisfying $dist(w,w')\leq t$ where $dist(.,.)$ is a distance function.
FE setting can be seen as a special case of the source model 
where  $\x $  and  $\y $ are samples  of the same source with a guaranteed upper bound on the distance between the two samples, and there is no initial information leakage to the adversary ($\Z =0$). One of the main application areas of FE is key establishment using sources that employ biometric data as the source of randomness. 
Security model of 
 FE is in part influenced by capturing attacks on biometric systems in practice \cite{DodisORS08,boyen2004reusable,boyen2005secure,dodis2006robust,canetti2016reusable}.

A third important 
direction in the study of information theoretic key agreement  is quantum key distribution (QKD) protocols that use quantum theoretic assumptions as the basis of security.
Protocols such as BB84 QKD  \cite{BB84}, use communication  over a quantum channel  to generate correlated random variables  between  two parties, which is later {\em reconciled}  into a shared  secret string that is partially leaked to Eve, and is  used to extract a shared (close to) random key 
between the two parties. 

In all above settings, there is an initial correlated randomness between Alice and Bob that is leveraged to  establish an information theoretically secure  shared secret   key. 
Definitions of security in these settings range from security against a passive eavesdropping adversary \cite{Maurer1993,Ahlswede1993,holenstein2005one,holenstein2006strengthening,sharif2020,eurocryptDodisRS04,boyen2004reusable,canetti2016reusable}, to security  against an active attacker with different levels of access to the system and communication channels~\cite{Maurer1997authencation,maurer2003authen2,ska2023,dodis2006robust}.   
 In all cases, security is against a computationally unbounded adversary and so  the protocol remains secure against  an adversary with access to a quantum computer.
 
 Extending secure key agreement protocols with information theoretic security,
to the establishment of  {\em secure  message transmission 
channels} 
using KEM/DEM approach, 
will allow 
the wealth of research and development in information theoretic key agreement protocols  to be used in  \qres~ cryptographic systems.
\remove{ Our goal is  to  develop 
a 
KEM/DEM-like paradigm for correlated randomness setting, and establish 
a secure message transmission channel using KEMs that are inspired by information theoretic key agreement.}

\noindent
{\em Cryptographic combiners}  combine 
cryptographic schemes with the same functionality 
into a single scheme 
with the guarantee that  the combined scheme is secure 
 if at least one of the component schemes is secure.
 Combiners mitigate the risk of 
 possible design flaws, attacks and breaks of each of the component cryptographic schemes, and provide  robustness for security systems. 
\remove{
Combining the two KEM/DEM paradigm can reduce to combing KEMs  if suitable composition theorem for the new setting can be proved.

They can be also be used to deflect the need to choose among primitives that rely on less studied computational assumptions by combining them into a single primitive, hence providing higher trust in  security systems.
}
Combiners for public key KEMs have been introduced,   their security properties have been formalized,  and secure constructions for KEM combiners have been proposed    
\cite{giacon2018kem,bindel2019hybrid}.
Cryptographic combination of  public key KEM  with KEMs with information theoretic security will seamlessly 
 integrate the new KEMs into the existing applications of KEM and expand  the range of KEMs that are available in designing cryptographic systems.
 
%
\subsection{ Our Results}
 We propose  KEM/DEM paradigm in {\em correlated randomness model} (which in cryptography, is  also referred to as {\em preprocessing model}\footnote{This is because correlated randomness is generated in an initialization stage and before the actual algorithms start.}).
 We 
 define security and prove a composition theorem that relates security of the HE to the security of the  KEM and DEM components.  

  \vspace{0.5 mm}
 {\em  Notation: To make distinction between traditional public key KEMs and KEMs in the new setting, we use {\em pKEM} to denote a KEM scheme in preprocessing model, and reserve {\em iKEM} and {\em cKEM} to refer to the information theoretic and computationally secure versions of pKEM.}

  The new paradigm allows KEM and DEM components to be defined with security against a computationally unbounded, or computationally bounded adversary.  While one can define pKEM and associated DEM with security against information theoretic and computational adversaries,  
  our focus  is on the design of an efficient  \qres~ encryption system (HE) that can be used in practice, and so we consider composition of iKEMs (KEMs with information theoretic security)  and 
  DEMs with computational security.
 We
design two iKEMs with proved security in our proposed security models, one with security against passive adversaries, and one with security against  active adversaries that tamper with the communication channel.  The two iKEMs will have CEA (Chosen Encapsulation Attack) and CCA (Chosen Ciphertext Attack) security, respectively, and when used with a DEM with appropriate security  will result in an HE with CPA (Chosen Plaintext Attack) and CCA security, respectively.
We also define and construct cryptographic combiners that combine a  public key KEM and an iKEM.
%
More details below.

\vspace{1mm}
\noindent
{\bf KEM/DEM  in correlated randomness model.}
A  KEM  in correlated randomness  model   is a tuple of algorithms denoted by $\pksch=  (\pkg, \pke, \pkd)$, where  $\pkg$ is a correlation generation algorithm  that  takes a distribution $P_{\X\Y\Z}$,  generates   correlated random  samples $\x$, $\y$ and side information $\z$ for Alice,  Bob, and 
Eve, respectively, 
and privately delivers the samples to the corresponding parties; 
 $\pke$ is an  {\em encapsulation algorithm} that uses the private sample of Alice and generates a pair  $(k, c_1)$, where $k$ is a random session key for DEM, and $c_1$ is an (encapsulation) ciphertext; $\pkd$ is a {\em decapsulation algorithm }  that uses $c_1$ and the private sample of Bob   to recover $k$. 
 
 Security  of KEM is defined using  
 {\em key indistinguishability games}  between a challenger and  an adversary (Figure \ref{fig1:thm1}). The adversary's power is modelled by its 
 query access
 to the {\em encapsulation oracle } and  {\em decapsulation oracle}. An oracle implements its corresponding algorithm and has access to the private information of the party that legitimately uses the algorithm, and so the encapsulation and decapsulation oracles have the private random samples of Alice and Bob, respectively. The oracles correctly answer queries of the adversary as defined by the security game.
We define these security games  similar to the corresponding ones in public key KEMs \cite{kiltz2006chosen,Cramer2003DesignAA},  with the difference that in public key KEM, the encapsulation algorithm has a public key for  encapsulation 
and so the adversary can  freely access the encapsulation algorithm, 
while in pKEM, the  encapsulation algorithm uses the private sample of Alice, and the adversary can query 
the {\em encapsulation oracle}.  
%
 A ({\em chosen encapsulation attack (CEA) }) query to the encapsulation oracle 
results in an output $(k, c_1)$.
Decapsulation  queries, also referred to as {\em chosen ciphertext attack (CCA)} queries, are the same as in public-key KEMs and allow the adversary to  verify validity of a chosen pair $(k', c')$ against the decapsulation algorithm when using the  private sample of Bob, 
and the response is either 
a key or $\perp$.
The two security notions of IND-CEA (indistinguishability against CEA) and IND-CCA (indistinguishability against CCA) capture indistinguishability of the final key from a uniform random string of the same length, when the attacker has access to CEA, or both CEA and CCA, queries 
respectively.
 %
%
%
%
Adversary can be 
 computationally unbounded (information theoretic), or its computation be  bounded by a polynomial function of the system's security parameter (computational). The number of allowed queries in the two cases are different: for information theoretic adversary the number of allowed queries is a predefined constant (system parameter), while for computational adversary, it is a  polynomial function of the security parameter of the system.
We use iKEM  to denote {\em information theoretically secure pKEMs} 
where the adversary is computationally 
unbounded, 
and use cKEM to refer to  {\em computationally secure pKEM}, 
where the adversary is computationally 
bounded.  This latter is to distinguish computationally secure pKEMs from traditional public key KEMs,  both providing security against a polynomial time adversary but cKEM using an initial correlated randomness instead of a public key.  

We define DEM and  its security
 against a computationally bounded adversary,  
 the same as DEMs  
 in  public-key setting \cite{Cramer2003DesignAA}.
 DEM  security notions are 
  variations of 
IND-CPA (indistinguishability against CPA) security and IND-CCA (indistinguishability against CCA) security for encryption systems.
%
DEM security can also be defined  against a computationally unbounded adversary.  
Our definition of  computationally secure DEM however is motivated by  our goal of constructing 
 \qres~ HE schemes that 
 use a short (constant length) key to encrypt  arbitrary long messages. 

\vspace{1mm}
\noindent
 {\bf Composition Theorem.}
 The following composition  theorem (which is a restatement of Theorem~\ref{theo:composition}) proves (computational) security of an  HE system that is obtained by the composition of a pKEM (iKEM or cKEM) and a computationally secure DEM.
 
 \begin{theorem*} 
 Let $\cksch$ and $\ikemsch$ be a 
cKEM and an iKEM, 
respectively,
and $\demsch$ denote a one-time symmetric key encryption scheme
that is compatible with the corresponding $\cksch$ or  $\ikemsch$.
Then the following composition results hold for the hybrid encryption  in preprocessing model, against a {\em computationally bounded} adversary with access  to the following queries for HE:    $q_e$    
encapsulation and  $q_d$ decapsulation queries when $\ikemsch$ is used, and  polynomially bounded number of queries  for both types of queries, when $\cksch$ is used. 
{
\begin{eqnarray*}
&& \hspace{-3.5mm}\mbox{1. IND-CEA } \cksch + \mbox {IND-OT } \demsch \rightarrow  \mbox{IND-CPA } \ike_{\cksch,\demsch}\\
&&\hspace{-3.5mm}\mbox{2. IND-CCA } \cksch + \mbox{IND-OTCCA }  \demsch   \rightarrow \mbox{IND-CCA } \ike_{\cksch,\demsch}\\
&&\hspace{-3.5mm}\mbox{3. IND-}q_e\mbox{-CEA } \ikemsch + \mbox{IND-OT } \demsch \rightarrow \mbox{IND-}q_e\mbox{-CPA } \ike_{\ikemsch,\demsch} \\
&&\hspace{-3.5mm}\mbox{4. IND-}(q_e; q_d)\mbox{-CCA } \ikemsch + \mbox{IND-OTCCA } \demsch \rightarrow 
\mbox{IND-}(q_e; q_d)\mbox{-CCA } \ike_{\ikemsch,\demsch}.
\end{eqnarray*}
}
\end{theorem*}


IND-OT and IND-OTCCA refer to indistinguishability security for one-time secure DEM with CPA and CCA security, respectively (see Definition~\ref{demseckind}).

 In all cases, security of the hybrid encryption system  is against  a computationally bounded adversary. 
 In  (1) and (2), $\cksch$ is secure against a computationally bounded adversary who has access to polynomially bounded number of encapsulation and decapsulation queries, and the final HE satisfies CPA and CCA definition of security of computationally secure encryption systems (see Definition in section~\ref{defn:hesecurity}).
In  (3) and (4) however, $\ikemsch$ is secure against a computationally unbounded  adversary with access to a constant number of  encapsulation ($q_e$) and  decapsulation ($q_d$) queries,  and the final HE is {\em bounded CPA and CCA secure}, respectively  \cite{Cramer2003DesignAA}.

\vspace{1mm}
\noindent
{\bf Constructions of iKEM.}
In section \ref{sec:ikem-inst}, we consider the case that the correlated randomness 
 is obtained by repeated sampling a public distribution, 
 \remove{
 $P_{\X\Y\Z}$, where the correlated random variables of Alice, Bob and Eve are $\X$, $\Y$ and $\Z$, respectively, and 
 } 
 and
$P_{\X\Y\Z}=\prod_{i=1}^nP_{X_i Y_i Z_i}$
where 
$P_{X_i Y_i Z_i}=P_{X Y Z}$ for $1 \le i \le n$.
We have 
$\X =(X_1,\cdots,X_n)$, $\Y =(Y_1,\cdots,Y_n)$, $\Z =(Z_1,\cdots,Z_n)$ respectively,
with the corresponding 
private samples, $\x = (x_1,\cdots,x_n)$, $\y = (y_1,\cdots,y_n)$ and $\z = (z_1,\cdots,z_n)$.

 We propose two constructions of iKEM  for Satellite scenario,
Construction~\ref{ikem:cea} and Construction~\ref{ikem:cca}, that provide IND-CEA and IND-CCA security, respectively. Both constructions are based on the OWSKA  in \cite{sharif2020}, 
 where Alice sends a single message to Bob  over a public  authenticated channel. The message includes  information  that will be used for {\em information reconciliation} that enables  Bob to recover Alice's sample with some leakage,
and  the description of a hash function to be used for {\em key extraction}. 
The OWSKA construction uses two   universal  hash functions $h$ and $h'$  for the two tasks.
%
This construction was first proposed 
in \cite{sharifian2021information} 
for an iKEM with IND-CEA security for $q_e$  encapsulation queries (and no decapsulation queries), and used 
two {\em strongly} universal hash functions, $h$  and $h'$.
Construction~\ref{ikem:cea} has the same 
security properties
but uses 
universal hash families. The construction  slightly modifies the initialization process of iKEM that  improves the length  of the established key without affecting security. 
%
%
The encapsulation ciphertext in   Construction~\ref{ikem:cea} 
is   $c=(h(\x , s), s')$, 
 where $s$ and $s'$ are  random strings that are used in   $h$ (reconciliation) and $h'$ (extraction) respectively.
Our observation is  that $s$, the seed for $h$ that is used for reconciliation,    can stay the same in all instances of the protocol and so can be generated and distributed to all parties (including to Eve)  during initialization.

  We prove security of this construction 
{\em for any pair  $h$  and $h'$ of    universal hash functions   with appropriate parameters.}

\vspace{1mm}
The second construction is a {\em pKEM with IND-CCA security} that removes the need for a public authenticated channel 
between Alice and  Bob, and provides security against an adversary who can tamper with the KEM   
ciphertext.
 We define INT-CTXT (ciphertext integrity) for  pKEM  
 (Definition~\ref{def: integrity}) that requires 
any tampering with the cipherext to be detectable by Bob. 
Theorem~\ref{theo:comp} proves  that in preprocessing model,
a KEM that is IND-CEA and INT-CTXT secure, is IND-CCA
secure.   Our IND-CCA secure  Construction~\ref{ikem:cca} 
requires a {\em specific construction of $h$
} (whereas in Construction~\ref{ikem:cea}, $h$ can be any  universal hash function).

 To show IND-CCA security of  the Construction~\ref{ikem:cca}, we show that it  is an IND-CEA and INT-CTXT secure KEM, and so 
 it provides IND-CCA security.  
 The construction is based on the OWSKA construction in \cite{ska2023} that provides security against an active adversary.  Our iKEM construction slightly modifies the reconciliation message of the OWSKA, {\em revises and corrects its security analysis}, and obtains new parameters for the system. 
The encapsulation ciphertext in 
Construction~\ref{ikem:cca} 
 is given by  
$c=(h(\x , (s',s)),s' , s)$  which includes  $s'$  as part of the input to $h$ also.
The 
%
hash function $h$ is designed  to $(i)$ provide
{\em information reconciliation} to allow  Bob 
to securely recover Alice's sample $\x$ 
and $(ii)$   
serves as a MAC (message authentication code) 
to protect integrity of the encapsulation ciphertext. 
The decapsulation algorithm checks the {\em validity of a received encapsulation ciphertext by}
computing the hash function $h$ using the candidate key {$\hat{\x}$}
 that is derived for Alice  and the received $(s',s)$, and compares the result with $h(\x , (s',s))$.
We bound the success probability of the adversary in forging a valid encapsulation ciphertext by bounding the guessing probability of the secret keys 
that are used in the encapsulation and decapsulation algorithms.

In Theorem~\ref{mac2:ctxt},
we prove
integrity of the ciphertext (IND-CTXT) 
of the  iKEM  Construction~\ref{ikem:cca} against an active adversary with access to one encapsulation and $q_d$ decapsulation queries.  
\remove{
To provide security against $q_e  > 1$ encapsulation queries, 
a {\rd $(q_e+1)$-independent   ?? $(q_e+1)$-universal } hash function  with  additional properties can be used. 
Our construction of
}
 The 
$h$ construction in section~\ref{ikem:robcca} can be  extended to provide
security against $q_e  > 1$ queries. 
\remove{
a {\rd $(q_e+1)$-independent   ?? $(q_e+1)$-universal } 
hash function. 
}
The final extracted key length however will be reduced  (almost) linearly with higher $q_e$.   
 We  note that security against $q_e  > 1$ encapsulation queries, is only necessary if the same sample $\x$ is used in multiple instances of HE, and not required in applications such as QKD where each message transmission will use its dedicated quantum communication round (and so new values of $\x$, $\y$ and $\z$). 

\remove{ 
 {\color{olive} SEEMS ALL THIS SHOULD GO
 {\em MAC with partially leaked key.}  
 A result of independent interest is the construction of a new MAC where the key is partially leaked.  
 A traditional MAC system  is a symmetric key cryptosystem  that provides 
 integrity for a message that is sent from Alice to Bob when an active adversary tampers with the communication. 
 In section \ref{sec:pre}, we first define  a one-time MAC in correlated randomness model (Satellite scenario)
\remove{
where Alice  and Bob's secret keys are samples $\x$ and $\y$ of a correlated randomness source
with distribution $P_{\X\Y\Z}$, and the keys are 
partially leaked  to the adversary through $\z$. 
}
with security 
against {\it impersonation} and {\it substitution attacks}, and in Lemma~\ref{lemma:maclemma2}
prove  that the $h$  function in Construction \ref{ikem:cca} gives a secure one-time MAC  
for the case  where the key $\x=\y$  
is partially leaked to the adversary through $\z$, and prove security of the MAC 
against impersonation and substitution attacks.
} 
}

\vspace{1mm}
\noindent
{\bf KEM Combiners.} 
 We define   KEM combiners that securely  combine a pKEM (iKEM or cKEM) and a public key KEM.
In this combination, if at least one of the component KEMs is an iKEM, the resulting KEM will be an iKEM and secure against a computationally unbounded adversary (for fixed number of encapsulation/decapsulation queries) and so {\em a \qres~ KEM}. 
The resulting KEM will also be computationally secure with polynomial (in security parameter) number of encapsulation/decapsulation queries, as long as at least one of the component KEMs is computationally secure (public-key KEM or cKEM). 

 \remove{
 Cryptographic combiners have been used to combine { \textit{component} 
cryptographic schemes} with the same functionality 
 into a single scheme with the same functionality
 and the guarantee that the resulting scheme  is as secure as any of the   component schemes.

Combiners provide robustness against  insecurity and possible break of the component schemes. 
They can be also be used to deflect the need to choose among primitives that rely on less studied computational assumptions by combining them into a single primitive, hence providing higher trust in  security systems.
}
We give two blackbox constructions of 
KEM combiners for an iKEM and a public key KEM that satisfy the above security properties  (information theoretic  security for fixed number of queries and computational  security for polynomial number of queries as long as   the corresponding component KEM is secure).
The constructions are based on 
the \textit{XOR combiner}  and \textit{PRF-then-XOR combiner}  of \cite{giacon2018kem} that were proposed  for public key KEMs. 
We extend these constructions to our setting where one of the KEMs is an iKEM. 
 The \textit{XOR combiner}  XORs the output keys of the component KEMs.  The construction 
maintains IND-CEA security of the resulting KEM (Theorem~\ref{theo:xorcomb}) but will not result in an IND-CCA KEM when the component KEMs are IND-CCA secure. 
The \textit{PRF-then-XOR combiner} uses  
PRFs (Pseudorandom functions).
 A PRF is a family of functions indexed by a secret key, that guarantees that for a uniformly chosen key, the function output is indistinguishable from  the output of  a random function  for an adversary who can see the evaluations of the function 
 on an adaptive adversary chosen set of 
 values (see Definition~\ref{defn:prf}).  
We use two types of PRFs:  with {\em statistical 
indistinguishability for constant number of queries}, and {\em computational indistinguishability for polynomial number of queries} (see Definition \ref{defn:prf}).

 The PRF-then-XOR combiner XORs the outputs of a set of PRFs, each associated with one of the  KEMs,
 where the $i^{th}$  PRF uses the secret key $k_i$ that is the output of the $i^{th}$ KEM, 
and computes the  value of the function on an input that is the concatenation of  the  ciphertexts of all other KEMs  (except the $i^{th}$ one).
  We require PRF with statistical indistinguishability for iKEMs, and with computational indistinguishability for computational KEMs. 
  \remove{
The $i^{th}$  PRF uses the output of the $i^{th}$ KEM, denoted by $k_i$ as the secret key, with an input that is the concatenation of  the  ciphertexts of all other KEMs  (except the $i^{th}$ one).  
}
 Theorems~\ref{theo:xorcomb} and~\ref{thmprfxorcomb} respectively,  prove IND-CEA, and IND-CCA security of the resulting KEMs, 
and relate their security 
to the security 
of the component KEMs and the PRFs.

\vspace{1mm}
\noindent
{\bf Discussion.} 
 Security of KEM/DEM paradigm in correlated randomness  model  
 does not rely  on any (unproven) computationally hard problem. 
  Hybrid encryption system in this model is neither a public key, nor a symmetric key encryption system.
  Rather, it relies on the communicating parties secret inputs (that we refer to as key) 
  that are not identical,
     but are correlated,   and can be partially leaked. The final security of the HE is computational.
  The paradigm provides flexibility to consider security against computationally  unbounded or bounded adversaries for each component (KEM and DEM).
  Our focus on iKEM and computationally bounded DEM is motivated by real-life application of HE in \qres~ systems.

\vspace{1mm}
\noindent
{\bf Organization.} 
Related work is in section \ref{ap:related}. Section~\ref{sec:pre} is preliminaries.
Section~\ref{section:preprocessing} is on KEM in preprocessing model.
Instantiations of iKEM and their security proofs 
are in Section~~\ref{sec:ikem-inst}.  
 Section~\ref{sec:combiner} is on combiners and their constructions. 
 Section \ref{sec:conclusion} provides concluding remarks.

\section{Related work}
\label{ap:related}
\remove{
Security notions of encryption systems have been widely studied in \cite{bellare1997concrete,bellare1998relations,bellare2000authenticated}.
\remove{
for public-key encryption schemes are classified in \cite{bellare1998relations}, and for symmetric encryption schemes are given under a concrete framework in \cite{bellare1997concrete}. Authenticated encryption scheme in a symmetric encryption scheme is consider in \cite{bellare2000authenticated} 
}

Ciphertext unforgeability captures adversary’s inability to generate valid  ciphertexts 
\cite{katz2006characterization}.
Ciphertext integrity was also considered for information theoretic cryptographic channels \cite{fischlin2020information}.
} 
\textit{KEM/DEM  paradigm} has been widely used  {in public key based hybrid encryption
for encrypting arbitrary length messages with proved security. 
The approach was first formalized by Cramer and Shoup \cite{Cramer2003DesignAA}
who proved that
that  a CCA secure KEM and one-time secure CCA symmetric key encryption system (DEM) 
result in a CCA secure 
hybrid  encryption system.
 \remove{
 Kurosawa and Desmedt \cite{kurosawa2004new} constructed a hybrid encryption scheme with CCA security using 
 that shows CCA-secure hybrid encryption  can be achieved using  a weaker security notion for KEM (the KEM in their construction was later shown to be not CCA secure \cite{herranz2006kurosawa}). 
}
The relation between different security notions of KEM and 
 DEM, and the resulting 
  hybrid encryption system is given in \cite{HERRANZ20101243}. 
 There are numerous generic and specific constructions of public-key KEM including \cite{dent2003designer,
 bentahar2008generic,haralambiev2010simple
 }.
   There are also constructions of  
   KEM that use hardness assumptions for which there is no known quantum algorithm. This  includes constructions  
  \cite{ding2012simple,peikert2014lattice,bos2018crystals} that use LWE (Learning with Error) and other lattice based assumptions. 
  \remove{
  (LWE based assumption) and \cite{peikert2014lattice} (lattice based) with CPA security,  and 
  \cite{bos2018crystals} ??MORE REF?? with CCA security.
  } 
  Quantum-resistant secure KEM has been part of NIST  post-quantum competition \cite{nistpqround2}  and CRYSTALS-Kyber is the standardized \qres ~ KEM \cite{bos2018crystals}. 
 KEM combiners
 \remove{
 combine multiple KEMs such that the resulting key is secure if any of the keys is secure 
 }
 are studied in \cite{giacon2018kem,matsuda2018new,bindel2019hybrid,harnik2005robust}.

All above works are in public-key setting. KEM/DEM in correlated randomness setting was introduced in  \cite{sharifian2021information} where authors 
considered passive adversaries with access to 
encapsulation queries, only.
 We extend this work in a number of ways. We consider security against active attackers and prove a general  composition theorem for CCA security of HE,  and construct  a CCA secure  
 iKEM that results in a CCA secure (\qres)  HE. 
We also construct combiners for iKEM and public-key KEM, that when used 
with a 
computationally secure DEM, result in a provably secure CCA encryption system.

\textit{Information theoretic key agreement} %
in source model was first studied by Maurer \cite{Maurer1993}, and Ahlswede and Csisz\'ar \cite{Ahlswede1993}, and has led to a long line of research on this topic and more specific related topics including information reconciliation \cite{bennett1988privacy,renner2004smooth,Holenstein2011,tomami2014}.
 OWSKA 
uses a single message from Alice to Bob to establish a shared key \cite{holenstein2005one,renes2013efficient,Chou2015a,sharif2020}.
Key establishment in correlated randomness model with security against active adversary was studied in \cite{maurer2003authen1,maurer2003authen2,Renner2004exact,kanukurthi2009key}.

\textit{Combining cryptographic primitives} was first considered by Shannon  who studied security of an encryption system that is obtained by  combining multiple encryption systems, and   
suggested 
``weighted sum'' and ``product ciphers'' to combine 
secrecy systems  to achieve 
stronger security \cite{shannon1949communication}. %
Combiners have been studied for numerous cryptographic primitives including encryption systems \cite{even1985power,maurer1993cascade}  and hash functions~\cite{Fischlin2007}. 
Robust combiners for cryptographic systems were   studied by Herzberg \cite{Herzberg20022007}
and later extended
~\cite{harnik2005robust} to include  parallel and cascade constructions, where 
constructions for various primitives including OWF (One Way Functions), signatures and MACs are given.
A {\em robust combiner} for a cryptographic primitive $\cal P$  takes multiple candidate schemes that implement  $\cal P$,  and combine them into a single scheme such that the resulting scheme remains secure even if some of the schemes become insecure. 
%
%
 In a $(k,n)$-robust combiner  \cite{harnik2005robust} security is guaranteed if at least $t$ out of $n$ constructions remains secure.
 
%
Combiners for public key  KEM  was studied  in
Giacon et al.~\cite{giacon2018kem},  
and with security against 
quantum adversaries  were  considered and constructed in Bindel et al.~\cite{bindel2019hybrid}.
}

\textit{Correlated randomness model}  has been used in cryptography 
to remove  impossibility results, including 
 key establishment in presence of computationally unbounded adversaries \cite{Maurer1993}, oblivious transfer \cite{beaver1995precomputing} and multi-party computation (MPC) protocols \cite{bendlin2011semi,ishai2013power,garg2018two}.
  Correlated randomness for key agreement can  be realized in  settings such as biometric authentication, transmission over noisy (wiretapped) channels, and using communication over quantum channel.

\section{Preliminaries}
\label{sec:pre}
  We denote 
 random variables (RVs) with upper-case letters, (e.g., $X$), and their realizations with lower-case letters (e.g., $x$). The 
 probability distribution  associated with 
   a random variable
	   $X$   is denoted by $\mathrm{P}_X(x)=\mathsf{Pr}(X=x)$, and the conditional 
    probability distribution associated with 
    $X$ given  
    $Y$ is denoted by $\mathrm{P}_{X|Y}(x|y)=\mathsf{Pr}(X=x|Y=y)$. 
    \emph{Shannon entropy} of an RV $X$ is defined  by $H(X)=-\sum_x\mathrm{P}_X(x)\log(\mathrm{P}_X(x))$.
The  \emph{min-entropy} $H_{\infty}(X)$ of a random variable $X \in \mathcal{X}$ with 
probability distribution $\mathrm{P}_X$  
is 
$H_{\infty}(X)= -\log (\max_{x} (\mathrm{P}_X({x})))$.
The \emph{average conditional min-entropy}  \cite{DodisORS08} is defined as,
$\tilde{H}_{\infty}(X|Y)= -\log \mathbb{E}_{{y} \leftarrow Y}\max_{{x} \in \mathcal{X}}\mathrm{P}_{X|Y}({x}|{y}).$
The statistical distance between two random variables $X$ and $Y$ with the same domain $\mc T$ is given by ${\rm \Delta}(X,Y)=\frac{1}{2} \sum_{v\in {\cal T}} |\Pr[X=v]-\Pr[Y=v]|$. For an $n$-bit variable $\x$, we use $[x]_{i\cdots j}$ to denote the  block of bits from the $i$th bit to the $j$th bit in  $x.$ For $\ell \in \mathbb{N}$, $U_\ell$ denotes an RV with uniform distribution over $\{0,1\}^\ell$. Vectors are denoted using boldface letters, e.g. $\X=(X_1,\cdots,X_n)$ is a vector of $n$ RVs, and its realization is given by $\x=(x_1,\cdots,x_n)$. 

 To define closeness of two families of  distributions that are indexed by $\lambda$ using  
 the notion of  indistinguishability (statistical and computational), we use  
two classes 
of functions called  ${SMALL}$ and ${NEGL}$ as defined in \cite{pfitzmann2000model}.  
\remove{
We use  
two classes 
of functions called  ${SMALL}$ and ${NEGL}$ as defined in \cite{pfitzmann2000model},  to define closeness of families of  distributions that are indexed by $\lambda$, using  
 the notions of statistical and computational indistinguishability, respectively.
 }
 The class of negligible functions ${NEGL}$, contains all functions $s: \mathbb{N} \to \mathbb{R}_{\ge 0}$  where for every positive polynomial $f(\cdot)$, $\exists n_0 \in \mathbb{N}$ such that $\forall n \ge n_0, |s(n)| < \frac{1}{f(n)}$, where $\mathbb{R}_{\ge 0}$ is the set of non-negative real numbers. 
%
 A set  ${SMALL}$ is a class of \textit{small} functions $\mathbb{N} \to \mathbb{R}_{\ge 0}$ if: $(i)$ it is closed under addition, and $(ii)$ a function $s' \in SMALL$ implies that all functions $f':\mathbb{N} \to \mathbb{R}_{\ge 0}$ with $f' \le s'$ are also in the set $SMALL$.

 {\em Universal hash functions}  have been used to generate  close to uniform RVs from 
non-uniform entropy sources with sufficient min-entropy.
This is proved in 
Leftover Hash Lemma~\cite{impagliazzo1989pseudo}. We use a variant of Leftover Hash Lemma, called Generalized Leftover Hash Lemma~\cite[Lemma~2.4]{DodisORS08}.

\begin{definition} [Universal hash family]\label{defn:uhf}
    A family of hash functions $h:\mathcal{X} \times \mathcal{S} \to \mathcal{Y}$ is called a universal hash family if $\forall x_1,x_2 \in \mathcal{X}$, $x_1 \ne x_2$, we have  $\pr[h(x_1,S)=h(x_2,S)] \le \frac{1}{|\mathcal{Y}|}$, where the probability is over the uniform choices of $\mathcal{S}$.
\end{definition}

 \begin{lemma}[Generalized Leftover Hash Lemma~\cite{DodisORS08}]
 \label{glhl}
Let $h: \mathcal{X} \times \mathcal{S} \rightarrow \{0,1\}^{\ell}$ be a universal hash family. Then for any two variables $A \in \mathcal{X}$ and $B \in \mathcal{Y}$, applying $h$ on $A$ can extract a uniform random variable whose length $\ell$ satisfies the following   $\Delta(h(A, S), S, B; U_\ell, S, B)\le \frac{1}{2}\sqrt{2^{-\tilde{H}_{\infty}(A|B)}\cdot 2^
\ell}$, where $S$ is chosen uniformly from $\mathcal{S}$. 
 \end{lemma}

 For $\lambda\in\mathbb{N}$,  the unary representation of $\lambda$ given by $1^\lambda$, 
is used  
to specify the running time of the algorithm as a function of $\lambda$.  For efficient algorithm, the running time is 
 a polynomial in $\lambda$.  {\em We use $\lambda$ as the {\em security parameter} of the system.}

  An algorithm $\msa$ that  takes inputs $x, y, \cdots$, and generates the output $u$, while having access to oracles 
 $\ms O_1, \ms O_2, \dots$, 
 by $u \gets\msa^{\ms O_1,\ms O_2,...}(x, y, \cdots)$. 
\remove{ We 
denote an algorithm $\msa$ that has access to oracles $\ms O_1, \ms O_2, \dots$,  takes inputs $x, y, \cdots$, and generates the output $u$ by $u \gets\msa^{\ms O_1,\ms O_2,...}(x, y, \cdots)$.}

 \vspace{1mm}
 \noindent
 {\bf KEM and DEM.}   Hybrid encryption and the notion of KEM was first introduced and formalized in \cite{Cramer2003DesignAA}.  Properties of KEM and DEM were formally defined in \cite{HERRANZ20101243}. 

\begin{definition}[KEM distinguishing advantage \cite{HERRANZ20101243}] \label{def:kem} Let $\msa = (\msa_1, \msa_2)$ be 
 an adversary and $\mathsf{kem} = (\gkem, \ekem, \dkem)$ be a KEM with security parameter $\lambda$ and key space $\{0,1\} ^{\keml(\lambda)}$. 
For $atk\in \{cpa, ,cca1,cca2\}$, 
 the \underline{k}ey \underline{ind}istinguishability (kind) advantage of $\ms {kem}$ is defined as
%
{
\begin{equation}\label{eq3}
Adv^{kind\text{-}atk}_{\ms {kem},\msa}(\lambda) \triangleq  |\pr[\ind_{\ms {kem},\msa}^{atk\text{-}0}(\lambda)= 1]-\pr[\ind_{\ms {kem},\msa}^{atk\text{-}1}(\lambda)= 1]|,
\end{equation}}

where the distinguishing game $\ind_{\ms{kem},\msa}^{atk\text{-}b}$ for $b\in\{0,1\}$ is defined in Figure.~\ref{fig:kem}.
 
\begin{figure}[H]
\begin{center}
 \begin{minipage}[c]{0.29\textwidth}

\underline{\textbf{Game} $\ind_{\ms {kem},\msa}^{atk\text{-}b}(\lambda)$} \ \ \ \ \ \ \ \ \ \ \ \ \ \ \ \ 
\begin{algorithmic}[1]
\State $(pk,sk)\stackrel{\$}\gets \gkem(1^{\lambda})$
\State $st\stackrel{\$}\gets \msa_1^{\ms O_1}(pk)$
\State $(k^*,c^*)\stackrel{\$}\gets \ekem(pk)$
\State $k_0\gets k^*$ \State $k_1\stackrel{\$}\gets\{0,1\}^{\keml(\lambda)}$
\State $b'\stackrel{\$}\gets \msa_2^{\ms O_2}(c^*,st,k_b)$
\State Return $b'$
\end{algorithmic}
\end{minipage}
\begin{minipage}[c]{0.39\textwidth}
\vspace{-3.9em}
\underline{\textbf{Oracles}$_{ }$ $\ms O_{1_{}}$ and $\ms O_2$}\\       
       \begin{tabular}{l l l}
    $atk$ & \ $\ms O_1(\cdot)$ & $\ms O_2(\cdot)$ \\
    \hline
     $cpa$& \ $\ \varepsilon$ & $\varepsilon$\\
     $cca1$& \ $
    \dkem(sk,\cdot)$ & $\varepsilon$\\
     $cca2$& \ $
    \dkem(sk,\cdot)$ \ &$
    \dkem(sk,\cdot)$
\end{tabular}
 \end{minipage}
\end{center} 
         \caption{
         The distinguishing game $\ind_{\ms {kem},\msa}^{atk\text{-}b}$, where $b\stackrel{\$}\gets\{0,1\}$, and $atk\in\{cpa,cca1,cca2\}$.
   The decapsulation  oracle  $\dkem(sk,\cdot)$  has the  private  key $sk$.  Oracle output $\ms O_i = \varepsilon, i\in\{1, 2\}$, means $\ms O_i$
returns the empty string $\varepsilon$. $\ms O_2$ 
    cannot be asked to decapsulate $c^*$.}
    \label{fig:kem}
\end{figure}
 \end{definition}
 
A KEM is
IND-CPA (CCA1 or CCA2) secure if for all
polynomial-time adversaries  $\msa$ that corresponds to $atk=cpa$,  $atk=cca1$ or $atk=cca2$, 
the advantage function (in equation~\ref{eq3})  is  negligible in ${\lambda}$. 
{\em In this paper, we only consider CCA2 security, 
and  refer to it as CCA-security. }

 \vspace{1mm}
\noindent
{\em Data Encapsulation Mechanism (DEM)}
is a symmetric key encryption algorithm. 
We use the following definition 
in \cite{HERRANZ20101243}.

\begin{definition}
[Security of DEM: IND-OT, IND-OTCCA, IND-CPA, IND-CCA1, IND-CCA2 \cite{HERRANZ20101243}]\label{demseckind} Let \\$\ms{dem} =  (\ms{dem.Gen},\ms{dem.Enc},\ms{dem.Dec})$ be a DEM scheme with security parameter $\lambda$ and key space $\{0,1\}^{{\ms{dem.Len}}(\lambda)}$,
and let $\msa = (\msa_1, \msa_2)$ be an adversary. For $atk \in\{ot,otcca, cpa, cca1, cca2\}$ and $\lambda\in \mathbb{N}$,  
the \underline{ind}istinguishability (ind) advantage of $\ms{dem}$ is defined as
%
\begin{equation}\label{eq2}
Adv^{ind\text{-}atk}_{\ms{dem},\msa}(\lambda) \triangleq  |\pr[\mathrm{IND}_{\ms{dem},\msa}^{atk\text{-}0}(\lambda)= 1]-\pr[\mathrm{IND}_{\ms{dem},\msa}^{atk\text{-}1}(\lambda)= 1]|,
\end{equation}
where the distinguishing game $\mathrm{IND}_{\ms{dem},\msa}^{atk\text{-}b}$ for $b\in\{0,1\}$ is defined in Figure.~\ref{fig:dem}.
 
\begin{figure}[!ht]
\begin{center}
\begin{minipage}[c]{0.29\textwidth}
{
\underline{\textbf{Game} $\mathrm{IND}_{\ms dem,\msa}^{atk\text{-}b}(\lambda)$} \ \ \ \ \ \ \ \ \ \ \ \ \ \ \ \ \ \ \ 
\begin{algorithmic}[1]
\State $k\stackrel{\$}\gets \ms{dem.Gen}(1^{\lambda})$
\State $(st,m_0,m_1)\stackrel{\$}\gets \msa_1^{\ms O_1}()$
\State $c^*\stackrel{\$}\gets \ms{dem.Enc}(k,m_b)$
\State $b'\stackrel{\$}\gets \msa_2^{\ms O_2}(c^*,st)$
\State Return $b'$
\end{algorithmic}}
\end{minipage}
\begin{minipage}[c]{0.59\textwidth}
\vspace{1em}
{
\parbox{0.51\textwidth}{
\underline{\textbf{Oracles}$_{ }$ $\ms O_{1_{}}$ and $\ms O_2$}   \\
\begin{tabular}{p{2em} p{13em} p{12em}}
    $atk$ &$\ms O_{1}$ &$\ms O_{2}$ \\
    \hline
     $ot$& $\varepsilon$ & $\varepsilon$\\
     $otcca$& $\varepsilon$&$ {\ms{dem.Dec}}(k,\cdot)$\\
     $cpa$&$\ms{dem.Enc}(k,\cdot)$ &$\varepsilon$\\
     $cca1$&$\{\ms{dem.Enc}(k,\cdot),{\ms{dem.Dec}}(k,\cdot)\}$& $\varepsilon$\\
     $cca2$&$\{\ms{dem.Enc}(k,\cdot),{\ms{dem.Dec}}(k,\cdot)\}$& $\{\ms{dem.Enc}(k,\cdot),{\ms{dem.Dec}}(k,\cdot)\}$\\
\end{tabular}
}
}
\end{minipage}
\end{center}
\caption{DEM distinguishing game. Here, $\ms{dem.Enc}(k,\cdot)$ and $\ms{dem.Dec}(k,\cdot)$ are  encryption and  decryption oracles  with 
 key $k$, respectively, and $\varepsilon$  denotes an empty string.}
 \label{fig:dem}
\end{figure}

A {\em DEM is $\sigma(\lambda)\text{-}IND\text{-}ATK$ for $ATK\in\{OT,OTCCA,CPA,CCA1,CCA2\}$} if for \textit{all}  polynomial-time adversaries $\msa$, $Adv^{ind\text{-}atk}_{\ms{dem},\msa}(\lambda)\leq \sigma(\lambda)$, where $\sigma(\cdot)$ is a non-negative negligible function in $\lambda$.
\end{definition}

 The formalization and construction of HE in  \cite{Cramer2003DesignAA} uses \textit{one-time} symmetric key encryption schemes with a specific security definition (adversary with access to decryption oracle only). 
 The  one-time symmetric key encryption that is defined below, is a DEM with OTCCA security.

A {\em one-time symmetric key encryption }$\dem=(\edem,\ddem)$ with  security parameter $\lambda$ and 
the  key space $\{0,1\}^{\ldem(\lambda)}$  consists of two  deterministic\footnote{Thus, {for all $k\in\{0,1\}^{\ldem(\lambda)}$ and
  $m \in \{0,1\}^* $, $\pr [\ddem\big(k,\edem(k,m)\big) = m] = 1$.}} algorithms such that: i) the encryption algorithm $\edem(k,m)$  encrypts a message $m\in\{0,1\}^*$ under a uniformly chosen key $k\in\{0,1\}^{\ldem(\lambda)}$ and outputs a ciphertext $c$, and ii) the decryption algorithm $\edem(c,k)$ that decrypts the ciphertext $c$ using  the key $k$, and either recovers the message $m$,  or outputs a special rejection symbol $\perp$.

Security of $\dem$ 
is tailored for its application in  hybrid encryption systems,  and matches the OTCCA security in Definition \ref{demseckind}.
\remove{{\color{olive} REMOVE
 One-time security against a \textit{passive attack (or OT security) 
}  is in terms of indistinguishability of   two equal size messages $m_0$ and $m_1$ that are chosen by the adversary  when it is   given the ciphretext of one of the two (randomly chosen) plaintexts, and without having access to any  oracle queries.
Security against \textit{active attacks} 
is defined similarly, to the case of 
but  allows the adversary to have access to ciphertext queries  before and after it receives  the challenge ciphertext.
(See Definition~\ref{demseckind}.)
}}

\section{KEM in correlated randomness  model}
\label{section:preprocessing}
 A {\em KEM in correlated randomness model} (also called {\em preprocessing model}) 
has two phases.  In 
the {\em initialization phase} that is 
also called {\em offline phase},
Alice, Bob  and Eve, respectively,    privately receive   $r_A$,  $r_B$ and $r_E$, that is obtained by sampling  
a public 
 joint distribution $P_{\X\Y\Z}$ (e.g. an efficient probabilistic experiment). In the \textit{online phase}  
 Alice and Bob use their private values  in the   encapsulation and decapsulation algorithms, respectively, to  obtain a shared  key. 
Here $r_E$ represents Eve's initial information about Alice and Bob's samples.

\begin{definition}[KEM in Preprocessing Model (pKEM)]  \label{def:ikem} 
KEM  in preprocessing model (pKEM) with security parameter $\lambda$,  joint distribution $P_{\X \Y \Z }$, and
key space $\ms{KeySP_{\pk}}(\lambda)=\{0,1\}^{\pkeml(\lambda)}$, is a triple of algorithms \\ $\pk=(\pkg,\pke,\pkd)$,  where $\pkg(1^\lambda,P_{\X \Y \Z })$ is a {\em randomized generation algorithm} that produces private samples $(r_A,r_B,r_E)$ that are privately given to the corresponding parties,   $\pke(r_A)$ is the {\em randomized encapsulation algorithm} that outputs a pair of ciphertext and key $(c,k)$ for $c\in \mc C$ and $k\in\{0,1\}^{\pkeml(\lambda)}$, and $\pkd(r_B,c)$ is the deterministic decapsulation algorithm that outputs a  key $k$ or a symbol $\perp$ (for an invalid ciphertext).
\end{definition}
%


\noindent
{\bf Correctness.} 
A pKEM is $\eps(\lambda)$-correct if  for all $\lambda\in\mathbb{N}$ and 
 $(r_A,r_B, r_E)\leftarrow \pkg(1^\lambda,P_{\X \Y \Z })$, \\
$\text{Pr}[\pkd(r_B, c)\neq \pke(r_A).{key}]\leq \eps(\lambda)$, where $\eps:\mathbb{N}\to [0,1)$
is a \textit{small} function of $\lambda$,  and $\pke(r_A).{key}=k$ and 
 the probability is over  all  random coins of $\pke(\cdot)$ and $\pkg(\cdot)$.


\vspace{1mm}
\noindent
{\bf Security of pKEM.}  
We consider three types of attacks: One-time attack (OT), Chosen Encapsulation  Attack  (CEA),  
and  Chosen  Ciphertext  Attack (CCA),  specified by access to the 
 encapsulation  and  decapsulation oracles $\pke(r_A,\cdot)$ and $\pkd(r_B,\cdot)$,  respectively. 
The  corresponding  security  notions  are  denoted  by IND-OT,  IND-CEA 
and  IND-CCA,  respectively. 
 An encapsulation  query to $\pke(r_A,\cdot)$ is a call to generate a key and ciphertext pair $(c,k)$ and 
 {does not take any input from the adversary}.
For a query to $\pkd(r_B,\cdot)$,  the attacker chooses  a ciphertext $c$, 
 and receives the corresponding   key $k$, or $\perp$.

 \begin{definition} [pKEM distinguishing advantage] \label{def:gkem}Let $\pk = ( \pkg, \pke, \pkd)$  be a pKEM 
and let $\msa = (\msa_1,\msa_2)$ be a distinguisher.
 The \underline{p}reprocessing  \underline{k}ey \underline{ind}istinguishability advantage (pkind) 
  is denoted by $ Adv^{pkind\text{-}atk}_{\pk,\msa}(\lambda)   $ and defined as follows:
\begin{eqnarray}\label{eq:gkemadv} 
   |\pr[\pind_{\pk,\msa}^{atk\text{-}0}(\lambda)= 1]-\pr[\pind_{\pk,\msa}^{atk\text{-}1}(\lambda)=1]|,
\end{eqnarray}
%
where the distinguishing game $\pind_{\pk,\msa}^{atk\text{-}b}$ for   a random bit $b\stackrel{\$}\gets\{0,1\}$, is defined in Figure.~\ref{fig1:thm1}.
 \end{definition}

\begin{figure}[!ht]
\small
\parbox{0.65\textwidth}
{\underline{\textbf{Game} $\pind_{\pk,\msa}^{atk\text{-}b}(\lambda)$}  \ \ \ \ \ \ \ \ \ \ \ \ \ \ \ \ \ \ \ \ \ \ \ \ \ \   
 \underline{\textbf{Oracles}$_{ }$ $\ms O_1$ and $\ms O_{2_{}}$}
\begin{multicols}{2}
\begin{algorithmic}[1]
  \State $(r_A ,r_B, r_E)\stackrel{\$}\gets \pkg(1^\lambda,
{P_{\X \Y \Z }})$ 
  \State {$st_1\stackrel{\$}\gets \msa_1^{\ms O_1}(r_E)$}
\State {$(k^*,c^*)\stackrel{\$}\gets \pke(r_A)$}
\State {$k_0\gets k^*$
\State $k_1\stackrel{\$}\gets \{0,1\}^{\pkeml(\lambda)}$}
\State $b'\stackrel{\$}\gets \msa_2^{\ms O_2}(st_1,c^*,k_b)$
\State Return $b'$
\end{algorithmic}
     \columnbreak
 {
\begin{tabular}{l l l}
\centering
    $atk$ &$\ms O_1(\cdot)$ & $\ms O_2(\cdot)$ 
    \vspace{.3em}\\
    \hline
     $ot$&$\varepsilon$ & $\varepsilon$\\
     $cea$&$\ms {\pke}(r_A,\cdot)$ & $\ms {\pke}(r_A,\cdot)$\\
     $cca$&$\{\ms {\pke}(r_A,\cdot),$ $\ms {\pkd}(r_B,\cdot)\}$&$\{\ms {\pke}(r_A,\cdot),$ $\ms {\pkd}(r_B,\cdot)\}$\\
\end{tabular}
}
\end{multicols}
}
    \caption{
    The security game $\pind_{\pk,\msa}^{atk\text{-}b}$   where
     $b\stackrel{\$}\gets\in\{0,1\}$ and $atk\in\{ot,cea,cca\}$. 
    Here $\ms O_1(\cdot)$ and $\ms O_2(\cdot)$  are oracles that are accessed before and after the challenge is seen, respectively.  $\ms O_i = \varepsilon$,
for $i\in\{1, 2\}$, means $\ms O_i$
returns the empty string $\varepsilon$.
The number of queries for computational (resp. unbounded) adversaries will be a polynomial in $\lambda$ (resp. constant number $q_e$ encapsulation and $q_d$ decapsulation queries).  
    The adversary $\msa_2$  cannot 
    ask 
    $c^*$ to decryption oracle.}
    \label{fig1:thm1}
\end{figure}

For $\text{ATK}\in\{\text{OT,CEA,CCA}\}$, a pKEM is $\sigma(\lambda)$-IND-ATK secure   if $Adv^{pkind\text{-}atk}_{\pk,\msa}(\lambda)$ is bounded by $\sigma(\lambda)$ for $\text{atk}\in\{\text{ot,cea,cca}\}$, respectively,  where $\sigma:\mathbb{N}\to [0,1)$
is a \textit{small} function of $\lambda$.
The adversary 
 $\msa$  may be   computationally $(i)$ bounded, or  $(ii)$   unbounded. 
  We call the KEM in the former case a \textit{computational KEM (cKEM)}, and in the latter case 
 an \textit{information theoretic KEM (iKEM)}, both in preprocessing model.  For a secure cKEM, $\sigma(\cdot) \in NEGL$ and for a secure  iKEM $\sigma(\cdot) \in SMALL $. 

\begin{remark}[iKEM with  bounded-query 
security]\label{remark:bounded}The number  of queries when the adversary is computationally bounded (Definition \ref{def:gkem})  is a 
 polynomial 
  in  $\lambda$. 
We  define  \textit{$q$-bounded} adversaries 
for  iKEM, where the number of queries is bounded by a {\em known predetermined} polynomial in $\lambda$. 
\textit{$q$-bounded}  CCA  security  for public-key encryption   has been considered in \cite{cramer2007bounded}  to overcome  impossibility results that hold for general CCA encryption.
In iKEM, the bound on the number of queries is because of the adversary's unlimited computation power.
Indistinguishability security against a  $q_e$-bounded CEA adversary with access to at most $q_e$ encapsulation queries is denoted by IND-$q_e$-CEA  security. Similarly,    
IND-$(q_e;q_d)$-CCA  security  is defined against  an attacker that is $q_e$-bounded  for encapsulation queries and $q_d$-bounded  for decapsulaton queries, where  the  queries can be asked according to the distinguishing game of Figure.~\ref{fig1:thm1}.

\end{remark}

\subsubsection{Ciphertext Integrity (INT-CTXT) in preprocessing model.}\label{sec:intctxt}
\textit{Ciphertext integrity} (INT-CTXT)  requires that the adversary's tampering of the ciphertext be detected by a high probability.
 Ciphertext integrity was defined by Bellare et al.  \cite{bellare2000authenticated} for symmetric  key encryption systems and it was  proved that in symmetric key encryption systems,
IND-CPA security together with INT-CTXT security implies IND-CCA security 
(\cite[Theorem 3.2]{bellare2000authenticated}). 
In \cite{katz2006characterization}, the notion of ciphertext {\em existential unforgeability} is proposed  and  a composition theorem (\cite[Theorem 1]{katz2006characterization}) is proved that shows existential unforgeability of the ciphertext together with CPA security of the encryption system, leads to CCA security of the encryption system.
In the following we define integrity for KEM in preprocessing model, and prove a
composition theorem to obtain CCA security.

\begin{definition}[pKEM ciphertext integrity]
A   pKEM $\pk=(\pkg, \pke, \pkd)$ with security parameter $\lambda$,  initial joint distribution 
$P_{\X \Y \Z }$, and  the key space $\ms{KeySP}(\lambda)=\{0,1\}^{\pkeml(\lambda)}$
 provides 
ciphertext integrity  (INT-CTXT), if for 
all  initial  correlated samples $\big(r_A,r_B,r_E\big)$ (generated by $\pkg(1^{\lambda},P_{\X \Y \Z })$), and all adversaries $\ms A$ with access to the encapsulation  and  decapsulation queries, the \underline{k}ey \underline{int}egrity advantage defined as $Adv^{kint}_{\pk, \ms 
A}(\lambda)\triangleq\pr[
\mathrm{KINT}_{\pk,\ms A}=1]$ 
is  upper bounded by $\delta(\lambda)$,  a \textit{small} function of $\lambda$, where the integrity game $\mathrm{KINT}_{\pk,\ms A}$ is given in Figure.~\ref{fig:int}.

\label{def: integrity} 
\end{definition}
\begin{figure}[h!]
\begin{center}
  \begin{minipage}[c]{0.35\textwidth}
  \parbox{0.99\textwidth}
{\underline{\textbf{Game} $\mathrm {KINT}_{\pk,\ms A}(\lambda)$} 
\begin{algorithmic}[1]
  \State $(r_A ,r_B, r_E)\stackrel{\$}\gets \pkg(1^{\lambda},
{P_{\X \Y \Z }})$ 
  \State {$\hat{c}\stackrel{\$}\gets \ms A^{\pke(r_A,\cdot),\pkd(r_B,\cdot)}(r_E)$}

\State If $\pkd(r_B,\hat{c})\neq \perp$:  {Return 1}
\end{algorithmic}
}
  \end{minipage}
         \caption{
         The integrity game of    pKEM. 
      Computationally bounded adversaries can make any-poly encapsulation and decapsulation queries. Unbounded  adversaries
        can  make fixed-poly $q_e$ encapsulation and $q_d$ decapsulation queries. 
        $\hat{c}$ cannot be a queries output of 
        $\pke(r_A,\cdot)$.}
        \label{fig:int}
\end{center}        
\end{figure}

We define INT-$(q_e;q_d)$-CTXT  for    an adversary  with $q_e\geq 0$ encapsulation and $q_d>0$ decapsulation queries, where the number of allowed queries depends on the adversary being computationally bounded or unbounded.


The following theorem shows that  a pKEM  that is IND-CEA and INT-CTXT secure is  IND-CCA secure.

\begin{theorem}
\label{theo:comp}
Let $\pk=(\pkg, \pke, \pkd)$ be a pKEM  with security parameter $\lambda$ and the input distribution $P_{\X \Y \Z }$. For an adversary $\msa=(\msa_1,\msa_2)$ in the CCA key distinguishing game $\pind_{\pk,\msa}^{cca\text{-}b}(\lambda)$, there are adversaries $\ms A$ and $\ms B$ for $\mathrm {KINT}_{\pk,\ms A}(\lambda)$ and the CEA key distinguishing game $\pind_{\pk,\ms B}^{cea\text{-}b}(\lambda)$, respectively 
that satisfy the following:  

$$Adv^{pkind\text{-}cca}_{\pk,\msa}(\lambda)\leq 2q_dAdv^{kint}_{\pk, \ms 
A}(\lambda)+Adv^{pkind\text{-}cea}_{\pk,\ms B}(\lambda).
$$

If $\msa$ makes $q_e$ encapsulation and $q_d$ decapsulation queries, $\ms A$ makes $q_e-1$ encapsulation and $q_d$ decapsulation queries to its decapsulation oracles, and $\ms B$ makes $q_e$ queries to its encapsulation oracle, 
%
we have the following. 
\begin{enumerate}
\item For computationally bounded adversaries 
\[\mbox{INT-CTXT cKEM }+ \mbox{IND-CEA
cKEM} \rightarrow \mbox{IND-CCA cKEM}.
\]

\item For computationally unbounded adversaries, 
\begin{eqnarray*}
 &&\mbox{INT-$(q'_e;q_d)$-CTXT iKEM} +  \mbox{IND-$q_e$-CEA iKEM} \rightarrow  \mbox{IND-$(q''_e; q_d)$-CCA iKEM},
\end{eqnarray*}
 where $q''_e = min(q_e; q'_e-1)$.
\end{enumerate} 
\end{theorem}

\textit{Proof sketch.}
For the proof of the first part, we use  a sequence of  two  games  $\mathrm{G}^{0\text{-}b}_{\pk,\msa}$ and $\mathrm{G}^{1\text{-}b}_{\pk,\msa}$
 played by the distinguisher $\ms D$, where 
 $b$  is    uniformly chosen from $\{0,1\}$. 
 The first game  $\mathrm{G}^{0\text{-}b}_{\pk,\msa}$ is the  CCA distinguishing game ($\pind_{\pk,\msa}^{cca\text{-}b}(\lambda)$ in Figure~\ref{fig1:thm1}). 
The second game, $\mathrm{G}^{1\text{-}b}_{\pk,\msa}$, is the same as $\mathrm{G}^{0\text{-}b}_{\pk,\msa}$ except for its decapsulation oracle that always outputs $\perp$, an empty string. We bound the CCA advantage of the adversary by bounding the advantage of these games.
The proof of the second part uses the same sequence of games, but is against a computationally unbounded adversary. See the complete {proof} in Appendix~\ref{pf:theoremcca}. $\qed$

\subsection{Hybrid encryption  in Preprocessing Model}\label{sec:heenc}


 We define hybrid encryption (HE) and KEM/DEM paradigm for preprocessing model, 
 where during the offline phase, Alice, Bob and possibly Eve, receive correlated private inputs, and
during the online phase,  pKEM generates the key that will be used by DEM. 

\begin{definition}[Hybrid encryption in preprocessing model]
For 
a security parameter $\lambda$,  let \\$\pk=(\pkg,\pke;\pkd)$ be a    pKEM  and $\dem=(\edem,\ddem)$ be a  DEM  with 
the same key space  $\{0,1\}^{\ell(\lambda)}$, for each $\lambda$.
A {\em hybrid encryption in preprocessing model} denoted by $\ike_{\pk,\dem}= (\gike,\eike,\dike)$  is defined as given in Figure~\ref{fig:hyb}.


\end{definition}

\begin{figure}[!ht]
   \begin{center}
\begin{tabular}{l}
  \underline{$\mathbf{Alg}\ \gike(1^\lambda,P_{\X \Y \Z })$} \\
       {$(r_A,r_B,r_E)\stackrel{\$}\gets\pkg(1^\lambda,P_{\X \Y \Z })$}\\
        {Return } $(r_A,r_B,r_E)$
       \vspace{1em}
\end{tabular}
\begin{tabular}{l l}
\\
\underline{$\mathbf{Alg}\ \eike(r_A,m)$} & \underline{$\mathbf{Alg}\ \dike(r_B,c_1,c_2)$} \\
 {$(c_1,k)\stackrel{\$}\gets\pke(r_A)$} &  If $\perp \gets\pkd(r_B,c_1)$:\\
 { $c_2\gets\edem(k,m)$} &  \qquad \quad Return $\perp$\\  
 { {Return }$(c_1,c_2)$} &  Else: $m\gets\ddem(c_2,k)$\\
 & \qquad \quad {Return } $m$
\end{tabular}
    \caption{
    Hybrid encryption $\ike_{\pk,\dem}$ in preprocessing model}
    \label{fig:hyb}
\end{center}    
\end{figure}

\noindent
\textbf{Security of hybrid encryption in preprocessing model.}\label{defn:hesecurity}
$\ike_{\pk,\dem}$ is a private input 
encryption system, where Alice and Bob's private inputs are not the same but are correlated.
We use indistinguishability security and  
consider three security notions 
depending on the attacker's access to the encryption and decryption oracles (during the online phase):  i) no oracle access (IND-OT), ii) access to encryption queries (IND-CPA),  and iii) access to encryption and decryption queries,  where access in the latter two cases will be  before and after receiving the challenge ciphertext (IND-CCA). 
The number of queries for computationally bounded adversaries is polynomial in $\lambda$, and for unbounded adversary is a  predetermined polynomial in $\lambda$.
The security notions in the latter case for $q_e$ encryption queries, and for $q_e$ encryption and $q_d$ decryption queries  are  denoted by IND-$q_e$-CEA  and IND-$(q_e;q_d)$-CCA, respectively. 
 The security games are  similar to the 
 security games in symmetric key encryption schemes: the adversary (after making 
  queries according to the game type) generates two  equal length (in bits) messages $m_0$ and $m_1$, and for a random $b\in\{0,1\}$, receives $c^*=\ike_{\pk,\dem}(m_b)$. It  then (after making enough queries) outputs a bit  $\hat{b}\in\{0,1\}$.
 The \underline{ind}istinguishability advantage 
for a computationally bounded adversary $\msa$ and $atk\in\{ot,cpa,cca\}$, and
computationally unbounded adversary $\msa'$ and $atk\in\{ot,q_e\text{-}cpa,(q_e;q_d)\text{-}cca\}$, are 
$
    Adv^{ind\text{-}atk}_{\ike_{\pk,\dem},x}(\lambda) \triangleq |\pr[\hat{b}=1|b=0]-\pr[\hat{b}=1|b=1]|
$, where $x\in\{\msa,\msa'\}$, $\pk=\ck$ when $x=\msa$ and $\pk=\ikem$ when $x=\msa'$. 
The advantage is bounded by $\sigma(\lambda)$, where   $\sigma(\cdot) \in NEGL$ for adversary $\msa$ and  $\sigma(\cdot) \in SMALL$ for adversary $\msa'$.

The following theorem is the counterpart for Theorem $7.2$ in \cite{Cramer2003DesignAA} in preprocessing  model. Theorem 7.2 in \cite{Cramer2003DesignAA} considers only public key KEM. We prove the following theorem for both computational and information-theoretic KEMs
(cKEM and iKEM, respectively).
 The theorem is proved for  two types of query accesses for the adversary. 
One can consider similar types of results for other KEMs as defined in \cite{HERRANZ20101243}.
\begin{theorem}[Hybrid encryption composition theorem]
\label{theo:composition}
For a security
parameter $\lambda \in \mathbb{N}$, let, 

--$\ck=(\ckg, \cke, \ckd)$ be an $\epsilon(\lambda)$-correct cKEM in preprocessing model, and 

--$\ikem=(\ikemg, \ikeme, \ikemd)$ be an $\epsilon(\lambda)$-correct iKEM in preprocessing model,  

and let $\dem$ denote a one-time symmetric key encryption scheme with security parameter $\lambda$ that is compatible with the corresponding $\ck$ or $\ikem$. Then, 
{\small
\begin{eqnarray*}
&& 1) \quad \sigma(\lambda)\mbox{-IND-CEA}~\ck + \sigma'(\lambda)\mbox{-IND-OT} ~\dem\rightarrow
 (2\sigma(\lambda)+\sigma'(\lambda))\mbox{-IND-CPA}~ \ike_{\ck,\dem}\\
&& 2) \quad \sigma(\lambda)\mbox{-IND-CCA}~ \ck + \sigma'(\lambda)\mbox{-IND-OT} ~\dem  \rightarrow
 (2\epsilon(\lambda) + 2\sigma(\lambda)+\sigma'(\lambda))\mbox{-IND-CCA}~ \ike_{\ck,\dem}\\
&& 3) \quad \sigma(\lambda)\mbox{-IND-}q_e\mbox{-CEA }\ikem+\sigma'(\lambda)\mbox{-IND-OT } \dem \rightarrow
 (2\sigma(\lambda)+\sigma'(\lambda))\mbox{-IND-}q_e\mbox{-CPA } \ike_{\ikem,\dem}\\
&& 4) \quad \sigma(\lambda)\mbox{-IND-}(q_e; q_d)\mbox{-CCA } \ikem+\sigma'(\lambda)\mbox{-IND-OTCCA } \dem  \rightarrow
 (2\eps(\lambda)+2\sigma(\lambda)+\sigma'(\lambda))\mbox{-IND-}(q_e; q_d)\mbox{-CCA } \ike_{\ikem,\dem}
\end{eqnarray*}
}

Security of the hybrid encryption scheme in all above cases is with respect to a computationally bounded adversary.

\end{theorem}

\textit{Proof Sketch.}
We prove the theorem for the first two cases; the proofs of the last two cases will be similar.
 We use  a sequence of  three  games  $\mathrm{G}^{0\text{-}b}_{\msa}$, $\mathrm{G}^{1\text{-}b}_{\msa}$, and $\mathrm{G}^{2\text{-}b}_{\msa}$, all
 played by a computationally bounded adversary (distinguisher) $\msa$.
 $\mathrm{G}^{0\text{-}b}_{\msa}$ is identical to the distinguishing game of the hybrid encryption in preprocessing model. 
 $\mathrm{G}^{1\text{-}b}_{\msa}$  differs 
  from 
 $\mathrm{G}^{0\text{-}b}_{\msa}$ in its decapsulation oracle response. 
For the challenge HE ciphertext $c^*=(c_1^*,c_2^*)$, where $c_1^*$ is the ciphertext output of 
 $\cke$ and $c_2^*$ is generated by $\edem$,  the response will be as follows.
If the decryption query $c=(c_1,c_2)\neq (c_1^*,c_2^*)$ and $c_1=c_1^*$ (and $c_2\neq c_2^*$), the response will be the key $k_1^*$ 
  that was generated by the encapsulation oracle (corresponding to $c_1^*$); otherwise 
  the decryption oracle of $\mathrm{G}^{1\text{-}b}_{\msa}$ uses $\ckd$ to decrypt $c$.
Finally, $\mathrm{G}^{2\text{-}b}_{\msa}$  differs from $\mathrm{G}^{1\text{-}b}_{\msa}$ in using a uniformly sampled key instead of the key generated by the cKEM. The differences between the success probabilities of the first two, and the last two games are bounded by the failure probability of encapsulation, and  the indistinguishability advantage of the  of the $\dem$.
 The complete proof is given in Appendix~\ref{pf:hecomp}.$\ \ \qed$

Note that in cases (3) and (4) of the theorem above, the iKEM is secure against a computationally unbounded adversary. Therefore, as long as the symmetric encryption scheme is quantum safe, {\em the resulting hybrid encryption scheme will be quantum-resistant}.

\section{Instantiating iKEM 
}
\label{sec:ikem-inst}
In this section, we construct two iKEMs in correlated random model and 
prove their security properties.
In the first construction (Construction~\ref{ikem:cea}) the adversary can only query the encapsulation oracle. In  the second construction (Construction~\ref{ikem:cca}) however, the adversary can   query both  the encapsulation and decapsulation oracles. 


  For both constructions we consider the 
probabilistic experiment that underlies the generation of correlated triplet 
to be  $n$ times independent repetition of a probabilistic experiment, and so $\X = (X_1,\cdots,X_n)$, $\Y = (Y_1,\cdots,Y_n)$ and $\Z =(Z_1,\cdots,Z_n)$ respectively, where 
$P_{\X \Y \Z }(\x , \y , \z )=
\prod_{i=1}^n P_{XYZ}(x_i,y_i,z_i)$, where $\x= (x_1,\cdots,x_n)$,  $\y=(y_1,\cdots,y_n)$ and $ \z= (z_1,\cdots,z_n)$.
Alice, Bob and Eve privately receive realizations of the random variables $\X$, $\Y$ and $\Z$ , respectively.

 This setting is considered  in  
commonly used  source model 
~\cite{Maurer1993,Maurer1997authencation}. 
%
\subsection{A CEA secure construction}
\label{constructionikem1}

 An IND-$q_e$-CEA secure iKEM allows adversary to query tha encapsulation oracle, and 
can be used to construct an IND-$q_e$-CPA secure hybrid encryption where the adversary has access to encryption queries.
The construction slightly modifies 
the 
IND-$q_e$-CEA secure iKEM construction in \cite{sharifian2021information} to increase 
the length of the  extracted key, without compromising its security.

\begin{construction}[CEA secure iKEM.]\label{ikem:cea}
We define an iKEM $\mathsf{ikem}_{cea}=(\ikemg,\ikeme,\ikemd)$, as follows:

Let  $P_{\X \Y \Z }=\prod_{i=1}^nP_{X_i Y_i Z_i}$ be the public joint distribution  as defined  above,  
and $P_{X_i Y_i Z_i}=P_{X Y Z}$ for all $ i \in \{1, \cdots, n\}$. 

Let $ h: \mathcal{X}^n \times \mathcal{S} \rightarrow \{0,1\}^t$ and $ h': \mathcal{X}^n \times \mathcal{S'} \rightarrow \{0,1\}^\ell$ be two  universal  hash families. Let the  ciphertext and key space be defined as $\mathcal{C}=\{0,1\}^t \times S'$ and $\mathcal{K}=\{0,1\}^\ell$, respectively. 
The iKEM $\mathsf{ikem}_{cea}$'s three algorithms $(\ikemg,\ikeme,\ikemd)$ are described  in Algorithm~\ref{alg:iK_{OWSKA}Gen}, Algorithm~\ref{alg:iK_{OWSKA}.Enc} and Algorithm~\ref{alg:iK_{OWSKA}.dec} respectively. 
The parameters $t$ and $\ell$ depend on the security parameter $\lambda$ 
and their relationship  with other system parameters is given in section~\ref{securityanalysisconst1}.

 In    $\ikemd$  (Algorithm~\ref{alg:iK_{OWSKA}.dec}), we use a parameter $\nu$ 
to define a set $\mathcal{R}$. This 
is a 
decapsulation algorithm parameter that depends on  $P_{\X\Y\Z}$,  the correlation between the RVs $\X $ and $\Y $, and is chosen using 
the required correctness (and security) of the derived key.  Higher correlation between the RVs $\X $ and $\Y $ leads to smaller $\nu$ for the same correctness level.
 The details of parameter derivation for $\nu$ and $t$ are  in \cite{sharif2020}    and are also repeated in the proof of Theorem \ref{Thm:ikemotsecurity} which uses the same reconciliation algorithm to obtain $\x$ from $\y$. 
Theorem~\ref{thm:ceaconst1} derives that the length of  the extracted key gives is upper bounded or 
$\ell$, the extracted key length by constructing a protocol, improving the results in \cite{sharif2020}. 


\end{construction}
 
Note that $\ikemg(P_{\X \Y \Z })$,   in addition to the random samples,  generates 
a random seed  $s$ of appropriate size 
that is distributed to the parties over public authenticated channels. 


{\setlength{\algoheightrule}{0pt}
\setlength{\algotitleheightrule}{0pt}
\SetAlgoNoLine%
 	\begin{algorithm}[!ht]
             \SetAlgoNoEnd
             \SetAlgoNoLine%
 		\DontPrintSemicolon
            \rule{.4\textwidth}{0.75pt}
            
 		\SetKwInOut{Input}{Input}\SetKwInOut{Output}{Output}
 		\Input{A public  distribution $P_{\X \Y \Z }$}
 		\Output{$(\x ,\y ,\z )$, the seed  $s$ (public)}
 		\SetKw{KwBy}{by}
 		\SetKwBlock{Beginn}{beginn}{ende}
 		1: 
  Samples  $(\x ,\y ,\z ) \xleftarrow{\$} P_{\X \Y \Z }$ and \\ send privately 
   to Alice, Bob and Eve, respectively. \\
 		2: 
   Sample and publish 
  $s \xleftarrow{\$} \mathcal{S}$ for $h(\cdot)$. \\
 		\parbox{\linewidth}{\caption{
   $\ikemg(P_{\X \Y \Z })$}}  
 		\label{alg:iK_{OWSKA}Gen}
 	\end{algorithm} }
\begin{minipage}{0.43\textwidth}
\vspace{-4.6em}
\setlength{\algoheightrule}{0pt}
\setlength{\algotitleheightrule}{0pt}  
 	\begin{algorithm}[H]
 		\SetAlgoLined
 		\DontPrintSemicolon
            \rule{.9\textwidth}{0.75pt}
               
 		\SetKwInOut{Input}{Input}\SetKwInOut{Output}{Output}
 		\Input{$\x $ and the seed $s$ (output of $\ikemg$)}
 		\Output{   The final 
   key  = $k$, 
  ciphertext = $c$}
 		\SetKw{KwBy}{by}
 		\SetKwBlock{Beginn}{beginn}{ende}
 		        1: 
           Sample $s' \xleftarrow{\$} \mathcal{S'}$ for $h'(\cdot)$\\
 		        2: $k$ = $h'(\x ,s')$ \\
 		        3: $c$ = $(h(\x ,s),s')$ \\
 		        4: Output = $(k,c)$ \\
 		\vspace{-1em}
 		\parbox{\linewidth}{\caption{
   $\ikeme(\x )$}} 
 		\label{alg:iK_{OWSKA}.Enc}
 	\end{algorithm} 
\end{minipage}
\begin{minipage}[c]{0.54\textwidth}
\setlength{\algoheightrule}{0pt}
\setlength{\algotitleheightrule}{0pt}  
 	\begin{algorithm}[H]
 		\SetAlgoNoLine
 		\DontPrintSemicolon
            \rule{.9\textwidth}{0.75pt}
                  
 		\SetKwInOut{Input}{Input}\SetKwInOut{Output}{Output}
 		\Input{$\y $, ciphertext $c$ and the seed $s$ (output of $\ikemg$)}
 		\Output{ The final 
   key $k$  or $\perp$}
 		\SetKw{KwBy}{by}
 		\SetKwBlock{Beginn}{beginn}{ende}
 		        1: Parse $c$ as $(v,s')$, where $v$ is a $t$-bit string \\
 		        2: Let $\mathcal{R}=\{\x :-\log(P_{\X |\Y }(\x |\y )) \le \nu\}$ \\
 		        3: For each $\hat{\x } \in \mathcal{R}$, Bob checks whether $v=h(\hat{\x },s)$ \\
 		         4: \eIf{there is a unique $\hat{\x } \in \mathcal{R}$ such that   $v=h(\hat{\x },s)$}{
 		                   Output $k=h'(\hat{\x },s')$ \\
 		                   }
 		                   {
 		                   Output $\perp$ \\
 		                   }
 		\vspace{-1em}
 		\parbox{\linewidth}{\caption{
   $\ikemd(\y ,c)$}} 
 		\label{alg:iK_{OWSKA}.dec}
 	\end{algorithm} 
\end{minipage}


\subsection{Security analysis of iKEM construction~\ref{ikem:cea}}\label{securityanalysisconst1}

Theorem~\ref{thm:ceaconst1} provides the relationship among parameters of construction~\ref{ikem:cea}. 

  The protocol is based on the OWSKA in \cite{sharif2020}. The parameters  $\nu$ and  $t$  are derived in ~\cite[Theorem 2]{sharif2020} such that  the error probability of 
the protocol will be upperbounded  by the correctness (reliability) parameter $\epsilon$. 
%
The iKEM construction~\ref{ikem:cea} uses the same reconciliation information in all queries. That is, in the encapsulation ciphertext $c=(h(\x,s), s')$,  the  value of $h(\x,s)$ that is used by Bob to recover recover $\x$ (reconciliation information) will be the same in all queries.
Each query however will include a new value of $s'$ and so a new final key. 
In our  construction~\ref{ikem:cea}, the randomness $s$ is generated during the initialization and published (or sent to Bob over a public authenticated channel).
The
 CEA secure iKEM protocol construction in~\cite{sharifian2021information} however updates both parts of  $c$ in each query, which results in higher information leakage from $\x$
and shorter length for the final key.
In appendix~\ref{app:setalgotheo} we have reproduced the protocol in ~\cite{sharifian2021information} for ease of reference.



 The following lemma 
for conditional min-entropy is proved in ~\cite[Lemma 2]{ska2023}, and will be used in Theorem \ref{thm:ceaconst1}. 
%

\begin{lemma} {\em \cite{ska2023}}  \label{lemma:minentropy}  For any $X_1Z_1, \cdots, X_nZ_n$ independently and identically distributed according to $P_{XZ}$, it holds that \\$\tilde{H}_\infty({\X }|{\Z })=n\tilde{H}_\infty(X|Z)$, where ${\X }=(X_1, \cdots, X_n)$ and ${\Z }=(Z_1, \cdots, Z_n).$ 
\end{lemma}

\begin{theorem}[IND-$q_e$-CEA]
\label{thm:ceaconst1}
The iKEM $\mathsf{ikem}_{cea}$ described in construction~\ref{ikem:cea} establishes a secret key of length \[
\ell \le \frac{n\tilde{H}_{\infty}(X |Z ) + 2\log(\sigma) + 2 - t}{q_e+1}\]
that is $2\sigma$-indistinguishable from random by an adversary with access to $q_e$ encapsulation queries, where $q_e \ge 0$ (i.e. $2\sigma$-IND-$q_e$-CEA secure). 
\end{theorem}

\textit{Proof sketch.}
 The response to an encapsulation query leaks some information about Alice's private sample $\x$ and this reduces the length of the final  shared key.
The response to the $i$th encapsulation query is a key and ciphertext pair $(K_i,C_i)$, where $C_i =(h(\x,s), s'_i)$  and $h(\x,s)$ 
is the same in all responses. Here $K_i$ and $C_i$ are RVs over $\{0,1\}^\ell$ and $\{0,1\}^t$,  respectively.

After seeing $(K_i,C_i)$,  the remaining min-entropy entropy of $\x$ 
using~\cite[Lemma 2.2(b)]{DodisORS08},  will be lowerbounded by
$\tilde{H}_\infty(\X |\Z ,\mathbf{W}^{cea}_i)=\tilde{H}_\infty(\X |\Z ,K_i,C_i) \ge \tilde{H}_\infty(\X |\Z ) - \ell - t$.
Since $h(\x,s)$ is the same in all responses, after $q_e$ queries, the min-entropy entropy of $\x$ will be  $\tilde{H}_\infty(\X |\Z ) - q_e \ell - t$ which will be used 
to bound the key distinguishing advantage of the adversary.  The complete proof of the theorem is in Appendix~\ref{pf:thm:ceaconst1}. 
$\qed$

 {\em Comparison.} 
The construction in \cite[Theorem 2]{sharifian2021information} 
provides a  key of length 
\[
\ell \le \frac{\tilde{H}_{\infty}(\X |\Z ) + 2\log(\sigma) + 2}{q_e+1} -t - \log(\frac{q_e}{\sigma})
\] 
that is $2\sigma$-indistinguishable from random ($2\sigma$-IND-$q_e$-CEA). Our Construction  \ref{ikem:cea} results 
in a key of length 
\[
\ell \le \frac{n\tilde{H}_{\infty}(X |Z ) + 2\log(\sigma) + 2 -t}{q_e+1}
\]
that is $2\sigma$-indistinguishable from random ($2\sigma$-IND-$q_e$-CEA), improving the  
result in \cite[Theorem 2]{sharifian2021information}.

\subsection{A CCA secure construction}
\label{constructionikem}
%
 In this section, we modify the 
IND-$q_e$-CEA secure iKEM construction~\ref{ikem:cea} in Section~\ref{constructionikem1} to provide stronger security. More specifically, we extend the construction to 
an IND-$(q_e;q_d)$-CCA secure iKEM which provides security against an adversary with access to $q_e$ encapsulation and $q_d$ decapsulation oracle queries.
Access to decapsulation queries models an adversary who can tamper with the encapsulation ciphertext, and see the result of the decapsulation algorithm on its manufactured fraudulent encapsulation ciphertext.

To provide security against such adversaries we use the notion of   ciphertext  integrity (INT-CTXT) that requires the iKEM to 
satisfy definition~\ref{def: integrity}   and guarantee that any tampering with $c$ will be detected by the decapsulation  algorithm, with a high probability.

%
 The  iKEM construction~\ref{ikem:cca}  provides ciphertext integrity.  According to Theorem~\ref{theo:comp},
the 
IND-$q_e$-CEA security and $(q'_e,q_d)$-ciphertext integrity (INT-$(q'_e,q_d)$-CTXT) of iKEM together lead to  IND- $(q''_e;q_d)$-CCA  security,  where $q''_e=\min\{q_e,q'_e-1\}$, which    is the strongest and widely 
used 
notion of security for   encryption systems. 
 The construction is based on  the construction \ref{ikem:cea}  but modifies its ciphertext, and requires the hash function $h$   to be a {\em universal hash function with additional properties.}

\begin{construction}[CCA secure iKEM.]\label{ikem:cca}
We define an iKEM, $\mathsf{ikem}_{cca}=(\ikemg,\ikeme,\ikemd)$, as follows:

 Let the joint  distribution of the three random variables $\X $, $\Y $ and $\Z $ be described by the distribution  $P_{\X \Y \Z }=\prod_{i=1}^nP_{X_i Y_i Z_i}$ that is obtained as product  of  $n$ independent  copies of the distribution $(X,Y,Z)$, where $\X =(X_1,\cdots,X_n)$, $\Y =(Y_1,\cdots,Y_n)$, $\Z =(Z_1,\cdots,Z_n)$  
 and $P_{X_i Y_i Z_i}=P_{X Y Z}$ for $ 1 \le i \le n$. The joint distribution $P_{\X \Y \Z }$ is used to generate the correlated random samples of $\X , \Y , \Z  \in {\mathcal{X}}^n$.

 Let $ h': \mathcal{X}^n \times \mathcal{S'} \rightarrow \{0,1\}^\ell $ be a  universal  hash family, and  $ h: \mathcal{X}^n \times (\mathcal{S'} \times \mathcal{S})  \rightarrow \{0,1\}^t $ denote a second    universal  hash family with extra properties as constructed  in Section \ref{ikem:robcca}. 
 
Let $\mathcal{C} = \{0,1\}^t \times \mathcal{S'} \times \mathcal{S}$ and $\mathcal{K} = \{0,1\}^\ell $ denote the ciphertext and key domains, respectively.


 The 
$\mathsf{ikem}_{cca}$ 
algorithms $(\ikemg,\ikeme,\ikemd)$ are, 
Algorithm~\ref{alg:iKGen}, Algorithm~\ref{alg:iK.Enc} and Algorithm~\ref{alg:iK.dec}, respectively.

The hash function  parameters $t$ and $\ell$ are functions of the security parameter $\lambda$ and together with the other iKEM parameters 
are derived in 
Theorem \ref{Thm:ikemotsecurity} and Theorem \ref{mac2:ctxt} in section \ref{securityfuzzyext}.
The parameter 
$\nu$ is defined and used the same as in Construction \ref{ikem:cea}.  Note that the initialization phase is only used to generate and securely distribute the private inputs of participants. The seed $s$ will be generated independently for each instance of the protocol and will be protected against the adversary's tampering through the additional properties of $h$.

\end{construction}

{\setlength{\algoheightrule}{0pt}
\setlength{\algotitleheightrule}{0pt}
 	\begin{algorithm}[!ht]
 		\SetAlgoLined
 		\DontPrintSemicolon
            \rule{.4\textwidth}{0.75pt}
            
 		\SetKwInOut{Input}{Input}\SetKwInOut{Output}{Output}
 		\Input{Distribution $P_{\X \Y \Z }$}
 		\Output{$(\x ,\y ,\z )$}
 		\SetKw{KwBy}{by}
 		\SetKwBlock{Beginn}{beginn}{ende}
 		1: 
   Samples  $(\x ,\y ,\z ) \xleftarrow{\$} P_{\X \Y \Z }$; and \\ 
   send privately to Alice, Bob and Eve, respectively. \\
 		\parbox{\linewidth}{\caption{
   $\ikemg(P_{\X \Y \Z })$}}  
 		\label{alg:iKGen}
 	\end{algorithm}

} 
\begin{minipage}{0.43\textwidth}
\vspace{-4.3em}
\setlength{\algoheightrule}{0pt}
\setlength{\algotitleheightrule}{0pt}
 	\begin{algorithm}[H]
 		\SetAlgoLined
 		\DontPrintSemicolon
            \rule{.9\textwidth}{0.75pt}
            
 		\SetKwInOut{Input}{Input}\SetKwInOut{Output}{Output}
 		\Input{$\x $}
 		\Output{extracted key = $k$, ciphertext = $c$}
 		\SetKw{KwBy}{by}
 		\SetKwBlock{Beginn}{beginn}{ende}
 		        1: Generate seed $s' \xleftarrow{\$} \mathcal{S'}$ for $h'(\cdot)$ \\
 		        2:  Generate seed $s \xleftarrow{\$} \mathcal{S}$ for $h(\cdot)$\\
 		        3: $k$ = $h'(\x ,s')$ \\
 		        4: $c$ = $(h(\x , (s',s)),s',s)$ \\
 		        5: Output = $(k,c)$ \\
 		\vspace{-1em}
 		\parbox{\linewidth}{\caption{
   $\ikeme(\x )$}} 
 		\label{alg:iK.Enc}
 	\end{algorithm} 
\end{minipage}
\begin{minipage}{0.54\textwidth}
\setlength{\algoheightrule}{0pt}
\setlength{\algotitleheightrule}{0pt}
 	\begin{algorithm}[H]
 		\SetAlgoNoLine
 		\DontPrintSemicolon
            \rule{.9\textwidth}{0.75pt}
             
 		\SetKwInOut{Input}{Input}\SetKwInOut{Output}{Output}
 		\Input{$\y $  and ciphertext $c$}
 		\Output{An extracted key $k$  or $\perp$}
 		\SetKw{KwBy}{by}
 		\SetKwBlock{Beginn}{beginn}{ende}
 		        1: Parse $c$ as $(v,s',s)$, where $v$ is a $t$-bit string \\
 		        \vspace{-1.6em}
 		        \begin{flalign}\label{reconset}
 		        \text{2: }\mathcal{R} =\{\x :-\log(P_{\X |\Y }(\x |\y )) \le \nu \} &&
 		        \end{flalign}
 		        3: For each $\hat{\x } \in \mathcal{R}$, Bob checks whether $v=h(\hat{\x }, (s',s))$ \\
 		         4: \eIf{there is a unique $\hat{\x } \in \mathcal{R}$ such that  $v=h(\hat{\x }, (s',s))$}{
 		                   Output $k=h'(\hat{\x },s')$ \\
 		                   }
 		                   {
 		                   Output $\perp$ \\
 		                   }
 		\vspace{-1em}
 		\parbox{\linewidth}{\caption{
   $\ikemd(\y ,c)$}} 
 		\label{alg:iK.dec}
 	\end{algorithm} 
\end{minipage}
\vspace{1em}


\subsubsection{Relation with CEA secure iKEM} 
\label{relation} 
 To provide CCA security 
in  Construction~\ref{ikem:cca},  we modify  Construction \ref{ikem:cea} and use 
the seeds of both hash functions  as input to   $h$, which is randomly  selected  from 
a  function family  that in addition to being a universal hash function family, can be interpreted as 
an information theoretic MAC with partially leaked secret key $\x$, that  detects tampering with the seeds $s'$ and seed $s$.
 More specifically, $h(\x , (s',s))$ is a universal hash function family with seed  $(s',s)$ that is  evaluated on the  input $\x$, and a MAC with key $\x$ that is
evaluated on the  message $(s',s)$. 
The construction of $h$ is given in Section \ref{ikem:robcca}, and proof of CCA security of iKEM is given in Section \ref{ikem:robcca}. 


\subsection{Security analysis of iKEM construction~\ref{ikem:cca}}
\label{securityfuzzyext}
  We prove security properties of the construction 
  using two main theorems.  The proofs also determine 
  parameters that must be used 
  to guarantee the required levels of correctness and security.
Theorem \ref{Thm:ikemotsecurity}
proves reliability and IND-$q_e$-CEA security of the iKEM. Theorem \ref{mac2:ctxt} proves ciphertext integrity of the construction, and together with Theorem \ref{Thm:ikemotsecurity} proves  IND-$(0, q_d)$-CCA 
security of the construction.

\begin{theorem}[reliability and IND-$q_e$-CEA]
\label{Thm:ikemotsecurity}
Let $\nu$ and $t$ 
satisfy,
{
\begin{eqnarray*}
\nu &=& nH(X |Y ) + \sqrt{n} \log (|\mathcal{X}|+3) \sqrt{\log (\frac{\sqrt{n}}{(\sqrt{n}-1)\epsilon})}, \\
 t &\ge& nH(X |Y ) + \sqrt{n} \log (|\mathcal{X}|+3) \sqrt{\log (\frac{\sqrt{n}}{(\sqrt{n}-1)\epsilon})} + \log (\frac{\sqrt{n}}{\epsilon}).
\end{eqnarray*}
    }

Then the iKEM $\mathsf{ikem}_{cca}$ 
in construction~\ref{ikem:cca} establishes a secret key of length $\ell \le  \frac{n\tilde{H}_{\infty}(X |Z ) + 2\log(\sigma) + 2}{q_e+1} - t$ that is $\epsilon$-correct and $2\sigma$-indistinguishable from random by an adversary with access to $q_e$ encapsulation queries, where $q_e \ge 0$ (i.e. $2\sigma$-IND-$q_e$-CEA  secure).
\end{theorem}

\textit{Proof sketch.}

\textit{Correctness (reliability)}. We first determine the values of $\nu$ and $t$   that guarantee 
correctness (reliability) for the given $\epsilon$, and then   prove 
security. 
 Decapsulation algorithm $\ikemd(\cdot)$ searches the set $\mathcal{R}$ that is defined by $P_{\X|\Y}$  and $\nu$, to find  a unique value $\hat{\x }$ that satisfies  $h(\hat{\x }, (s',s))=v$  were $v$ is 
the received hash value. 
 \remove{
 $\hat{\x }$ where $h(\hat{\x }, (s',s))=v$, were $v$ is 
the received hash value, 
and succeeds if a unique 
$\hat{\x }$ can be found.  
}
The algorithm fails if  at least one of the following events occurs:
\begin{align} \nonumber
 & \mathcal{E}_{1} = \{\x : \x  \notin \mathcal{R}\}=\{\x : -\log(P_{\X |\Y }(\x |\y )) > \nu\} \text{ and } \\\nonumber
 &\mathcal{E}_2 = \{\x  \in \mathcal{R}: \exists \text{ } \hat{\x } \in \mathcal{R} \text{ s.t. }  h(\x , (s',s)) = h(\hat{\x }, (s',s)\}.
 \end{align}
We use \cite[Theorem 2]{Holenstein11} and the property of universal hash function $h$ to bound these two probabilities and prove that   with appropriate choice of parameters, the sum of these two probabilities is  bounded by $\epsilon$. \\

 \textit{Security: Key indistinguishability.} The response to and 
 encapsulation query, 
 $(K_i,C_i)$, 
 leaks
information about the secret key $\x$.
%
We use  \cite[Lemma 2.2(b)]{DodisORS08} to
estimate the remaining min-entropy entropy of $\x$ as,  \[
\tilde{H}_\infty(\X |\Z ,\mathbf{W}^{cea}_i)=\tilde{H}_\infty(\X |\Z ,K_i,C_i) \ge \tilde{H}_\infty(\X |\Z ) - \ell - t,\]
 where $K_i$ and $C_i$ are RVs over $\{0,1\}^\ell$ and $\{0,1\}^t$ respectively. By bounding the total leakage of 
$q_e$ queries, we bound the key distinguishing advantage of the adversary.

The complete proof of the theorem is in Appendix~\ref{pf:Thm:ikemotsecurity}. $\qed$



\subsection{Ciphertext integrity of construction~\ref{ikem:cca}}\label{ikem:robcca}
To achieve ciphertext integrity, we use the   construction of a universal  hash function $h:\mathcal{X}^n \times (\mathcal{S'} \times \mathcal{S}) \to \{0,1\}^t$  described below.

{\bf Construction of $h$.}
For 
a  vector of  $n$ components denoted by $\x $, let $\x_{1}=[\x ]_{1\cdots t}$   and $\x_{2}=[\x ]_{t+1 \cdots n}$, where 
$\x =\x_{2} \parallel \x_{1}$ and  $t \le n/2$.  

%
\remove{

 Let $s' \in GF(2^{n})$ denote the seed 
 of the  universal  hash family 
 $ h': \mathcal{X}^n \times \mathcal{S'} \rightarrow \{0,1\}^\ell$  {\mg where $\mathcal{X}^n  = GF(2^n)$.} 

 We write $s'$ as a  vector  of elements $(s'_1,\cdots,s'_r)$ over $GF(2^{n - t})$, where $r$ is an even number. (We use padding with 1's when needed.)
}


\remove{
, and choose ???randomly $s=s_2 \parallel s_1$ where 
 $s_2 \in_R GF(2^{n-t})$ and $s_1 \in_R GF(2^{t})$. 
 }
  We define a universal hash family 
 with seed  space 
 $(\mathcal{S'} \times  \mathcal{S})$ and input space ${\mathcal{X}}^n $, where  $\mathcal{S}=GF(2^{n-t}) \times  GF(2^t)$, $\mathcal{S'}= GF(2^{w})$, for some suitable $w\in \mathbb{N}$,   
 and $\mathcal{X}^n  = GF(2^n)$. Let $s' \in \mathcal{S'}$. 
 We write $s'$  as a  vector  of elements $(s'_1,\cdots,s'_r)$ where each element is from $GF(2^{n - t})$, where $r$ is an even number  satisfying:\\ $(r-2)(n-t) < w \le r(n-t)$. (We use padding  with 1's  for $s'_r$ and $s'_{r-1}$, when needed.). Let $s=(s_2,s_1)\in \mathcal{S}$ with $s_2 \in GF(2^{n-t})$ and $s_1 \in GF(2^t)$.

 The hash function $h\big(\x ,(s', s)\big)$ with seed $(s',s)$ and input $\x \in {\mathcal{X}}^n$ is given by,
\remove{
 \begin{align} \nonumber
 &h\big(\x ,(s', s)\big)=h\big(\x ,(s', s_2, s_1)\big) \\\label{defn:uhash}
 &=\big[(\x_{2})^{r+3} + {\sum}_{i=1}^r s'_i(\x_{2})^{i+1} +s_2 \x_{2}\big]_{1\cdots t}+ (\x_{1})^{3} + s_1 \x_{1}.
 \end{align}} 
 {
 \begin{eqnarray} \label{defn:uhash}
 &&h\big(\x ,(s', s)\big)=h\big(\x ,(s', s_2, s_1)\big) \\
 &&\qquad\qquad\quad\hspace{.3em}=\big[(\x_{2})^{r+3} + {\sum}_{i=1}^r s'_i(\x_{2})^{i+1} +s_2 \x_{2}\big]_{1\cdots t}+ (\x_{1})^{3} + s_1 \x_{1}.  \nonumber 
 \end{eqnarray}
 }

\begin{lemma}\label{lemma:uhf}
$h$ is a universal  hash family. 
\end{lemma}
Proof is 
in Appendix~\ref{appn:uhfproof}.

\remove{
{\rd I THINK IN THIS DEFINITION $s$ AND $x$ ARE INPUT - $s'$ IS RANDOM, 
HERE WE ARE INTRODUING A UNIV HASH FUNCTION.

 MESSAGE HAS TWO PARTS (x,s). WE CAN SAY WE EVALUATE PROPERTIES OF THIS CONSTRUCTION WHEN $\x$ IS SECRET AND $s$ IS RANDOM.

 FOR EXAMPLE WE CAN WRITE:

 The mesage space of the unoversal hash function is $(\x, s) \in ...$

}

 \remove{Let us interpret $s'$, suitably padded with 1s, as a sequence $(s'_1,\cdots,s'_r)$ of elements of $GF(2^{n - t})$, where $r$ is an even number. We choose two random elements:  $s_2 \in_R GF(2^{n-t})$ and $s_1 \in_R GF(2^{t})$. 
 Let $s=  s_2 \parallel s_1$. 
 Define
}
 \begin{align} \nonumber
 &h\big(\x ,(s', s)\big)=h\big(\x ,(s', s_2, s_1)\big) \\\label{defn:uhash}
 &=\big[(\x_{2})^{r+3} + {\sum}_{i=1}^r s'_i(\x_{2})^{i+1} +s_2 \x_{2}\big]_{1\cdots t}+ (\x_{1})^{3} + s_1 \x_{1}.
 \end{align}
 {\rd REMOVE  -  WE SHOULD INTRODUCE THE FUNCTON AS A UHF WITH DEFINES SEED AND MESSAHE SPACE --Note that  $h$ is a universal  hash family. }

}

\vspace{1mm}
\noindent
 {\bf Proving ciphertext integrity.}
In Theorem \ref{mac2:ctxt},  we prove that the  construction~\ref{ikem:cca} is an iKEM that satisfies ciphertext integrity as given in Definition~\ref{def: integrity}, for $q_e=1$ and $q_d$. The proof of the theorem relies on 
 Lemmas \ref{lemma:fuzzy}, \ref{maxprob}, and \ref{claim:rootx}. 

 \begin{lemma}\label{lemma:fuzzy}  Consider a joint distribution $P_{\X\Y}$, and let 
  $A$ denote a random variable over a set of size 
  at most $2^{\alpha}$.
 Then, 
 \begin{align} \nonumber
     &\underset{a \leftarrow A}{\mathbb{E}}\max_{\x}\sum_{\y:
      P_{\X|\Y}(\x|\y)) \ge 2^{-\nu}}\pr[\Y=\y|A=a] \\\nonumber
     &\le  2^\alpha  \max_{\x}\sum_{\y:P_{\X|\Y}(\x|\y)) \ge 2^{-\nu}}\pr[\Y=\y ].
 \end{align}
 \end{lemma}

\begin{proof}
{
\begin{align}\nonumber
&\underset{a \leftarrow A}{\mathbb{E}} \max_{\x}\sum_{\y: P_{\X|\Y}(\x|\y)) \ge 2^{-\nu}}\pr[\Y=\y|A=a] \\\nonumber
&=\sum_a   \pr[A=a]  \max_{\x}\sum_{\y:P_{\X|\Y}(\x|\y))  \ge 2^{-\nu}}\pr[\Y=\y|A=a] \\\nonumber
&=\sum_a  \max_{\x}\sum_{\y:P_{\X|\Y}(\x|\y))  \ge 2^{-\nu}}\pr[\Y=\y|A=a]\pr[A=a] \\\nonumber
&=\sum_a\max_{\x}\sum_{\y:P_{\X|\Y}(\x|\y)) \ge 2^{-\nu}}\pr[\Y=\y, 
A=a]  \\\nonumber
&\le \sum_a\max_{\x}\sum_{\y:P_{\X|\Y}(\x|\y)) \ge 2^{-\nu}}\pr[\Y=\y ] \\\nonumber
&\le 2^{\alpha} \max_{\x}\sum_{\y:P_{\X|\Y}(\x|\y)) \ge 2^{-\nu}}\pr[\Y=\y ] 
\end{align}
}   
\end{proof}

Let \ksucc~ denote the best success probability of the adversary in guessing a key $\x_f$ for the encapsulation  algorithm \ref{alg:iK.Enc}, such  that it is considered valid by the decapsulation  algorithm \ref{alg:iK.dec}. That is, $\Pr(\x_f|\y)\geq 2^{-\nu}$ for (the unknown) decapsulation key $\y$.


\begin{lemma}\label{maxprob}
\remove{
Assume that the highest success probability of constructing 
a valid ciphertext $c_f$ be by guessing a key  $\x_f$ for $h$  that  satisfies $\Pr(\x_f|\y)\geq  2^{-\nu}$. We then have the following.
}

The  success probability of constructing 
a  ciphertext $c_f$ that is accepted by the decapsulation algorithm is bounded as follows.

1. {
\begin{align} 
&\mbox{\ksucc}  \geq 
\max \{  \max_\x \sum_{\y': \Pr(\x|\y')\geq 2^{-\nu}} \Pr(\x,\y'|\z),
\max_\y \sum_{\x': \Pr(\x'|\y)  \geq
2^{-\nu}} \Pr(\x',\y|\z)\}. \label{PS}
\end{align}}

2.   Assuming equality in the above bound,
\begin{eqnarray*}
&\mbox{\ksucc}  \leq  \max \{\max_\x\sum_{\y':\Pr(\x|\y')\geq 2^{-\nu}} P_\Y(\y'|\z), 
\max_\y\sum_{\x':\Pr(\x'|\y)\geq 2^{-\nu}} P_\X(\x'|\z)
\}
\end{eqnarray*}

\remove{
and taking into account the adversary's information $(k,c)$, the bound can be written as below.
\begin{eqnarray*}
&\max \{\max_\x\sum_{\y':\Pr(\x|\y')\geq 2^{-\nu}} P_\Y(\y'|(k,c)), \\\
&\max_\y\sum_{\x':\Pr(\x'|\y)\geq 2^{-\nu}} P_\X(\x'|(k,c))
\}
\end{eqnarray*}
}
\end{lemma} 

\begin{proof} 
1.  The encapsulation algorithm uses the key $\x$, and the decapsulation algorithm uses the key $\y$, both unknown to the adversary. To be accepted by the decapsulation algorithm, a guessed value $\x'$ must belong to the set $\cal R$ defined by the decapsulation algorithm \ref{alg:iK.dec}. 
That is $\Pr(\x'|\y)\geq 2^{-\nu}$ for the unknown $\y$.
The adversary may use two types of guessing strategies to find a candidate $\x_f$:      guess Alice's key  from 
${\cal X}^n$ such that it belongs to $\cal R$  
for the unknown $\y$, or  
guess  a Bob's key  $\y$, 
and choose one of the $\x'$  that satisfy $\cal R$ defined 
with respect to $\y$. 
The best success probabilities of these two types of guessing strategies  are denoted by
$P^{(\X)}_{S} $ and $  P^{(\Y)}_{S}$, respectively.


We have
 \begin{eqnarray} 
 P_S & \geq & \max  \{ \Pr(\mbox{Guess $\x$ from ${\cal X}^n$}, \Pr(\mbox{Guess $\y$ from ${\cal Y}^n$}) \} \nonumber\\
 & \geq & \max\{ P^{(\X)}_{S},   P^{(\Y)}_{S} \} \label{PSDef}
\end{eqnarray}
%
The encapsulation and decapsulation algorithms are deterministic and probabilities are 
over the probability space $\Pr(\x,\y,\z)$.
\remove{
and after observing the ciphertext 
$(k,c)$ that is a function of $\x$,  using $\Pr(\x,\y|(k,c))$.
}


$(i)$ To bound $P^{(\X)}_{S}=\Pr(\mbox{Guess} \x \mbox{from } {\cal X}^n)$, we note that each $\x$ will be accepted by all $\y$ that satisfy $\Pr(\x|\y)\geq 2^{-\nu}$. 

\remove{
For a value of $\y$ if the adversary's guess 
must

for each $\z$ satisfies $\Pr(\x|\y\z)\geq 2^{-\nu}$, then 
\[
\sum_{\z\in {\cal Z}^n}{\Pr(\z)\Pr(\x|\y\z)}= Pr(\x|\y)\geq 2^{-\nu},
\]
and the guessed values will be in $\cal R$
}

This means that the 
adversary's success probability that a ciphertext $c=(v, s',s)$ 
that is constructed using a guessed key 
$\X=\x$ be accepted by decapsulation algorithm that uses the 
unknown key $\y$, corresponds to 
the probability of the set of  sample points $(\X=\x,\Y=\y' )$ (key pairs) where \[\sum_{\y': \Pr(\x|\y' )\geq 2^{-\nu}} \Pr(\x,\y'|\z )\] which can be computed by the adversary (conditional distribution $\Pr(\X,\Y| \Z)$).
Therefore, to each $\x \in {\cal X}^n$ we associate a weight  $\sum_{\y': \Pr(\x|\y')\geq 2^{-\nu}} \Pr(\x,\y'|\z)$ that is the acceptance probability of the ciphertext by some $\y \in {\cal Y}^n $. 
The best guess for  $\x$ will be by finding the element of ${\cal X} ^n$  with the highest acceptance probability, 
\[P^{(\X)}_{S}  = \max_\x \sum_{\y':\Pr(\x|\y' )\geq 2{-\nu}} \Pr(\x,\y'|\z ).\]

(This also determines the value ${\x}^*$ 
(i.e. 
$\x_f$) that can be used to construct $c_f$.) (We note that the   acceptance probabilities attached to elements of ${\cal X}^n$ do not form a probability distribution on $\cal X$.)

$(ii)$ To find 
$P^{(\Y)}_S$ using ${\cal Y}^n$, we note that
{\em each $\y$ will   accept   all Alice's key values $\x'$s that satisfy 
$\Pr(\x'|\y)\geq 2^{-\nu}$}. 

This attaches an acceptance probability  to each $\y\in {\cal Y}^n$ that is the total  
 probability of  ciphertexts $c=(v, s',s)$
being accepted by a $\y$ when Alice's key is not known, and is obtained by summing probabilities of
the set  $(\x',\y)$ of sample points (key instances) as follows 
\[\sum_{\x': \Pr(\x'|\y)\geq 2^{-\nu}} \Pr(\x',\y|\z).\] 
Thus, the best guess for Bob's key 
$\y$ for accepting a ciphertext when Alice's key is unknown, is given by
\begin{align}\label{eqn:guessy}
P^{(\Y)}_{S}= \max_\y \sum_{\x': \Pr(\x'|\y)\geq 2^{-\nu}} \Pr(\x',\y|\z).
\end{align}

Therefore,  
{
\begin{align} \nonumber
&P_S   \geq
\max \{ P^{(\X)}_{S},   P^{(\Y)}_{S} \} \\
&\quad \geq \max \{  \max_\x \sum_{\y': \Pr(\x|\y')\geq 2^{-\nu}} \Pr(\x,\y'|\z),
\max_\y \sum_{\x': \Pr(\x'|\y) \geq 2^{-\nu}} \Pr(\x',\y|\z)\} \label{PSproof}
\end{align}}

{\em 2. Simplifying the bound:} 
Consider the case that the expression \ref{PS} holds with equality.  That is the 
$\x$ value 
that results in the highest  success probability for successful ciphertext forgery 
 can be obtained  by using one of the two key guessing strategies
outlined above to guess a key $\x_f$ and compute $h(\x, (s',s))$. This is true because any ciphertext that is accepted by the decapsulation algorithm must be well formed, and correspond to the evaluation of a polynomial defined by $(s',s)$ using a key that satisfies $P(\x|\y)\geq 2^{-\nu}$.  A computationally unbounded adversary can always find the roots of such a polynomial, 
and so any forged ciphertext can be generated by choosing a key $\x_f$  that satisfies the required condition, and using the encapsulation algorithm.
 This is somewhat similar to the notion of plaintext awareness in computational security~\cite{Bellare1994,BellarePalacio2004}, 
where it is assumed that
 the adversary can create ciphertexts for which it is able to ``extract'' the corresponding plaintext.

%
Thus we have, \[P_S = \max \{ P^{(\X)}_{S},   P^{(\Y)}_{S} \}. \] 
We then 
use the following approximation in terms of marginal distributions of $\Pr_\X(\x)$ and $\Pr_\Y(\y)$. 

Let $\x^*$  and  $\y^*$ be the $\x$ and $\y$ values that maximize the expressions,
$\max_\x \sum_{\y': \Pr(\x|\y')\geq 2^{-\nu}} \Pr(\x,\y'|\z)$  and 
\\$ \max_\y \sum_{\x': \Pr(\x'|\y)\geq 2^{-\nu}} \Pr(\x',\y|\z)$, respectively, and $\Pr_\X(\x)$ and $\Pr_\Y(\y)$ denote 
marginal distributions of $\X$ and $\Y$. 

Since $ \Pr_{\X,\Y}(\x,\y|\z) \leq \Pr_\X(\x|\z)$,  we have 

{
\begin{align}\nonumber
    P_S &\leq \max \{\sum_{\y':\Pr_{\X|\Y}(\x^*|\y')\geq 2^{-\nu}} \Pr_\Y(\y'|\z), 
    \sum_{\x':\Pr_{\X|\Y}(\x'|\y^*)\geq 2^{-\nu}} \Pr_\X(\x'|\z)\} 
\end{align}
}

Note that
{
\begin{eqnarray}
&&\sum_{\y':\Pr(\x^*|\y')\geq 2^{-\nu}} \Pr_\Y(\y'|\z) \leq \max_\x\sum_{\y':\Pr(\x|\y')\geq 2^{-\nu}} \Pr_\Y(\y'|\z) \label{eq1}  \\
&&\sum_{\x':\Pr(\x'|\y^*)\geq 2^{-\nu}} \Pr_\X(\x'|\z) \leq \max_\y\sum_{\x':\Pr(\x'|\y)\geq 2^{-\nu}} \Pr_\X(\x'|\z) \label{eq2prob}
\end{eqnarray}
}
This is true because the RHSs of \ref{eq1} and \ref{eq2prob} are maximizing over all $\x$ values of $\x$ and $\y$, respectively.

Therefore,
\begin{align}\label{eqn:prob}
P_S&\leq \max \{\max_\x\sum_{\y':\Pr(\x|\y')\geq 2^{-\nu}} \Pr_\Y(\y'|\z),
\max_\y\sum_{\x':\Pr(\x'|\y)\geq 2^{-\nu}} \Pr_\X(\x'|\z)
\}.
\end{align}

    
\end{proof}

 {\bf Note.}
We will use  the above calculation for conditional distributions that takes into account all the adversary's   information about $\x$,  in particular after one query, that is  $(k,c)$, 
\begin{align}\label{eqn:guessprob}
P_S&\leq \max \{\max_\x\sum_{\y':\Pr(\x|\y')\geq 2^{-\nu}} P_\Y(\y'|(k,c), \z),
\max_\y\sum_{\x':\Pr(\x'|\y)\geq 2^{-\nu}} P_\X(\x'|(k,c), \z)
\}.
\end{align}

 In the following we will use $h(\x,(s', s))$   and recall the following notations: $(i)$ $\x \in {\cal X}^n $ is written as  
 $\x = (\x_{2}\parallel  \x_{1})$ and $\x_{2}= (x_{n}, x_{n-1}, \cdots x_{ t+1})$, and $\x_{1}= (x_{ t}, x_{t-1}, \cdots x_{1})$
where  
$``\parallel"$  denotes 
concatenation of two vectors; and $(ii)$ 
$s'$, suitably padded, is written as a sequence $(s'_1,\cdots,s'_r)$ where $s'_i\in GF(2^{n-t})$, $\forall i \in \{1,\cdots,r\}$, 
and $s=(s_2,s_1)$ where $s_2\in GF(2^{n-t})$ and $s_1\in GF(2^t)$.
%

\begin{lemma}\label{claim:rootx} 
The lemma has two parts.
\begin{enumerate}
\item[(i)] 
The number of 
 $\x = (\x_{2}  \parallel \x_{1})$ that satisfies  the following  two equations (in  $GF(2^t)$)  for two values of $v$ and $v_f$: 
 {
 \begin{eqnarray}
v&=& h(\x, (s', s))=\big[(\x_{2})^{r+3} + {\sum}_{i=1}^r s'_i(\x_{2})^{i+1} + 
s_2 \x_{2}\big]_{1\cdots t} + (\x_{1})^{3} +s_1 \x_{1} \label{eqn:robxappn1}\\
v_f&=&h(\x, s_f', s_f)=\big[(\x_{2})^{r+3} + {\sum}_{i=1}^r s'_{f,i}(\x_{2})^{i+1} +  s_{f,2} \x_{2}\big]_{1\cdots t} 
+ (\x_{1})^{3} +s_{f,1} \x_{1},  \label{eqn:improbappn21}
 \end{eqnarray}
 }
is at most $3(r+1)2^{n-2t}$.

 In these equations, 
 $\x_{2}, s_2, s_{f,2}, s_{f,1}',..~,s_{f,r}', s'_{1},..~,s'_{r} \in GF(2^{n-t})$,
 $v,v_f,\x_{ 1},s_1,s_{f,1}\in GF(2^t)$, 
and \\
   $((s'_{f,1},\cdots,s'_{f,r}),(s_{f,2},s_{f,1})) \ne ((s'_{1},\cdots,s'_{r}),(s_2, s_1))$.

\item[(ii)] Let $\x = (\x_{ 2}  \parallel \x_{1})$ and
$\x' = (\x'_{ 2}  \parallel \x'_{1})$
satisfy 
$v=h(\x, (s', s))$ and $v_f=h(\x', s'_f, s_f)$, respectively,  where $s'_f$ and $s_f$ are defined as in $(i)$.
Assume  $\x=\x' + \e$ for  some $\e=(\e_2 \parallel \e_1) \in GF(2^n)$, 
$\e_2\in GF(2^{n-t}), \e_1\in GF(2^t)$ and
$\e \ne {\bf 0}$. 
 Then the number of $\x' = (\x'_{ 2}  \parallel \x'_{1})$ that satisfies  the following equations:
{
\begin{eqnarray}
v&=&\big[(\x'_{2} + \e_2)^{r+3} + {\sum}_{i=1}^r s'_i(\x'_{2} + \e_2)^{i+1} 
+ s_2 (\x'_{2} + \e_2)\big]_{1\cdots t}  + (\x'_{ 1} + \e_1)^{3} +s_1 (\x'_{ 1} + \e_1)  \label{eqn:robxappn11}\\
v_{f}&=&\big[(\x'_{2})^{r+3} + {\sum}_{i=1}^r s'_{f,i}(\x'_{ 2})^{i+1} + s_{f, 2} \x'_{2}\big]_{1 \cdots t} + (\x'_{1})^{3}+  s_{f,1} \x'_{ 1}, \label{eqn:robx11} 
\end{eqnarray}
} 
 is at most $(r+3)(r+2)2^{n-2t}$ where,
 $(\e_2 \parallel \e_1)$ is a 
 non-zero vector 
  and \\$(v_{f},(s'_{f, 1},\cdots,s'_{f,r }),(s_{f,2},s_{f,1})) \ne (v,(s'_{1},\cdots,s'_{r}),(s_2,s_1))$. 

\end{enumerate}
\end{lemma}
\begin{proof} 
$(i)$ From equation~\ref{eqn:robxappn1} and equation~\ref{eqn:improbappn21}, we have 
\begin{align} \label{mac:seconddiff21}
v-v_f=&\big[{\sum}_{i=1}^r (s'_{i} - s'_{f,i})(\x_{2})^{i+1} + (s_2-s_{f,2})\x_{2}\big]_{1 \cdots t} +  (s_1 - s_{f,1})\x_{1}
\end{align}
 where arithmetic operations are  
in the corresponding binary extension finite fields. 
    If $(s_1=s_{f,1})$, then \\$((s'_1,\cdots,s'_r),s_2) \ne ((s'_{f,1},\cdots,s'_{f,r}),s_{f,2})$ as $(s'_f,s_f) \ne (s',s)$. Therefore, the degree of the equation~\ref{mac:seconddiff21} in $\x_{2}$ is at most $(r+1)$. The term $\big[{\sum}_{i=1}^r (s'_{i} - s'_{f,i})(\x_{2})^{i+1} +(s_2-s_{f,2})\x_{2}\big]$ takes on each element of the field $GF(2^{n-t})$ at most $(r+1)$ times as $\x_{2}$ varies.  This is because the degree  of the polynomial 
    is $(r+1)$ and so 
    there are at most  $(r+1)(2^{n-t}/2^t)=(r+1)2^{n-2t}$ values of $\x_{2}$ that satisfy  equation~\ref{mac:seconddiff21}. 

Equation~\ref{eqn:robxappn1}, for fixed $v_f$ and $\x_{ 2}$,  is a polynomial of degree three, and hence for each value of $\x_{ 2}$,  will be satisfied by 
at most three values of $\x_{1}$, and
so there are at most $3(r+1)2^{n-2t}$ values of $(\x_{2} \parallel \x_{1})$ that satisfy both  equations~\ref{eqn:robxappn1} and 
~\ref{mac:seconddiff21}.

If $(s_1 \ne s_{f,1})$,  we use equation~\ref{mac:seconddiff21} to express $\x_{ 1}$ as a polynomial in $\x_{2}$, 
and by substituting it in equation~\ref{eqn:robxappn1},  obtain \\$v=[-(s_1 - s_{f,1})^{-3}(s_{r} - s_{f,r})^3 (\x_{2})^{3(r+1)} ]_{_{1 \cdots t}} + g(\x_{2})$ for some polynomial  $g(\x_{ 2})$ of degree at most $3r+2$. Therefore, there are at most $3(r+1)2^{n-2t}$ values of $\x_{2}$ that satisfy this equation. From  equation~\ref{mac:seconddiff21}, 
for each value of $\x_{2}$, there is a unique $\x_{1}$ that satisfies the equation.

Therefore, in both cases, there are at most $3(r+1)2^{n-2t}$ values of $(\x_{2}|| \x_{1})$ that satisfy both the equation~\ref{eqn:robxappn1} and equation~\ref{mac:seconddiff21}.

$(ii)$ 
From equation~\ref{eqn:robxappn11} and equation~\ref{eqn:robx11}, we have 
{
\begin{align}\nonumber 
v-v_{f}=&\Big[\big[(\x_{ 2} + \e_2)^{r+3} + {\sum}_{i=1}^r s'_i(\x_{ 2} + \e_2)^{i+1} +    s_2 (\x_{ 2} + \e_2)\big]_{1\cdots t} + (\x_{ 1} + \e_1)^{3} +   s_1 (\x_{ 1} + \e_1) \Big] \\\label{eqn:robusttag1111} 
& -\Big[\big[(\x_{ 2})^{r+3} +   {\sum}_{i=1}^r s'_{f,i }(\x_{ 2})^{i+1} + s_{f,2 } \x_{2}\big]_{1 \cdots t} 
+ (\x_{1})^{3} + s_{f,1 } \x_{1}\Big].
\end{align}
}

This is an equation in two indeterminates $\x_{ 2}$ and $\x_{ 1}$ of degree at most $(r+2)$. The  equation~\ref{eqn:robx11} is also an equation in two indeterminants $\x_{ 2}$ and $\x_{1}$ of degree at most $(r+3)$. Since 
{\footnotesize $(v_{f},(s'_{f,1 },\cdots,s'_{f,r }),(s_{f,2},s_{f,1 })) \ne (v,(s'_{1},\cdots,s'_{r}),(s_2,s_1))$},
by B\'{e}zout's  theorem~\cite{Coolidge1959,BEZOUTtheorem}, recalled in Section~\ref{appn:bezout},  we have that there are at most $(r+3)(r+2)2^{n-t}/2^t=(r+3)(r+2)2^{n-2t}$ values of $(\x'_{2} \parallel \x'_{1})$ (i.e. $\x'$) that satisfy both  equation~\ref{eqn:robusttag1111} and equation~\ref{eqn:robx11}.

\end{proof}

\begin{theorem}[Ciphertext integrity (INT-$(1;q_d)$-CTXT)]\label{mac2:ctxt}
For an adversary that makes at most one encapsulation query and $q_d$ decapsulation queries, 
the ciphertext integrity defined in Definition~\ref{def: integrity} is 
broken with probability at most 
\begin{align} \nonumber
&q_d(r+3)(r+2)2^{n+\ell -t} \max \{\mathbb{E}_{\z \leftarrow \Z } \Big[ \max_\x\sum_{\y':\Pr(\x|\y')\geq 2^{-\nu}} P(\y'|\Z=\z)\Big], 
\mathbb{E}_{\z \leftarrow \Z } \Big[ \max_\y\sum_{\x':\Pr(\x'|\y)\geq 2^{-\nu}} P(\x'|\Z=\z)\Big]
\}
\end{align}
 For the above number of queries
the  iKEM $\mathsf{ikem}_{cca}$ 
 construction~\ref{ikem:cca} establishes a secret key of length 
 {
\begin{eqnarray*} 
\ell &\le& t + 
\min\{-\log (\mathbb{E}_{\z \leftarrow \Z } \Big[ \max_\x\sum_{\y':\Pr(\x|\y')\geq 2^{-\nu}} P(\y'|\Z=\z)\Big]), 
-\log(\mathbb{E}_{\z \leftarrow \Z } \Big[ \max_\y\sum_{\x':\Pr(\x'|\y)\geq 2^{-\nu}} P(\x'|\Z=\z)\Big])\} \\ 
&& - n - \log\big(\frac{q_d (r+3)(r+2)}{\delta}\big),
\end{eqnarray*}
}
\remove{
\begin{align}\nonumber
\ell &\le t + \\\nonumber
&\quad \min\{-\log (\mathbb{E}_{\z \leftarrow \Z } \Big[ \max_\x\sum_{\y':\Pr(\x|\y')\geq 2^{-\nu}} P(\y'|\Z=\z)\Big]), \\\nonumber
&\quad -\log(\mathbb{E}_{\z \leftarrow \Z } \Big[ \max_\y\sum_{\x':\Pr(\x'|\y)\geq 2^{-\nu}} P(\x'|\Z=\z)\Big])\} \\\nonumber
& \quad - n - \log\big(\frac{q_d (r+3)(r+2)}{\delta}\big),
\end{align}
} 
that is $\delta$-INT-$(1;q_d)$-CTXT  secure.
\end{theorem}

 \begin{proof}
The proof   uses Lemma~\ref{lemma:fuzzy},~\ref{maxprob} and~\ref{claim:rootx}. 
We first 
provide an outline of the main  proof steps, and then  expand each step.

Let
 \psucc~ denote the maximum success probability of the adversary that has access to $(k,c)$ (i.e.  response to an encapsulation query) and constructs a forged ciphertext  $c_f=(v_f, s'_f, s_f)$ where $c_f\neq c$. 
That is,  \psucc~  is the highest success probability of constructing $c_f$ that is accepted by the decapsulation algorithm Algorithm~\ref{alg:iK.dec}. Let \fpsucc ~denote  the expected  value of  \psucc~ 
over all 
query responses $(k,c)$, and \fdpsucc~ denote the expected final success probability with one encapsulation,
 and $q_d$ decapsulation query.

The upper bound on \fdpsucc~ will be obtained in three steps: (1)  bounding \psucc, 
(2)   bounding \fpsucc~ by
finding the expectation  over the random  variables corresponding to the adversary's information, that is the received response $(k,c)$, and finally (3) bounding \fdpsucc ~ that takes into account the decapsulation queries.

\remove{
That is, it will be computed  as follow. 
{\footnotesize
\begin{align}\nonumber
&\mathbb{E}_{(k,c, \z) \leftarrow (K,C, \Z )}\Big[\mathsf{Pr}\big[ v_f=h(\x_f,s'_f,s_f) \text{ $|$ } K=k, C=c, \Z = \z \big]\Big],
\end{align}
}
where $K$, $C$ and $\Z$ are the random variables corresponding to $k$, $c$ and $\z$ respectively, and the expectation is taken over the randomness of $P_{\X|K=k,C=c, \Z=\z}$.    
}

\vspace{1mm}
{\bf Step 1. Bounding  \psucc.}

The adversary has the  key and the ciphertext pair $(k, c) = (k,(v, s', s))$, 
where $v$ is computed using Alice's secret key $\x$ and 
$h(\x,(s',s))$  given 
by the equation~\ref{defn:uhash} (section \ref{ikem:robcca}). The ciphertext  will be accepted by the decapsulation algorithm  ikem.Dec() (Algorithm~\ref{alg:iK.dec}) using Bob's key $\y$ with probability at least $1-\epsilon$.

A forged ciphertext  $c_f=(v_f, s'_f, s_f)$ that is accepted by the decapsulation algorithm 
must  pass the test 
$v\stackrel{?}=h(\x',(s', s ))$ for a unique $\x' \in {\cal R} $  that is found by the decapsulation  algorithm using Bob's key $\y$.
Thus a ciphertext that is accepted by the decapsulation algorithm must be generatable by the generation Algorithm 5 using some {\em (guessed) key}.  We call  ciphertexts that satisfy $v =h(\x',s', s )$  as {\em well-formed}.

{\em 
We assume the adversary can only make a well-formed $c_f$ 
by using the encapsulation 
algorithm (Algorithm~\ref{alg:iK.Enc}) for a guessed key. } 
That is there is no shortcut algorithm can be used by the adversary to  generate a new well-formed ciphertext from other available information.
This assumption holds if the encapsulation algorithm is modelled as a random function (random oracle) for the generation of $c$.

The encapsulation algorithm is deterministic, and so \psucc~ can be 
obtained by,
\begin{enumerate}
    \item  Finding $P_S$, the best guessing probability of a key $\x_f$ that satisfies
    $\x_f \in {\cal R}$
   for Bob's (unknown) $\y$. 
   We use Lemma \ref{maxprob}, part $(ii)$, that assumes the best guessing probability is by using one of the two direct guessing strategies outlined in the lemma.

    \item Take into account the number  of $\x'\neq \x_f$ that results in the  same $c_f=(v_f, s'_f,s_f) $ 
    that is constructed using the key $\x_f$.
      An upper bound on this  number, denoted by $L$,  is obtained 
      in Lemma~\ref{claim:rootx}.

\end{enumerate}


{\bf Step 2. \fpsucc: Expectation over the adversary's view.}
For fixed $(s',s)$, 
let $K$, $C$, $\X $, $\X_{1}$, $\X_{2}$, $V$  be random variables corresponding to $k$, $c$, $\x $, $\x_{1}$, $\x_{2}$ and $v$  respectively, where 
the randomness is over $P_{\X\Y\Z}$. 
The adversary has 
$(k,c)$ derived from  $\x$, 
guesses $\x_f$, and generates the 
ciphertext $c_f=(v_f,s'_f,s_f)$ where 
$v_f =h(\x_f,s'_f,s_f)$.

We consider two cases: $(i)$ the adversary's  guess  $\x_f = \x$ 
where $\x$  is Alice's key, and $(ii)$ 
the adversary's guessed key $ \x_f= \x' \ne \x$ where 
$\x'\in \mathcal{R}$ for the unknown $\y$. 
Let the success probabilities of the adversary in generating a ciphertext $c_f$ corresponding to the above two cases be $\delta_{\x}$ and $\delta_{\x'}$ respectively. 
The decapsulation algorithm $\ikemd(\cdot)$ searches for a {\em unique element} in $\mathcal{R}$ and so  only one of the above two cases will occur,  and 
the success probability of the adversary in generating a $c_f$ will be 
 \begin{eqnarray} \label{eq:forgetag}
     \mbox{\fpsucc} =\max\{\delta_{\x},\delta_{\x'}\}
\end{eqnarray}

where  probability is over $P_{\X\Y\Z}$. 


\vspace{1.5mm}
{\bf  Computing $\delta_{\x}$.} 
The success probability of forging a ciphertext, given a key and ciphertext  pair $(k,c)$, is: 
{
\begin{align}\nonumber
&\mathbb{E}_{(k,c, \z) \leftarrow (K,C, \Z )}\Big[\mathsf{Pr}\big[ v_f=\big[(\x_{2})^{r+3} + {\sum}_{i=1}^r s'_{f,i}(\x_{2})^{i+1} +s_{f,2} \x_{2}\big]_{1\cdots t}  + (\x_{1})^{3} +s_{f,1} \x_{1} \text{ $|$ } K=k, C=c, \Z = \z \big]\Big].
\end{align}
}
 The  known 
 ciphertext $c=(v,s',s)$ and the forged ciphertext $(v_f,s'_f,s_f)$ must satisfy equation~\ref{eqn:robxappn1} and equation~\ref{eqn:improbappn21}, respectively, with $(v_f,s'_f,s_f) \ne (v,s',s)$.
 Note that if $(s'_f,s_f)=(s',s)$, then it must be that $v_f\ne v$ and because $h(\x, (s',s))$ is a single value, 
$v_f \ne h(\x,s_f,s) $ and $\ikemd(\cdot)$ will reject,  which is  a contradiction. Therefore, we only need to consider the case that $(s'_f,s_f) \ne (s',s)$.

From Lemma~\ref{claim:rootx} part $(i)$, the number of 
$\x=(\x_{2} \parallel \x_{1})$  
that satisfy   both   equation~\ref{eqn:robxappn1} and equation~\ref{eqn:improbappn21} is at most \\$3(r+1)2^{n-2t}$. 

Since the adversary is  given a key and ciphertext pair $\big(k,(v,s',s)\big)$,  from Lemma~\ref{maxprob} part (2) and equation~\ref{eqn:guessprob}, we have that the adversary can guess   $(\x_{2}\parallel \x_{1})$ with probability at most 
\begin{align} \label{guessprob}
\max \{&\max_{\x}\sum_{\y':\Pr(\x|\y')\geq 2^{-\nu}} P(\y'|\z,v,k), 
\max_\y\sum_{\x':\Pr(\x'|\y)\geq 2^{-\nu}} P(\x'|\z,v,k)
\},
\end{align}
where $\Z $ 
is the attacker's initial side information. Since $|k|=\ell$ and $|v|=t$, from Lemma~\ref{lemma:fuzzy},
 we have  
\begin{align} \nonumber
&\mathbb{E}_{(k, c, \z) \leftarrow (K,C, \Z )}\Big[\max \{
\max_\x\sum_{\y':\Pr(\x|\y')\geq 2^{-\nu}} P(\y'|\Z=\z,V=v,K=k), 
\max_\y\sum_{\x':\Pr(\x'|\y)\geq 2^{-\nu}} P(\x'|\Z=\z,V=v,K=k)
\}\Big] \\\label{eqn3:condent3}
&\le 2^{ t + \ell} \mathbb{E}_{\z \leftarrow \Z }\Big[\max \{\max_\x\sum_{\y':\Pr(\x|\y')\geq 2^{-\nu}} P(\y'|\Z=\z), 
\max_\y\sum_{\x':\Pr(\x'|\y)\geq 2^{-\nu}} P(\x'|\Z=\z)
\}\Big].
\end{align}
%

Therefore,
{
\begin{align}\nonumber
\delta_\x &= \text{
Success probability of the  adversary with }  (v_f,s'_f,s_f),  
\text{ 
 when verified with $\x $, given the pair $(k,(v,s',s))$  } \\
\nonumber
&=\mathbb{E}_{(k,c, \z) \leftarrow (K,C, \Z )}\Big[\mathsf{Pr}\big[ v_f=\big[(\x_{2})^{r+3} + {\sum}_{i=1}^r s'_{f,i}(\x_{2})^{i+1} + 
s_{f,2} \x_{2}\big]_{1\cdots t}+ (\x_{1})^{3}+s_{f,1} \x_{1} \text{ $|$ } K=k, C=c, \Z = \z \big]\Big] \\\nonumber
&=\mathbb{E}_{(k,c, \z) \leftarrow (K,C, \Z )}\Big[\mathsf{Pr}\big[ v_f=\big[(\x_{2})^{r+3} + {\sum}_{i=1}^r s'_{f,i}(\x_{2})^{i+1} + 
s_{f,2} \x_{2}\big]_{1\cdots t}+ (\x_{1})^{3}+s_{f,1} \x_{1} \\\label{eqn1:xsub1}
&\qquad\qquad\qquad\qquad\qquad \wedge v= \big[(\x_{2})^{r+3} + {\sum}_{i=1}^r s'_i(\x_{2})^{i+1} +s_2 \x_{2}\big]_{1\cdots t} 
+ (\x_{1})^{3} +s_1 \x_{1} \text{ $|$ } K=k, C=c,\Z = \z \big]\Big] \\\nonumber 
&=\mathbb{E}_{(k,c, \z) \leftarrow (K, C,\Z )}\Big[\mathsf{Pr}\Big[ v-v_f=\big[{\sum}_{i=1}^r (s'_{i} - s'_{f,i})(\x_{2})^{i+1} 
+(s_2-s_{f,2})\x_{2}\big]_{1 \cdots t} + (s_1 - s_{f,1})\x_{1} \\\label{eqn2:xsub2}
&\qquad\qquad\qquad\qquad\qquad \wedge v= \big[(\x_{2})^{r+3} + {\sum}_{i=1}^r s'_i(\x_{2})^{i+1} +s_2 \x_{2}\big]_{1\cdots t} 
+ (\x_{1})^{3} +s_1 \x_{1} \text{ $|$ } K=k, C=c,\Z = \z \Big]\Big] \\\label{eqn4:prob1} 
&\le \mathbb{E}_{(k,c, \z) \leftarrow (K,C, \Z )}\big[3(r+1)2^{n-2t}\cdot 
\max \{\max_{\x}\sum_{\y':\Pr(\x|\y')\geq 2^{-\nu}} P(\y'|\Z=\z,v,k), 
\max_\y\sum_{\x':\Pr(\x'|\y)\geq 2^{-\nu}} P(\x'|\Z=\z,v,k)
\} \big] \\\label{eqn5:prob1}
&\le 3(r+1)2^{n-2t}2^{ t + \ell} \mathbb{E}_{\z \leftarrow \Z }\Big[
\max \{\max_\x\sum_{\y':\Pr(\x|\y')\geq 2^{-\nu}} P(\y'|\Z=\z), 
\max_\y\sum_{\x':\Pr(\x'|\y)\geq 2^{-\nu}} P(\x'|\Z=\z)
\}\Big] \\\label{eqn:forgedletax} 
&= 3(r+1)2^{n+\ell - t} \mathbb{E}_{\z \leftarrow \Z }\Big[
\max \{\max_\x\sum_{\y':\Pr(\x|\y')\geq 2^{-\nu}} P(\y'|\Z=\z), 
\max_\y\sum_{\x':\Pr(\x'|\y)\geq 2^{-\nu}} P(\x'|\Z=\z)
\}\Big]
\end{align}
}

where  equation~\ref{eqn2:xsub2} is obtained    from subtracting the two equations within the probability expression in   equation~\ref{eqn1:xsub1}; 
 equation~\ref{eqn4:prob1} follows from  equation~\ref{guessprob};  equation~\ref{eqn5:prob1} follows from  equation~\ref{eqn3:condent3}.  The expectation is taken over the distribution of $P_{\X|K=k, C=c,\Z=\z}$. 

\vspace{0.5mm}
{\bf Computing  $\delta_{\x'}$.}  
\remove{
We bound 
the success probability of 
that a 

The forged ciphertext $(v_f, s'_f,s_f)$ corresponds  to $\x' \neq \x $ where 
}

Let $\x_f= \x' \neq \x $.  The forged ciphertext $(v_f, s'_f,s_f)$ will be, 
{
\begin{eqnarray} 
v_f&=&h\big(\x' , (s'_f,s_f)\big)  \nonumber \\ 
 &=&\big[(\x'_{2})^{r+3} + {\sum}_{i=1}^r s'_{f,i}(\x'_{2})^{i+1} + 
 s_{f, 2} \x'_{2}\big]_{1 \cdots t} 
+ (\x'_{ 1})^{3}+  s_{f,1} \x'_{ 1} \label{eqn:robx1} 
\end{eqnarray}
}

\remove{

\vspace{-.5em}
{\small
\begin{align} \nonumber
v_f&=h\big(\x' , (s'_f,s_f)\big) \\\nonumber 
 &=\big[(\x'_{2})^{r+3} + {\sum}_{i=1}^r s'_{f,i}(\x'_{2})^{i+1} + \\\label{eqn:robx1} 
&\quad s_{f, 2} \x'_{2}\big]_{1 \cdots t} + (\x'_{ 1})^{3}+  s_{f,1} \x'_{ 1}
\end{align}
}
}
, where $(v_f,s'_f,s_f) \ne (v,s',s)$, 
and all other variables are defined as in Lemma~\ref{claim:rootx} and equation~\ref{defn:uhash}.

 Let  $(k,c)  =(k,(v,s',s)$ is constructed using $\x$. 
We can write 
$ \x= \x'+\e$ for some (unknown) vector $\e=(\e_2 \parallel \e_1) \in GF(2^{n})$ and $(\x_{2}\parallel \x_{1})=((\x'_{2} + \e_2) \parallel (\x'_{ 1} + \e_1))$.
\remove{
This is true because the decapsulation algorithm $\ikemd(\cdot)$ succeeds if there is a unique element $\x'$ in the set $\mathcal{R}$ such that the received hash value $v_f$ matches with $h(\x',s'_f,s_f)$, and hence $e$ is a non-zero fixed constant for a given $(k,c)$ (i.e. $(k,(v,s',s))$) and does not vary with $\x$ that satisfies $v=h(\x,(s',s))$. 
}
%
Replacing $\x_{2}$ and $ \x_{1}$ with $(\x'_{2} + \e_2)$ and  $ (\x'_{ 1} + \e_1)$, respectively in equation~\ref{defn:uhash}, we obtain 

{
\begin{align}
v&=\big[(\x'_{2} + \e_2)^{r+3} + {\sum}_{i=1}^r s'_i(\x'_{2} + \e_2)^{i+1} 
+ s_2 (\x'_{2} + \e_2)\big]_{1\cdots t}+ (\x'_{ 1} + \e_1)^{3} +s_1 (\x'_{ 1} + \e_1) \label{eqn:robxappn111}
\end{align}
}

  From Lemma~\ref{claim:rootx} part $(ii)$, the number of $(\x'_{2}  \parallel \x'_{ 1})$ (i.e. $\x'$) that satisfy both the equation~\ref{eqn:robxappn111} and equation~\ref{eqn:robx1} is at most $(r+3)(r+2)2^{n-2t}$.

 Let $\X'_{ 1}$, $\X'_{2}$ and $\X'$ denote  the random variables corresponding to $\x'_{ 1}$, $\x'_{2}$ and $\x'$ respectively.

Define $h_1(\x_{2},s',s)=(\x_{2})^{r+3} + {\sum}_{i=1}^r s'_i(\x_{2})^{i+1} + s_2 \x_{2}$. 
{
\begin{eqnarray}\nonumber
\delta_{\x'} &=& \text{ \psucc~ with }  (v_f,s'_f,s_f)
\text{ corresponding to } \x'  
\\\nonumber
&=&\mathbb{E}_{(k,c, \z) \leftarrow (K,C, \Z )}\Big[
\Pr\big[  v_f=\big[(\x'_{2})^{r+3} + 
{\sum}_{i=1}^r s'_{f,i}(\x'_{2})^{i+1}  +s_{f, 2} \x'_{2} \big]_{1 \cdots t}  
+ (\x'_{ 1})^{3} + s_{f,1} \x'_{ 1} | K=k, C=c, \Z = \z \big]\Big] \\\nonumber
&=&\mathbb{E}_{(k,c, \z) \leftarrow (K,C, \Z )}\Big[\mathsf{Pr}\big[ v_f=\big[h_1(\x'_{2},s'_f,s_f)\big]_{1 \cdots t}  
+  (\x'_{ 1})^{3} + s_{f,1} \x'_{ 1} 
\wedge \\\nonumber 
&&\qquad \qquad\qquad\qquad\quad v = \big[h_1(\x_{2},s',s)\big]_{1 \cdots t} + (\x_{1})^{3} +s_1 \x_{1} | K=k, C=c, \Z = \z \big]\Big] \\\nonumber
&\le& \mathbb{E}_{(k,c, \z ) \leftarrow (K,C, \Z )}\big[ (r+3)(r+2)2^{n-2t} 
\max \{\max_{\x}\sum_{\y':\Pr(\x|\y')\geq 2^{-\nu}} P(\y'|\Z=\z,V=v,K=k), \\\label{eqn8:sub8}
&&\qquad \qquad\qquad\qquad\quad\qquad \qquad\qquad\qquad\quad \max_\y\sum_{\x':\Pr(\x'|\y)\geq 2^{-\nu}} P(\x'|\Z=\z,V=v,K=k)
\} \big]  \\\label{eqn7:condprob7}
&\le & (r+3)(r+2)2^{n-2t}2^{ t + \ell} \mathbb{E}_{\z \leftarrow \Z }\Big[
\max \{\max_\x\sum_{\y':\Pr(\x|\y')\geq 2^{-\nu}} P(\y'|\Z=\z), 
\max_\y\sum_{\x':\Pr(\x'|\y)\geq 2^{-\nu}} P(\x'|\Z=\z)
\}\Big] 
\\\label{eqn:robsub2}
&=& (r+3)(r+2)2^{n+\ell -t} \mathbb{E}_{\z \leftarrow \Z }\Big[
\max \{\max_\x\sum_{\y':\Pr(\x|\y')\geq 2^{-\nu}} P(\y'|\Z=\z), 
\max_\y\sum_{\x':\Pr(\x'|\y)\geq 2^{-\nu}} P(\x'|\Z=\z)
\}\Big]
\end{eqnarray}
}

where  equation~\ref{eqn8:sub8} follows from Lemma~\ref{maxprob} part (2) and equation~\ref{eqn:guessprob} since the adversary is  given a key and ciphertext pair $\big(k,(v,s',s)\big)$;  equation~\ref{eqn7:condprob7} follows from Lemma~\ref{lemma:fuzzy}.

Therefore, from equations~\ref{eq:forgetag},~\ref{eqn:forgedletax} and~\ref{eqn:robsub2}, we have that after one encapsulation query, the   probability that an adversary will be able to forge a ciphertext is at most 

\begin{eqnarray*}\nonumber
&&(r+3)(r+2)2^{n+\ell -t} \mathbb{E}_{\z \leftarrow \Z }\Big[ \max
\{\max_\x\sum_{\y':\Pr(\x|\y')\geq 2^{-\nu}} P(\y'|\Z=\z),  
\max_\y\sum_{\x':\Pr(\x'|\y)\geq 2^{-\nu}} P(\x'|\Z=\z)
\}\Big].
\end{eqnarray*}

\remove{
{\small
\begin{align}\nonumber
&(r+3)(r+2)2^{n+\ell -t} \mathbb{E}_{\z \leftarrow \Z }\Big[\\\nonumber
&\quad \max \{\max_\x\sum_{\y':\Pr(\x|\y')\geq 2^{-\nu}} P(\y'|\Z=\z), \\\nonumber
&\quad \max_\y\sum_{\x':\Pr(\x'|\y)\geq 2^{-\nu}} P(\x'|\Z=\z)
\}\Big].
\end{align}
}
}

\vspace{1mm}
{\bf Step 3: \fdpsucc: Including decapsulation queries.}
For each decapsulation query,  the adversary  receives  either a key, %
if the forged ciphertext is accepted by the decapsulation algorithm, and \remove{
corresponds  to 
a unique element of $\mathcal{R}$.
}
or $\perp$, otherwise.
%
The adversary succeeds with the first query that is successful.
After $q_d$ unsuccessful decapsulation queries, 
 the size of the set of  possible guesses will reduce by $\log(q_d)$.
\remove{will gain $\log(q_d)$ bit information about 
elements of $\mathcal{R}$.
and hence the guessing 
entropy of the elements of 
$\mathcal{R}$ 
decreases by $\log(q_d)$.
}
Thus,  
after one encapsulation query and $q_d$ decapsulation queries, \fpsucc~ is bounded by 
\begin{eqnarray}\nonumber
\mbox{\fpsucc} &\leq&2^{\log(q_d)}(r+3)(r+2)2^{n+\ell -t} \mathbb{E}_{\z \leftarrow \Z }\Big[
\max \{\max_\x\sum_{\y':\Pr(\x|\y')\geq 2^{-\nu}} P(\y'|\Z=\z), 
\max_\y\sum_{\x':\Pr(\x'|\y)\geq 2^{-\nu}} P(\x'|\Z=\z)
\}\Big] \\\nonumber
&=&q_d(r+3)(r+2)2^{n+\ell -t} 
\max \{\mathbb{E}_{\z \leftarrow \Z } \Big[ \max_\x\sum_{\y':\Pr(\x|\y')\geq 2^{-\nu}} P(\y'|\Z=\z)\Big], 
\\\nonumber
&&\qquad\qquad\qquad\qquad\qquad\qquad\quad \mathbb{E}_{\z \leftarrow \Z } \Big[ \max_\y\sum_{\x':\Pr(\x'|\y)\geq 2^{-\nu}} P(\x'|\Z=\z)\Big]
\}\le \delta. 
\end{eqnarray}

Therefore, if 
\begin{align}\nonumber
&\ell \le t + 
\min\{-\log (\mathbb{E}_{\z \leftarrow \Z } \Big[ \max_\x\sum_{\y':\Pr(\x|\y')\geq 2^{-\nu}} P(\y'|\Z=\z)\Big]), 
-\log(\mathbb{E}_{\z \leftarrow \Z } \Big[ \max_\y\sum_{\x':\Pr(\x'|\y)\geq 2^{-\nu}} P(\x'|\Z=\z)\Big])\} \\\nonumber
& \quad - n - \log\big(\frac{q_d (r+3)(r+2)}{\delta}\big),
\end{align}
the iKEM $\mathsf{ikem}_{cca}$ given in construction~\ref{ikem:cca} is $\delta$-INT-$(1;q_d)$-CTXT  secure.
\end{proof}

\noindent
\begin{coro}[CCA security]\label{thm:ccasec}
{\em The iKEM construction~\ref{ikem:cca} is an 
IND-$(0;q_d)$-CCA secure iKEM.}
\end{coro}

\begin{proof}
 According to 
Theorems~\ref{Thm:ikemotsecurity} and~\ref{mac2:ctxt} 
the iKEM construction~\ref{ikem:cca} is both IND-$q_e$-CEA secure with  $q_e$ encapsulation queries, 
 and INT-$(1;q_d)$-CTXT secure with one encapsulation and $q_d$ decapsulation queries. 
 Then according to Theorems~\ref{theo:comp}, the iKEM is also IND-$(0;q_d)$-CCA secure with $q_d$ decapsulation queries 
 and zero encapsulation query. 
 Therefore, if the parameters $\ell$, $t$ and $\nu$ are chosen to satisfy both  Theorems~\ref{Thm:ikemotsecurity} and 
 ~\ref{mac2:ctxt}, then the iKEM construction~\ref{ikem:cca} is also IND-$(0;q_d)$-CCA secure.  
\end{proof}

\section{KEM combiners for iKEM}\label{sec:combiner}
 Cryptographic combiners provide robustness for cryptographic schemes against possible flaws or security breaks of the component schemes.  
Combiners for KEM were introduced
by Giacon et al. \cite{giacon2018kem} who  defined a framework for combining two or more  public-key KEMs.
Our goal in this section is to extend their framework to allow pKEMs to be combined with public key KEMs. 
This is 
well-motivated because iKEMs are post-quantum secure and  so a much wider set of KEMs with post-quantum security becomes available to the system designers.

In this section, 
we first define combiners for combining pKEMs (i.e. in correlated randomness model) with a 
public-key KEM, and 
then give constructions and prove their security.
We focus on combiners for iKEM and public key KEMs because of subtleties of combining security of two tyoes of schemes: security against computationally unbounded and computationally bounded adversaries. 

\textbf{Combiners.} Using the framework of Giacon et al.  \cite{giacon2018kem}, 
for  security  parameter $\lambda$, we define a {\em core function} for combining an iKEM $\ikem=(\ikemg,\ikeme,\ikemd)$
 with correlation generating distribution $P_{\X\Y\Z}$, output key space 
$\ms{KeySP}_{\ikem}(\lambda)=\mc K_1$, and   ciphertext space ${\cal C}_1$, with
a public-key KEM $\ms{kem}=(\ms{kem}.\ms{Gen},\ms{kem}.\ms{Enc},\ms{kem}.\ms{Dec})$ with public-key space ${\cal PK}$,  output key space 
$\ms{KeySP}_{\ms {kem}}(\lambda)=\mc K_2$, and ciphertext space $\mc C_2$. 
 The combiner \\$\ckem_{\ikem,\kem}=(\gckem,\eckem,\dckem)$ is a KEM with three algorithms $\gckem$; $\eckem$; and $\dckem$ for key generation, encapsulation and decapsulation, respectively, that uses a core function,  $\mathsf{W} : \mc K_1\times\mc K_2\times \mathcal{C}_1\times\mc C_2\to \mathcal{K}^*$, to generate a session key in the key space $\mc K^*$, using the algorithms defined in Figure.~\ref{fig:gkemcomb}.


\begin{figure}[!ht]
   \begin{center}
\begin{tabular}{l  l }
     \underline{$\mathbf{Algo}\ \gckem({1^\lambda},P_{\X \Y \Z })$} &  \underline{$\mathbf{Algo}\ \eckem(r_A,pk)$} \\
     { $(r_A ,r_B, r_E)\stackrel{\$}\gets \ikemg(1^{\lambda},
{P_{\X \Y \Z }})$} &$(c_1,k_1)\stackrel{\$}\gets \ikeme(r_A)$\\
      $(pk,sk)\stackrel{\$}\gets \gkem(1^{\lambda})$& $(c_2,k_2)\stackrel{\$}\gets \ekem({pk})$\\
      {Return $(r_A,r_B,r_E,{pk},{sk})$} &${k}\gets\ms W(k_1,k_2,c_1,c_2)$\\
      &{Return $(k,c_1,c_2)$}
    
\end{tabular}
\begin{tabular}{l}
\\
      \underline{$\mathbf{Algo}\ \dckem(r_B,sk,c_1,c_2)$} \\
     $k_1\gets \ikemd(r_B,c_1)$\\
      $k_2\gets \dkem({sk},c_2)$\\
      $\quad \text{If }k_1=\perp \vee\  k_2=\perp:{\text{Return} \perp}$\\
      ${k}\gets\ms W(k_1,k_2,c_1,c_2)$\\
      {Return $k$}
    
\end{tabular}
\end{center}
\vspace{-1em}
    \caption{
    Combining an iKEM with a public-key KEM}
    \label{fig:gkemcomb}
\end{figure}

 One can also define combiners for other combinations of component KEMs, i.e.
two iKEMs , two cKEMs, 
an iKEM and a cKEM, and a cKEM and public key KEM, 
 with private samples $(r_A, r_B, r_E$) 
that generates   a pair $(c_1, k_1)$  where
key $k_1\in \{0,1\}^{\ikeml(\lambda)}$ and 
ciphertext $c_1\in\mc C_1$,  and let 
 KEM $\ms K$ be a public-key with public and private key pair $(pk, sk)$ that generates

\begin{construction}[XOR combiner.]
 Let  $\ikem$  be   an iKEM  with with private samples $(r_A, r_B, r_E)$, and  $\kem$ be a public-key KEM with public and private key pair $(pk, sk)$ that generate keys
$k_1\in \{0,1\}^{\ikeml(\lambda)}$ and $k_2\in \{0,1\}^{\keml(\lambda)}$, respectively,  and let $\{0,1\}^{\ikeml(\lambda)}=\{0,1\}^{\keml(\lambda)}=\{0,1\}^{\ell(\lambda)}$.
The combiner $\ckem_{\ikem,\kem}^\oplus$ with an XOR core function $\ms W$, outputs $k=\ms W(k_1,k_2)=k_1\oplus k_2$  when
none of $k_i$'s for $i\in\{1,2\},$ is $\perp$, and outputs $\perp$ otherwise.
\end{construction}
The following theorem shows that for a given $q_e\geq 0$, the XOR combiner retains the IND-$q_e$-CEA security of the component iKEM. The proof is given in Appendix~\ref{pf:xorcomb}.
\begin{theorem}\label{theo:xorcomb} 
For security parameter $\lambda$,
let $\ikem=(\ikemg,\ikeme,\ikemd)$ be an IND-$q_e$-CEA secure iKEM
that generates $k_1\in \{0,1\}^{\ell(\lambda)}$, and $\ms {kem}=(\kemg,\keme,\kemd)$ be a public-key KEM with the same security parameter that generates $k_2\in \{0,1\}^{\ell(\lambda)}$   of the same length. 
Consider a combiner KEM $\ckem^\oplus_{\ikem,\ms {kem}}$ using the XOR core function that combines $\ikem$ and $\ms {kem}$, and generates the key $k=k_1\oplus k_2$. For a computationally unbounded adversary $\msa$, there exists a computationally unbounded adversary $\msa'$, such that 
\[Adv^{pkind\text{-}q_e\text{-}cea}_{\ckem_{\ikem,\kem}^\oplus,\msa}(\lambda)\leq Adv^{pkind\text{-}q_e\text{-}cea}_{\ikem,\msa'}(\lambda).\]

\end{theorem}
 In the above theorem, computational security of $\ckem^\oplus_{\ikem,\ms {kem}}$   follows \cite[Lemma 1]{giacon2018kem} as an iKEM can be seen as an insecure KEM for polynomial number of queries. 

\vspace{1.5 mm}
{\bf CCA security.}
The XOR combiner cannot retain the IND-$(q_e;q_d)$-CCA security of the component iKEM (with similar reasoning as \cite[Lemma 2]{giacon2018kem}).

 We show the PRF-then-XOR core function in  \cite{giacon2018kem} can   be used to combine an iKEM with a public-key KEM such that,
in addition to resulting in a secure public-key KEM,  if the PRF output is indistinguishable from uniform by a computationally unbounded and  query-bounded  adversary,
  the resulting KEM will be an IND-$(q_e;q_d)$-CCA secure iKEM.


\begin{definition}[PRF and its security]\label{defn:prf}
Let $\lambda$ be a security parameter.  We use $\lambda$ as 
an argument for values  to make dependence on  $\lambda$ as a parameter, explicit. A family of functions 
 $\ms F : \mc K \times \mc X \to \mc Y$,
where $\mc K$,  $\mc X$ and $\mc Y$, respectively, are finite sets corresponding to   
key,   input  
and outputs,  
is a 
 secure PRF, if the  advantage of an adversary in the distinguishing game of PRF, defined in Figure.~\ref{fig4:thm3}, satisfies the following: 
%
\begin{enumerate}[label=(\roman*)]
    \item   {\em Computationally secure PRF}:  For any computationally bounded adversary $\ms B$ with access to $q(\lambda)$ queries,  where $q$ is a   polynomial, the advantage of the adversary defined as, 
%
$Adv^{PRF}_{\ms F,\ms B}(\lambda)\triangleq |\Pr[\mathrm{PRI}^0_{\ms F,\ms B}(\lambda)=1] -\Pr[\mathrm{PRI}^1_{\ms F,\ms B}(\lambda)=1]|$, 
is a negligible function of $\lambda$.

\item   {\em Information theoretic  PRF}:  For any computationally unbounded adversary $\ms U$ with access to $q(\lambda)$ queries, where $q$ is a pre-defined polynomial in $\lambda$, the advantage of the adversary defined as, \\
$Adv^{q\text{-}PRF}_{\ms F,\ms U}(\lambda)\triangleq |\Pr[\mathrm{PRI}^{q\text{-}IND\text{-}0}_{\ms F,\ms U}(\lambda)=1] -\Pr[\mathrm{PRI}^{q\text{-}IND\text{-}1}_{\ms F,\ms U}(\lambda)=1]|$, is a  \textit{small} function $\sigma(\cdot)$  of $\lambda$ i.e. $\sigma(\cdot) \in SMALL $.

\end{enumerate}

The games  $\mathrm{PRI}_{\ms F,\ms B}^b$ (or the games $\mathrm{PRI}^{q\text{-}IND\text{-}b}_{\ms F,\ms U}$)  are chosen using a uniformly random bit $b\in\{0,1\}$.\\

\textbf{{PRF distinguishing game.}}  The PRF distinguishing game for a function family $\ms F : \mc K \times \mc X \to \mc Y$ where $\mc K$,  $\mc X$ and $\mc Y$, respectively, are finite sets corresponding to key,  input and
 output, 
 \remove{
family of functions {
$\ms F_\lambda: \mc K_\lambda\times \mc X_\lambda \to \mc Y_\lambda$, for an indexed families of finite key spaces $\mc K_\lambda$,  input spaces $\mc X_\lambda$ and  finite output spaces $\mc Y_\lambda$,
}}
is defined  
in Figure~\ref{fig4:thm3}.
\begin{figure}[!ht]
\begin{center}
 \parbox{0.5\textwidth}{
\underline{\textbf{Game} $\mathrm{PRI}^{b}_{\ms F,\msa}(\lambda)$} \ \ \ \ \ \ \ \ \ \ \ \ \  \ \ \ \  
\underline{\textbf{Oracle} $\ms{Eval}(x)$}
\vspace{-0.8em}
\begin{multicols}{2}
\centering
\begin{algorithmic}[1]
\State $\mc X\gets \emptyset$
\State $k\stackrel{\$}\gets \mc K$
\State $b'\stackrel{\$}\gets \msa_2^{\ms{Eval}}$
\State Return $b'$
\end{algorithmic}
       \columnbreak
 {\centering
\begin{algorithmic}[1]
\State If $x\in\mc X$: Abort
\State $\mc X=\mc X \cup \{x\}$
\State $y \gets \ms F(k,x)$
\State $y_0\gets y; y_1\stackrel{\$}\gets\mc Y$
\State Return $y_{b}$
\end{algorithmic}
}
     \end{multicols} 
}%
  \caption{
    PRF distinguishing game}
    \label{fig4:thm3}
\end{center}    
\end{figure}

\end{definition}

\begin{construction}[PRF-then-XOR combiner.]\label{const}
 Let  $\lambda$ be the   security parameter.
Consider an iKEM $\ikem$  with private samples $(r_A, r_B, r_E)$   
that generates   a pair $(c_1, k_1)$  where
key $k_1\in \{0,1\}^{\ikeml(\lambda)}$ and 
ciphertext $c_1\in\mc C_1$,  and let 
 KEM $\ms K$  be a public-key with public and private key pair $(pk, sk)$ that generates  a pair $(c_2, k_2)$    
where  key $k_2\in \{0,1\}^{\keml(\lambda)}$ 
and 
ciphertext $c_2\in\mc C_2$. Further, let $\ms F_1: \mc \{0,1\}^{\ikeml(\lambda)}\times \mc C_2\to \mc K$ and $\ms F_2: \{0,1\}^{\keml(\lambda)}\times \mc C_1\to \mc K$ be two PRFs  with information theoretic and computational security, respectively. 

The combiner $\ms{Comb}_{\ikem,\kem}^{\mathrm{PtX}}$ with the core function PRF-then-XOR outputs $\ms W(k_1,k_2,c_1,c_2)=\ms F_1(k_1,c_2)\oplus\ms F_2(k_2,c_1)$  when neither $k_1$ nor $k_2$ is $\perp$, and outputs $\perp$ otherwise.
\end{construction}

\begin{theorem}
\label{thmprfxorcomb}
In Construction~\ref{const}, let $\ikem$ and $\kem$ be an IND-$(q_e; q_d)$-CCA secure
iKEM and an IND-CCA secure KEM, respectively, and let $\ms F_1(\cdot)$ and $\ms F_2(\cdot)$ be   two
PRFs,  with security against a computationally unbounded adversary with $(q_d+1)$ queries,
and a computationally bounded adversary with polynomial number of queries, respectively. Then for any

(a)    computationally   bounded distinguisher $\ms B$, there exists computationally bounded adversaries
$\ms B_1$ and $\ms B_2$ for games $\pind^{cca}_{\kem}$ and $\mathrm{PRI}^{b}_{\ms F_2}$, respectively, such that,
\begin{equation}\label{eq:second}
Adv^{kind\text{-}cca}_{\ckem^{\mathrm{PtX}}_{\ikem,\kem},\msa}(\lambda)\leq 2\Big( Adv^{kind\text{-}cca}_{\kem,\ms B_1}(\lambda)+Adv^{PRF}_{\ms F_2,\ms B_2}(\lambda)\Big),
\end{equation}

(b)   computationally unbounded distinguisher $\msa'$, there exists   a computationally unbounded adversaries $\ms U_1$ and $\ms U_2$ for games $\pind^{(q_e;q_d)\text{-}cca}_{\ikem}$ and $\mathrm{PRI}^{(q_d+1)\text{-}IND\text{-}b}_{\ms F_1}$, respectively, such that
\begin{eqnarray*}\label{eq:third}
&&Adv^{pkind\text{-}(q_e;q_d)\text{-}cca}_{\ckem^{\mathrm{PtX}}_{\ikem,\kem},\msa'}(\lambda)\leq 
2\Big( Adv^{pkind\text{-}(q_e;q_d)\text{-}cca}_{\ikem,\ms U_1}(\lambda)+Adv^{(q_d+1)\text{-}PRF}_{\ms F_1,\ms U_2}(\lambda)\Big).
\end{eqnarray*}
\end{theorem}

\textit{Proof sketch.} 
The proof  for a computationally bounded adversary will follow the approach of 
Theorem 3 in \cite{giacon2018kem},  noting that the iKEM 
will loose its security  when the number of queries exceed the design parameter of iKEM
after repeated queries.   The proof for a computationally unbounded adversary  is given in Appendix~\ref{pf:thmprfxorcomb}. $\qed$

\vspace{1.5mm}
{\em Instantiating PRF for construction~\ref{const}.} 
To construct a PRF with security against a computationally unbounded adversary with access to  
$(q_d+1)$-queries,
we can use a $(q_d+2)$-independent  hash function.  An example construction using polynomials over finite fields is given in \cite[Section 4.1]{TRNG}. \\

The drawback of this PRF however is its large key size.  
We leave  more efficient constructions of 
information-theoretic PRF $\ms F_1(\cdot)$ for the required number of queries for
future  work. 

 Note that 
 security of PRF $\ms F_1(\cdot)$ in the combiner construction~\ref{const} against computationally unbounded adversaries,  does not depend on the number of encapsulation queries to the combiner. Intuitively, this is because in each encapsulation query to combiner, the component 
 iKEM generates a fresh uniform and independent key which is  used as the secret key in PRF $\ms F_1(\cdot)$, and so the output of PRF $\ms F_1(\cdot)$ is 
 independent of previous encapsulation and decapsulation queries.  
\remove{
HMAC {\rd has been proved?? to be a PRF~\cite{BellareHMAC2006} ??? FIND REFERENCE } can be used as the computational PRF $F_2(\cdot)$.
}

\subsection{ Composing a {\em ``combined"} KEM with a DEM}
Security requirements of DEM in Theorem~\ref{theo:composition} is identical to 
Cramer et. al's  \cite[Theorem 7.2]{Cramer2003DesignAA} and so the same DEM can be used for secure hybrid encryption for information theoretically secure KEM and public key KEM. 
The KEM combiner's output will be used with a secure DEM (example construction is given in  \cite{Cramer2003DesignAA}), 
and depending on the security of the component KEM, will result in a  secure hybrid encryption with one of the following security properties: 
\begin{enumerate}[label=(\roman*)]
\item   If the component KEM is a 
secure  iKEM with  IND-$q_e$-CEA (IND-$(q_e;q_d)$-CCA)  security, the combiner's  output key will be secure against computationally unbounded attackers, and the resulting hybrid encryption  provides security according  to Theorem~\ref{theo:composition}, cases (3) and (4).
\item If an IND-CEA (IND-CCA)  cKEM is used as a component in the preprocessing model, the combiner's output will be secure against computationally bounded adversary, and the resulting hybrid encryption will be secure according to Theorem ~\ref{theo:composition}, cases (1) and (2).  
\item If the 
 public-key  KEM is secure,  the resulting hybrid encryption 
 provides security  according  to \cite[Theorem 7.2]{Cramer2003DesignAA}.
\end{enumerate}

\section{Concluding remarks.}\label{sec:conclusion}
 KEM/DEM in the preprocessing model is a natural and useful extension of KEM/DEM paradigm that 
does not require public keys and so computational assumptions.  The paradigm is defined for information theoretic and computational security.
That is each of the two components KEM and DEM, and the final HE, may be defined against a computationally  unbounded or bounded adversaries.
We prove a general composition theorem for KEM and DEM when security of KEM is  against a computationally unbounded or bounded adversary, and  security of DEM is against a computationally bounded adversary. We focused on these combinations of adversaries to obtain a key efficient HE.  Defining DEM with information theoretic security  will lead to HE constructions with high secret key requirement (i.e. essentially similar to one-time-pad).

We also defined and constructed  combiners with provable security that combine KEMs  in preprocessing model 
with public-key KEMs. Efficient and secure construction of core functions for combining iKEM and public key KEMs that retain CCA security of component KEMs is an interesting direction for future work.

An HE in preprocessing model is a private key encryption where the private keys are correlated random strings (not symmetric), and so security notions are  defined similar to symmetric key encryption systems.
Combiners for iKEM and public-key KEM widens the  range of possible KEMs, and allow fuzzy data to be used for communication with  provable security.

Construction of KEMs with computational security in preprocessing model for specific $P_{\X\Y\Z}$ is an interesting direction for future work.

\remove{
We initiated the study of KEM and hybrid encryption in preprocessing model, introduced  information theoretic KEM, and proved composition theorems for secure hybrid encryption. We defined
and constructed secure cryptographic combiners 
for iKEMs and public-key KEMs that ensure 
secure hybrid encryption in the corresponding setups.  
Using iKEM will guarantee post-quantum security, with the unique security property of being secure against offline  attacks, that is particularly important in long-term security.

Our work raises many interesting research   
 questions 
 including construction of 
 fuzzy extractor based  iKEM  with security against higher number of queries, 
 iKEM construction in other correlated randomness setups 
such as  wiretap setting, source model  and satellite model of Maurer \cite{Maurer1993},
and 
construction of computational KEMs in preprocessing model.
}

\bibliographystyle{IEEEtran}
\bibliography{main}

\begin{thebibliography}{10}
\providecommand{\url}[1]{#1}
\csname url@samestyle\endcsname
\providecommand{\newblock}{\relax}
\providecommand{\bibinfo}[2]{#2}
\providecommand{\BIBentrySTDinterwordspacing}{\spaceskip=0pt\relax}
\providecommand{\BIBentryALTinterwordstretchfactor}{4}
\providecommand{\BIBentryALTinterwordspacing}{\spaceskip=\fontdimen2\font plus
\BIBentryALTinterwordstretchfactor\fontdimen3\font minus \fontdimen4\font\relax}
\providecommand{\BIBforeignlanguage}[2]{{%
\expandafter\ifx\csname l@#1\endcsname\relax
\typeout{** WARNING: IEEEtran.bst: No hyphenation pattern has been}%
\typeout{** loaded for the language `#1'. Using the pattern for}%
\typeout{** the default language instead.}%
\else
\language=\csname l@#1\endcsname
\fi
#2}}
\providecommand{\BIBdecl}{\relax}
\BIBdecl

\bibitem{Cramer2003DesignAA}
\BIBentryALTinterwordspacing
R.~Cramer and V.~Shoup, ``Design and analysis of practical public-key encryption schemes secure against adaptive chosen ciphertext attack,'' \emph{{SIAM} J. Comput.}, vol.~33, no.~1, pp. 167--226, 2003. [Online]. Available: \url{https://doi.org/10.1137/S0097539702403773}
\BIBentrySTDinterwordspacing

\bibitem{HERRANZ20101243}
\BIBentryALTinterwordspacing
J.~Herranz, D.~Hofheinz, and E.~Kiltz, ``Some (in)sufficient conditions for secure hybrid encryption,'' \emph{Inf. and Computat.}, vol. 208, no.~11, pp. 1243--1257, 2010. [Online]. Available: \url{https://www.sciencedirect.com/science/article/pii/S089054011000132X}
\BIBentrySTDinterwordspacing

\bibitem{kiltz2006chosen}
\BIBentryALTinterwordspacing
E.~Kiltz, ``{Chosen-Ciphertext Security from Tag-Based Encryption},'' in \emph{Theory of Cryptogr. Conf.}\hskip 1em plus 0.5em minus 0.4em\relax Springer, 2006, pp. 581--600. [Online]. Available: \url{http://link.springer.com/10.1007/11681878{\_}30}
\BIBentrySTDinterwordspacing

\bibitem{kurosawa2004new}
K.~Kurosawa and Y.~Desmedt, ``A new paradigm of hybrid encryption scheme,'' in \emph{Annu. Int. Cryptol. Conf.}\hskip 1em plus 0.5em minus 0.4em\relax Springer, 2004, pp. 426--442.

\bibitem{herranz2006kurosawa}
J.~Herranz, D.~Hofheinz, and E.~Kiltz, ``The kurosawa-desmedt key encapsulation is not chosen-ciphertext secure.'' \emph{IACR Cryptol. ePrint Arch.}, vol. 2006, p. 207, 2006.

\bibitem{Masayuki2005}
M.~Abe, R.~Gennaro, K.~Kurosawa, and V.~Shoup, ``Tag-kem/dem: A new framework for hybrid encryption and a new analysis of kurosawa-desmedt kem,'' in \emph{Advances in Cryptology -- EUROCRYPT 2005}.\hskip 1em plus 0.5em minus 0.4em\relax Springer Berlin Heidelberg, 2005, pp. 128--146.

\bibitem{Shacham2007}
\BIBentryALTinterwordspacing
H.~Shacham, ``A cramer-shoup encryption scheme from the linear assumption and from progressively weaker linear variants,'' Cryptology ePrint Archive, Paper 2007/074, 2007, \url{https://eprint.iacr.org/2007/074}. [Online]. Available: \url{https://eprint.iacr.org/2007/074}
\BIBentrySTDinterwordspacing

\bibitem{SchwabePQTLS2020}
\BIBentryALTinterwordspacing
P.~Schwabe, D.~Stebila, and T.~Wiggers, ``Post-quantum tls without handshake signatures,'' in \emph{Proceedings of the 2020 ACM SIGSAC Conference on Computer and Communications Security}, ser. CCS '20.\hskip 1em plus 0.5em minus 0.4em\relax New York, NY, USA: Association for Computing Machinery, 2020, p. 1461–1480. [Online]. Available: \url{https://doi.org/10.1145/3372297.3423350}
\BIBentrySTDinterwordspacing

\bibitem{Shor1994}
P.~Shor, ``Algorithms for quantum computation: discrete logarithms and factoring,'' in \emph{Proceedings 35th Annual Symposium on Foundations of Computer Science}, 1994, pp. 124--134.

\bibitem{bos2018crystals}
\BIBentryALTinterwordspacing
J.~Bos, L.~Ducas, E.~Kiltz, T.~Lepoint, V.~Lyubashevsky, J.~M. Schanck, P.~Schwabe, G.~Seiler, and D.~Stehle, ``{CRYSTALS - Kyber: A CCA-Secure Module-Lattice-Based KEM},'' in \emph{2018 IEEE Eur. Symp. Secur. Priv. (EuroS{\&}P)}, IEEE.\hskip 1em plus 0.5em minus 0.4em\relax IEEE, apr 2018, pp. 353--367. [Online]. Available: \url{https://ieeexplore.ieee.org/document/8406610/}
\BIBentrySTDinterwordspacing

\bibitem{nistpqccompet}
N.~I. of~Standards and T.~group, ``Post-quantum cryptography standardization,'' \url{https://csrc.nist.gov/Projects/post-quantum-cryptography/post-quantum-cryptography-standardization/round-3-submissions}, 2022, national Institute of Standards and Technology.

\bibitem{grover1996}
\BIBentryALTinterwordspacing
L.~K. Grover, ``A fast quantum mechanical algorithm for database search,'' in \emph{Proceedings of the Twenty-Eighth Annual ACM Symposium on Theory of Computing}, ser. STOC '96.\hskip 1em plus 0.5em minus 0.4em\relax New York, NY, USA: Association for Computing Machinery, 1996, p. 212–219. [Online]. Available: \url{https://doi.org/10.1145/237814.237866}
\BIBentrySTDinterwordspacing

\bibitem{Maurer1993}
\BIBentryALTinterwordspacing
U.~Maurer, ``{Secret Key Agreement by Public Discussion from Common Information},'' \emph{IEEE Trans. Inf. Theory}, vol.~39, no.~3, pp. 733--742, may 1993. [Online]. Available: \url{https://ieeexplore.ieee.org/document/256484/}
\BIBentrySTDinterwordspacing

\bibitem{Ahlswede1993}
\BIBentryALTinterwordspacing
R.~Ahlswede and I.~Csiszar, ``{Common Randomness in Information Theory and Cryptography. I. Secret Sharing},'' \emph{IEEE Trans. Inf. Theory}, vol.~39, no.~4, pp. 1121--1132, 1993. [Online]. Available: \url{http://ieeexplore.ieee.org/document/243431/}
\BIBentrySTDinterwordspacing

\bibitem{holenstein2005one}
T.~Holenstein and R.~Renner, ``One-way secret-key agreement and applications to circuit polarization and immunization of public-key encryption,'' in \emph{Annu. Int. Cryptol. Conf.}\hskip 1em plus 0.5em minus 0.4em\relax Springer, 2005, pp. 478--493.

\bibitem{holenstein2006strengthening}
T.~Holenstein, ``{Strengthening Key Agreement using Hard-core Sets},'' Ph.D. dissertation, ETH Zurich, 2006.

\bibitem{renes2013efficient}
\BIBentryALTinterwordspacing
J.~M. Renes, R.~Renner, and D.~Sutter, ``{Efficient One-Way Secret-Key Agreement and Private Channel Coding via Polarization},'' in \emph{Int. Conf. Theory Appl. Cryptol. Inf. Secur.}, ser. LNCS, K.~Sako and P.~Sarkar, Eds.\hskip 1em plus 0.5em minus 0.4em\relax Springer, 2013, vol. 8269, pp. 194--213. [Online]. Available: \url{http://link.springer.com/10.1007/978-3-642-42033-7{\_}11}
\BIBentrySTDinterwordspacing

\bibitem{Chou2015a}
\BIBentryALTinterwordspacing
R.~A. Chou, M.~R. Bloch, and E.~Abbe, ``{Polar Coding for Secret-Key Generation},'' \emph{IEEE Trans. Inf. Theory}, vol.~61, no.~11, pp. 6213--6237, nov 2015. [Online]. Available: \url{http://ieeexplore.ieee.org/document/7217814/}
\BIBentrySTDinterwordspacing

\bibitem{sharif2020}
\BIBentryALTinterwordspacing
S.~Sharifian, A.~Poostindouz, and R.~Safavi{-}Naini, ``A capacity-achieving one-way key agreement with improved finite blocklength analysis,'' in \emph{Int. Symp. on Inf. Theory and Its Appl., {ISITA} 2020}.\hskip 1em plus 0.5em minus 0.4em\relax {IEEE}, 2020, pp. 407--411. [Online]. Available: \url{https://ieeexplore.ieee.org/document/9366148}
\BIBentrySTDinterwordspacing

\bibitem{eurocryptDodisRS04}
\BIBentryALTinterwordspacing
Y.~Dodis, L.~Reyzin, and A.~D. Smith, ``Fuzzy extractors: How to generate strong keys from biometrics and other noisy data,'' in \emph{Advances in Cryptol. - {EUROCRYPT} 2004,}, ser. LNCS, C.~Cachin and J.~Camenisch, Eds., vol. 3027.\hskip 1em plus 0.5em minus 0.4em\relax Springer, 2004, pp. 523--540. [Online]. Available: \url{https://doi.org/10.1007/978-3-540-24676-3\_31}
\BIBentrySTDinterwordspacing

\bibitem{DodisORS08}
\BIBentryALTinterwordspacing
Y.~Dodis, R.~Ostrovsky, L.~Reyzin, and A.~D. Smith, ``Fuzzy extractors: How to generate strong keys from biometrics and other noisy data,'' \emph{{SIAM} J. Comput.}, vol.~38, no.~1, pp. 97--139, 2008. [Online]. Available: \url{https://doi.org/10.1137/060651380}
\BIBentrySTDinterwordspacing

\bibitem{boyen2004reusable}
X.~Boyen, ``{Reusable cryptographic fuzzy extractors},'' in \emph{Proceedings of the 11th ACM Conf. on Computer and communications security}, 2004, pp. 82--91.

\bibitem{boyen2005secure}
X.~Boyen, Y.~Dodis, J.~Katz, R.~Ostrovsky, and A.~Smith, ``{Secure remote authentication using biometric data},'' in \emph{Annual Int. Conf. Theory Appl. Cryptographic Techniques}.\hskip 1em plus 0.5em minus 0.4em\relax Springer, 2005, pp. 147--163.

\bibitem{dodis2006robust}
\BIBentryALTinterwordspacing
Y.~Dodis, B.~Kanukurthi, J.~Katz, L.~Reyzin, and A.~D. Smith, ``Robust fuzzy extractors and authenticated key agreement from close secrets,'' \emph{{IEEE} Trans. Inf. Theory}, vol.~58, no.~9, pp. 6207--6222, 2012. [Online]. Available: \url{https://doi.org/10.1109/TIT.2012.2200290}
\BIBentrySTDinterwordspacing

\bibitem{canetti2016reusable}
R.~Canetti, B.~Fuller, O.~Paneth, L.~Reyzin, and A.~Smith, ``{Reusable fuzzy extractors for low-entropy distributions},'' in \emph{Annual Int. Conf. Theory Appl. Cryptographic Techniques}.\hskip 1em plus 0.5em minus 0.4em\relax Springer, 2016, pp. 117--146.

\bibitem{BB84}
C.~H. Bennett and G.~Brassard, ``Quantum cryptography: Public key distribution and coin tossing,'' in \emph{Proceedings of IEEE International Conference on Computers, Systems and Signal Processing}, 1984, pp. 175--179.

\bibitem{Maurer1997authencation}
U.~Maurer, ``Information-theoretically secure secret-key agreement by not authenticated public discussion,'' in \emph{EUROCRYPT '97}, 1997, pp. 209--225.

\bibitem{maurer2003authen2}
\BIBentryALTinterwordspacing
U.~Maurer and S.~Wolf, ``{Secret-Key Agreement over Unauthenticated Public Channels-Part II: The Simulatability Condition},'' \emph{IEEE Trans. Inf. Theory}, vol.~49, no.~4, pp. 832--838, apr 2003. [Online]. Available: \url{http://ieeexplore.ieee.org/document/1193794/}
\BIBentrySTDinterwordspacing

\bibitem{ska2023}
S.~Panja, S.~Jiang, and R.~Safavi-Naini, ``A one-way secret key agreement with security against active adversaries,'' in \emph{2023 IEEE International Symposium on Information Theory (ISIT)}, 2023, pp. 2314--2319.

\bibitem{giacon2018kem}
\BIBentryALTinterwordspacing
F.~Giacon, F.~Heuer, and B.~Poettering, ``{KEM Combiners},'' in \emph{IACR Int. Work. Public Key Cryptogr.}\hskip 1em plus 0.5em minus 0.4em\relax Springer, 2018, pp. 190--218. [Online]. Available: \url{http://link.springer.com/10.1007/978-3-319-76578-5{\_}7}
\BIBentrySTDinterwordspacing

\bibitem{bindel2019hybrid}
N.~Bindel, J.~Brendel, M.~Fischlin, B.~Goncalves, and D.~Stebila, ``{Hybrid key encapsulation mechanisms and authenticated key exchange},'' in \emph{Int. Conf. Post-Quantum Cryptogr.}\hskip 1em plus 0.5em minus 0.4em\relax Springer, 2019, pp. 206--226.

\bibitem{sharifian2021information}
S.~Sharifian and R.~Safavi-Naini, ``Information-theoretic key encapsulation and its application to secure communication,'' in \emph{2021 IEEE Int. Symp. on Inf. Theory (ISIT)}.\hskip 1em plus 0.5em minus 0.4em\relax IEEE, 2021, pp. 2393--2398.

\bibitem{dent2003designer}
\BIBentryALTinterwordspacing
A.~W. Dent, ``{A Designer's Guide to KEMs},'' in \emph{IMA Int. Conf. on Cryptogr. and Coding}.\hskip 1em plus 0.5em minus 0.4em\relax Springer, 2003, pp. 133--151. [Online]. Available: \url{http://link.springer.com/10.1007/978-3-540-40974-8{\_}12}
\BIBentrySTDinterwordspacing

\bibitem{bentahar2008generic}
\BIBentryALTinterwordspacing
K.~Bentahar, P.~Farshim, J.~Malone-Lee, and N.~P. Smart, ``{Generic Constructions of Identity-Based and Certificateless KEMs},'' \emph{J. Cryptol.}, vol.~21, no.~2, pp. 178--199, apr 2008. [Online]. Available: \url{http://link.springer.com/10.1007/s00145-007-9000-z}
\BIBentrySTDinterwordspacing

\bibitem{haralambiev2010simple}
\BIBentryALTinterwordspacing
K.~Haralambiev, T.~Jager, E.~Kiltz, and V.~Shoup, ``{Simple and Efficient Public-Key Encryption from Computational Diffie-Hellman in the Standard Model},'' in \emph{IACR Int. Work. Public Key Cryptogr.}\hskip 1em plus 0.5em minus 0.4em\relax Springer, 2010, pp. 1--18. [Online]. Available: \url{http://link.springer.com/10.1007/978-3-642-13013-7{\_}1}
\BIBentrySTDinterwordspacing

\bibitem{ding2012simple}
J.~Ding, X.~Xie, and X.~Lin, ``{A Simple Provably Secure Key Exchange Scheme Based on the Learning with Errors Problem.}'' \emph{IACR Cryptol. ePrint Archive}, vol. 2012, p. 688, 2012.

\bibitem{peikert2014lattice}
\BIBentryALTinterwordspacing
C.~Peikert, ``{Lattice Cryptography for the Internet},'' in \emph{Int. Work. Post-Quantum Cryptogr.}\hskip 1em plus 0.5em minus 0.4em\relax Springer, 2014, pp. 197--219. [Online]. Available: \url{http://link.springer.com/10.1007/978-3-319-11659-4{\_}12}
\BIBentrySTDinterwordspacing

\bibitem{nistpqround2}
M.~Albrecht, C.~Cid, K.~Paterson, C.~Tjhai, and M.~Tomlinson, ``Nts-kem — round 2 submission,'' \url{https://csrc.nist.gov/CSRC/media/Presentations/nts-kem-round-2-presentation/images-media/nts-kem.pdf}, 2019, national Institute of Standards and Technology.

\bibitem{matsuda2018new}
T.~Matsuda and J.~C.~N. Schuldt, ``{A New Key Encapsulation Combiner},'' in \emph{2018 Int. Symp. Inf. Theory Its Appl. (ISITA)}.\hskip 1em plus 0.5em minus 0.4em\relax IEEE, 2018, pp. 698--702.

\bibitem{harnik2005robust}
\BIBentryALTinterwordspacing
D.~Harnik, J.~Kilian, M.~Naor, O.~Reingold, and A.~Rosen, ``{On Robust Combiners for Oblivious Transfer and Other Primitives},'' in \emph{Annu. Int. Conf. Theory Appl. Cryptographic Techniques}.\hskip 1em plus 0.5em minus 0.4em\relax Springer, 2005, pp. 96--113. [Online]. Available: \url{http://link.springer.com/10.1007/11426639{\_}6}
\BIBentrySTDinterwordspacing

\bibitem{bennett1988privacy}
\BIBentryALTinterwordspacing
C.~H. Bennett, G.~Brassard, and J.-M. Robert, ``{Privacy Amplification by Public Discussion},'' \emph{SIAM J. Comput.}, vol.~17, no.~2, pp. 210--229, apr 1988. [Online]. Available: \url{http://epubs.siam.org/doi/10.1137/0217014}
\BIBentrySTDinterwordspacing

\bibitem{renner2004smooth}
\BIBentryALTinterwordspacing
R.~Renner and S.~Wolf, ``{Smooth Renyi Entropy and Applications},'' in \emph{2004 IEEE Int. Symp. Inf. Theory (ISIT).}, IEEE.\hskip 1em plus 0.5em minus 0.4em\relax IEEE, 2004, pp. 232--232. [Online]. Available: \url{http://ieeexplore.ieee.org/document/1365269/}
\BIBentrySTDinterwordspacing

\bibitem{Holenstein2011}
\BIBentryALTinterwordspacing
T.~Holenstein and R.~Renner, ``{On the Randomness of Independent Experiments},'' \emph{IEEE Trans. Inf. Theory}, vol.~57, no.~4, pp. 1865--1871, apr 2011. [Online]. Available: \url{http://ieeexplore.ieee.org/document/5730579/}
\BIBentrySTDinterwordspacing

\bibitem{tomami2014}
M.~Tomamichel, J.~Martinez-Mateo, C.~Pacher, and D.~Elkouss, ``Fundamental finite key limits for information reconciliation in quantum key distribution,'' in \emph{2014 IEEE Int. Symp. on Inf. Theory}, 2014, pp. 1469--1473.

\bibitem{maurer2003authen1}
\BIBentryALTinterwordspacing
U.~Maurer and S.~Wolf, ``{Secret-Key Agreement over Unauthenticated Public Channels-Part I: Definitions and a Completeness Result},'' \emph{IEEE Trans. Inf. Theory}, vol.~49, no.~4, pp. 822--831, apr 2003. [Online]. Available: \url{http://ieeexplore.ieee.org/document/1193793/}
\BIBentrySTDinterwordspacing

\bibitem{Renner2004exact}
R.~Renner and S.~Wolf, ``The exact price for unconditionally secure asymmetric cryptography,'' in \emph{Adv. Cryptol. - EUROCRYPT 2004}, C.~Cachin and J.~L. Camenisch, Eds.\hskip 1em plus 0.5em minus 0.4em\relax Berlin, Heidelberg: Springer Berlin Heidelberg, 2004, pp. 109--125.

\bibitem{kanukurthi2009key}
\BIBentryALTinterwordspacing
B.~Kanukurthi and L.~Reyzin, ``{Key Agreement from Close Secrets over Unsecured Channels},'' in \emph{Annu. Int. Conf. Theory Appl. Cryptographic Techniques}.\hskip 1em plus 0.5em minus 0.4em\relax Springer, 2009, pp. 206--223. [Online]. Available: \url{http://link.springer.com/10.1007/978-3-642-01001-9{\_}12}
\BIBentrySTDinterwordspacing

\bibitem{shannon1949communication}
\BIBentryALTinterwordspacing
C.~E. Shannon, ``{Communication Theory of Secrecy Systems*},'' \emph{Bell System Technical Journal}, vol.~28, no.~4, pp. 656--715, oct 1949. [Online]. Available: \url{http://ieeexplore.ieee.org/lpdocs/epic03/wrapper.htm?arnumber=6769090}
\BIBentrySTDinterwordspacing

\bibitem{even1985power}
\BIBentryALTinterwordspacing
S.~Even and O.~Goldreich, ``{On the power of cascade ciphers},'' in \emph{Adv. Cryptol.}\hskip 1em plus 0.5em minus 0.4em\relax Boston, MA: Springer US, 1985, vol.~3, pp. 43--50. [Online]. Available: \url{http://link.springer.com/10.1007/978-1-4684-4730-9{\_}4}
\BIBentrySTDinterwordspacing

\bibitem{maurer1993cascade}
\BIBentryALTinterwordspacing
U.~M. Maurer and J.~L. Massey, ``{Cascade Ciphers: The Importance of Being First},'' \emph{J. Cryptol.}, vol.~6, no.~1, pp. 55--61, mar 1993. [Online]. Available: \url{http://link.springer.com/10.1007/BF02620231}
\BIBentrySTDinterwordspacing

\bibitem{Fischlin2007}
M.~Fischlin and A.~Lehmann, ``Security-amplifying combiners for collision-resistant hash functions,'' in \emph{Advances in Cryptology - CRYPTO 2007}, A.~Menezes, Ed.\hskip 1em plus 0.5em minus 0.4em\relax Berlin, Heidelberg: Springer Berlin Heidelberg, 2007, pp. 224--243.

\bibitem{Herzberg20022007}
\BIBentryALTinterwordspacing
A.~Herzberg, ``Folklore, practice and theory of robust combiners,'' Cryptology ePrint Archive, Paper 2002/135, 2002, \url{https://eprint.iacr.org/2002/135}. [Online]. Available: \url{https://eprint.iacr.org/2002/135}
\BIBentrySTDinterwordspacing

\bibitem{beaver1995precomputing}
D.~Beaver, ``Precomputing oblivious transfer,'' in \emph{Annu. Int. Cryptol. Conf.}\hskip 1em plus 0.5em minus 0.4em\relax Springer, 1995, pp. 97--109.

\bibitem{bendlin2011semi}
\BIBentryALTinterwordspacing
R.~Bendlin, I.~Damg{\aa}rd, C.~Orlandi, and S.~Zakarias, ``{Semi-homomorphic Encryption and Multiparty Computation},'' in \emph{Annu. Int. Conf. Theory Appl. Cryptographic Techniques}.\hskip 1em plus 0.5em minus 0.4em\relax Springer, 2011, pp. 169--188. [Online]. Available: \url{http://link.springer.com/10.1007/978-3-642-20465-4{\_}11}
\BIBentrySTDinterwordspacing

\bibitem{ishai2013power}
Y.~Ishai, E.~Kushilevitz, S.~Meldgaard, C.~Orlandi, and A.~Paskin-Cherniavsky, ``On the power of correlated randomness in secure computation,'' in \emph{Theory of Cryptogr. Conf.}\hskip 1em plus 0.5em minus 0.4em\relax Springer, 2013, pp. 600--620.

\bibitem{garg2018two}
S.~Garg, Y.~Ishai, and A.~Srinivasan, ``Two-round mpc: information-theoretic and black-box,'' in \emph{Theory of Cryptogr. Conf.}\hskip 1em plus 0.5em minus 0.4em\relax Springer, 2018, pp. 123--151.

\bibitem{pfitzmann2000model}
B.~Pfitzmann and M.~Waidner, ``{A model for asynchronous reactive systems and its application to secure message transmission},'' in \emph{Proc. 2001 IEEE Symp. Secur. Privacy. S{\&}P 2001}.\hskip 1em plus 0.5em minus 0.4em\relax IEEE, 2000, pp. 184--200.

\bibitem{impagliazzo1989pseudo}
\BIBentryALTinterwordspacing
R.~Impagliazzo, L.~A. Levin, and M.~Luby, ``{Pseudo-Random Generation from One-Way Functions},'' in \emph{Proc. 21st Annu. ACM Symp. Theory Comput. -STOC '89}.\hskip 1em plus 0.5em minus 0.4em\relax New York, New York, USA: ACM Press, 1989, pp. 12--24. [Online]. Available: \url{http://portal.acm.org/citation.cfm?doid=73007.73009}
\BIBentrySTDinterwordspacing

\bibitem{cramer2007bounded}
\BIBentryALTinterwordspacing
R.~Cramer, G.~Hanaoka, D.~Hofheinz, H.~Imai, E.~Kiltz, R.~Pass, A.~Shelat, and V.~Vaikuntanathan, ``{Bounded CCA2-Secure Encryption},'' in \emph{Int. Conf. Theory Appl. Cryptol. Inf. Secur.}\hskip 1em plus 0.5em minus 0.4em\relax Springer, 2007, pp. 502--518. [Online]. Available: \url{http://link.springer.com/10.1007/978-3-540-76900-2{\_}31}
\BIBentrySTDinterwordspacing

\bibitem{bellare2000authenticated}
M.~Bellare and C.~Namprempre, ``{Authenticated encryption: Relations among notions and analysis of the generic composition paradigm},'' in \emph{Int. Conf. Theory Appl. Cryptol. Inf. Secur.}\hskip 1em plus 0.5em minus 0.4em\relax Springer, 2000, pp. 531--545.

\bibitem{katz2006characterization}
J.~Katz and M.~Yung, ``{Characterization of security notions for probabilistic private-key encryption},'' \emph{J. Cryptol.}, vol.~19, no.~1, pp. 67--95, 2006.

\bibitem{Holenstein11}
T.~Holenstein and R.~Renner, ``On the randomness of independent experiments,'' \emph{IEEE Trans. Inf. Theor.}, vol.~57, no.~4, pp. 1865--1871, 2011.

\bibitem{Bellare1994}
M.~Bellare and P.~Rogaway, ``Optimal asymmetric encryption,'' in \emph{Advances in Cryptology --- EUROCRYPT'94}, A.~De~Santis, Ed.\hskip 1em plus 0.5em minus 0.4em\relax Berlin, Heidelberg: Springer Berlin Heidelberg, 1995, pp. 92--111.

\bibitem{BellarePalacio2004}
M.~Bellare and A.~Palacio, ``Towards plaintext-aware public-key encryption without random oracles,'' in \emph{Advances in Cryptology - ASIACRYPT 2004}, P.~J. Lee, Ed.\hskip 1em plus 0.5em minus 0.4em\relax Berlin, Heidelberg: Springer Berlin Heidelberg, 2004, pp. 48--62.

\bibitem{Coolidge1959}
J.~L. Coolidge, ``A treatise on algebraic plane curves,'' New York: Dover, 1959, p.~10.

\bibitem{BEZOUTtheorem}
\BIBentryALTinterwordspacing
E.~W. Weisstein, ``B\'{e}zout's theorem,'' From MathWorld--A Wolfram Web Resource, \url{https://mathworld.wolfram.com/BezoutsTheorem.html}. [Online]. Available: \url{https://mathworld.wolfram.com/BezoutsTheorem.html}
\BIBentrySTDinterwordspacing

\bibitem{TRNG}
\BIBentryALTinterwordspacing
B.~Barak, R.~Shaltiel, and E.~Tromer, ``{True Random Number Generators Secure in a Changing Environment},'' in \emph{Int. Workshop on Cryptographic Hardware and Embedded Systems}.\hskip 1em plus 0.5em minus 0.4em\relax Springer, 2003, pp. 166--180. [Online]. Available: \url{http://link.springer.com/10.1007/978-3-540-45238-6{\_}14}
\BIBentrySTDinterwordspacing

\bibitem{Bellare2005}
\BIBentryALTinterwordspacing
M.~Bellare and P.~Rogaway, ``Introduction to modern cryptography,'' 2005, \url{https://web.cs.ucdavis.edu/~rogaway/classes/227/spring05/book/main.pdf}. [Online]. Available: \url{https://web.cs.ucdavis.edu/~rogaway/classes/227/spring05/book/main.pdf}
\BIBentrySTDinterwordspacing

\end{thebibliography}

\normalsize
\appendices


\section{{Proof of Theorem~\ref{theo:comp}}}\label{pf:theoremcca}
$1.$ To prove the first part of the theorem, we define two consecutive games: the first game $\mathrm{G}^{0\text{-}b}_{\pk,\msa}$ 
is the CCA distinguishing game $\pind_{\pk,\msa}^{cca\text{-}b}(\lambda)$ in Fig \ref{fig1:thm1} and $\mathrm{G}^{1\text{-}b}_{\pk,\msa}$ is the same game except for its decapsulation oracle that always outputs $\perp$. We have:

{
%
\begin{eqnarray}\nonumber
Adv^{pkind\text{-}cca}_{\pk,\msa}(\lambda) 
&&= |\pr[\pind_{\pk,\msa}^{cca\text{-}0}(\lambda)= 1]-\pr[\pind_{\pk,\msa}^{cca\text{-}1}(\lambda)=1]|\\\label{eq1:thm1}
    &&= |\pr[\mathrm{G}^{0\textbf{-}0}_{\pk,\msa}=1]-\pr[\mathrm{G}^{0\textbf{-}1}_{\pk,\msa}=1]|\\\nonumber
    &&=  |\pr[\mathrm{G}^{0\textbf{-}0}_{\pk,\msa}=1] - \pr[\mathrm{G}^{1\textbf{-}0}_{\pk,\msa}=1] + \pr[\mathrm{G}^{1\textbf{-}0}_{\pk,\msa}=1] \\\nonumber
    &&\quad
    - \pr[\mathrm{G}^{1\textbf{-}1}_{\pk,\msa}=1] + \pr[\mathrm{G}^{1\textbf{-}1}_{\pk,\msa}=1] -  \pr[\mathrm{G}^{0\textbf{-}1}_{\pk,\msa}=1]| \\\nonumber
&&\leq |\pr[\mathrm{G}^{0\textbf{-}0}_{\pk,\msa} = 1] - \pr[\mathrm{G}^{1\textbf{-}0}_{\pk,\msa} = 1]|  
+ |\pr[\mathrm{G}^{1\textbf{-}0}_{\pk,\msa} = 1] - \pr[\mathrm{G}^{1\textbf{-}1}_{\pk,\msa} = 1]|  \\\label{eq3:thm1}
&&\quad + |\pr[\mathrm{G}^{1\textbf{-}1}_{\pk,\msa} = 1] - \pr[\mathrm{G}^{0\textbf{-}1}_{\pk,\msa} = 1]|
\end{eqnarray}
}


\remove{
{\footnotesize
%
\begin{align}\nonumber
Adv&^{pkind\text{-}cca}_{\pk,\msa}(\lambda) \\\nonumber
&= |\pr[\pind_{\pk,\msa}^{cca\text{-}0}(\lambda)= 1]-\pr[\pind_{\pk,\msa}^{cca\text{-}1}(\lambda)=1]|\\\label{eq1:thm1}
    &= |\pr[\mathrm{G}^{0\textbf{-}0}_{\pk,\msa}=1]-\pr[\mathrm{G}^{0\textbf{-}1}_{\pk,\msa}=1]|\\\nonumber
    &=  |\pr[\mathrm{G}^{0\textbf{-}0}_{\pk,\msa}=1] - \pr[\mathrm{G}^{1\textbf{-}0}_{\pk,\msa}=1] + \pr[\mathrm{G}^{1\textbf{-}0}_{\pk,\msa}=1] \\\label{eq:comp2}
    &\qquad - \pr[\mathrm{G}^{1\textbf{-}1}_{\pk,\msa}=1] + \pr[\mathrm{G}^{1\textbf{-}1}_{\pk,\msa}=1] -  \pr[\mathrm{G}^{0\textbf{-}1}_{\pk,\msa}=1]| \\\nonumber
\end{align}
\begin{equation}
\begin{aligned}
\label{eq3:thm1}
\leq |\pr[\mathrm{G}^{0\textbf{-}0}_{\pk,\msa} = 1] - \pr[\mathrm{G}^{1\textbf{-}0}_{\pk,\msa} = 1]|\\
+ |\pr[\mathrm{G}^{1\textbf{-}0}_{\pk,\msa} = 1] - \pr[\mathrm{G}^{1\textbf{-}1}_{\pk,\msa} = 1]|  \\ 
+ |\pr[\mathrm{G}^{1\textbf{-}1}_{\pk,\msa} = 1] - \pr[\mathrm{G}^{0\textbf{-}1}_{\pk,\msa} = 1]|
\end{aligned}
\end{equation}
}
}
where equation~\ref{eq1:thm1} is simply using $\mathrm{G}^{0\textbf{-}0}_{\pk,\msa}$ and $\mathrm{G}^{0\textbf{-}1}_{\pk,\msa}$ in lieu of $\pind_{\pk,\msa}^{cca\text{-}0}$ and $\pind_{\pk,\msa}^{cca\text{-}1}$ respectively, and inequality~\ref{eq3:thm1} is by  triangle  inequality.

To bound the first and the last terms of inequality~\ref{eq3:thm1}, let $U_1$ be the event that $\msa$ outputs 1 in game $\mathrm{G}^{0\text{-}b}_{\pk,\msa}$ and $U_2$ be the event that $\msa$ outputs 1 in game $\mathrm{G}^{1\text{-}b}_{\pk,D}$ for $b \in \{0,1\}$. These two games are identical except when the decapsulation oracle output is not $\perp$ in $\mathrm{G}^{0\text{-}b}_{\pk,D}$, lets call this event F. The event F is the union of $q_d$ events $\cup^{q_d}_{i=1}F_i$, where $F_i$ is the event that the output of the decapsulation oracle in the $i$-th call is not $\perp$. Let A be the adversary in game $\mathrm{KINT}_{\pk,A}$ that
makes $q_d$ queries to the decapsulation oracle. We have $\pr[F_i] \leq Adv^{kint}_{\pk,A}(\lambda)$ and therefore, from the union bound $\pr[F] \leq \sum^{q_d}_{i=1} \pr[F_i] \leq q_{d}\pr[\mathrm{KINT}_{\pk,A} = 1] =
q_{d}Adv^{kint}_{\pk,A}(\lambda)$, and from Lemma $6.2$ of \cite{Cramer2003DesignAA}:
{ 
\begin{align}
\label{eq4:thm1}
    &|\pr[\mathrm{G}^{0\text{-}b}_{\pk,D} = 1] - \pr[\mathrm{G}^{1\text{-}b}_{\pk,D} = 1]| = |\pr[U_1] - \pr[U_2]| 
    \leq \pr[F] 
    \leq q_{d}Adv^{kint}_{\pk,A}(\lambda)
\end{align}
}
To bound the second term in inequality~\ref{eq3:thm1}, we note that in $\mathrm{G}^{1\text{-}b}_{\pk,D}$ for $b \in \{0,1\}$, the decapsulation oracle always output $\perp$ and simulates the IND-CEA game $\pind^{cea\text{-}b}_{\pk,B}(\lambda)$. Therefore,
\begin{equation}
\label{eq5:thm1}
|\pr[\mathrm{G}^{1\text{-}0}_{\pk,D} = 1] - \pr[\mathrm{G}^{1\text{-}1}_{\pk,D} = 1]| \leq Adv^{pkind\text{-}cea}_{\pk,B} (\lambda)
\end{equation}
Finally, from inequalities~\ref{eq4:thm1} and~\ref{eq5:thm1} we have
\[Adv^{pkind\text{-}cca}_{\pk,D} (\lambda) \leq 2q_{d}Adv^{kint}_{\pk,A}(\lambda) + Adv^{pkind\text{-}cea}_{\pk,B} (\lambda)\].

$2.$ The proof of the second part of the theorem uses the same sequence of games, but against a computationally unbounded adversary. We can similarly bound the CCA advantage of the adversary by bounding the advantage of these games.
$\qed$

\section{{Proof of Theorem~\ref{theo:composition}}}\label{pf:hecomp}
We first show the claim of the theorem for the second case that is, an IND-CCA secure KEM in preprocessing model and an IND-OTCCA secure DEM construct an IND-CCA secure hybrid encryption scheme in preprocessing model. The proof of the first case will follow from the proof of the second case. The proof of the third and forth cases are identical to the proof of first and second cases respectively and noting that the adversary for the iKEM is query-bounded and computationally unbounded.

We define a sequence of three games $\mathrm{G}^{0\text{-}b}$, $\mathrm{G}^{1\text{-}b}$, and $\mathrm{G}^{2\text{-}b}$ that simulate adversary's actual or modified interaction with the encryption system during the attack procedure. Each game operates on the same underlying probability space. In particular, private inputs of parties, randomness of the adversary's algorithm, and the hidden bit b take on identical values across all games. At the end of each game, the adversary outputs a bit $\hat{b}$. For a game $\mathrm{G}^{i\text{-}b}$, where $i \in \{0,1,2\}$
with output $\hat{b}$, $T_i$ denotes the event that $\hat{b} = b$. All games are played by a computationally bounded distinguisher D. $\mathrm{G}^{0\text{-}b}_D$
 is identical to the distinguishing game of hybrid encryption in preprocessing model explained above. $\mathrm{G}^{1\text{-}b}_D$  only differs from $\mathrm{G}^{0\text{-}b}_D$ in its decapsulation oracle. Suppose the challenge HE ciphertext
$c^* = (c_1^*, c_2^*)$, where $c_{1}^*$ is generated by $\cke$ and $c_2^*$ is generated by SE.Enc. Then for any decryption query $c = (c_1, c_2) \ne (c_1^*, c_2^*)$, the decryption oracle of $\mathrm{G}^{1\text{-}b}_D$ uses $\ckd$ to decrypt the  ciphertext unless $c_1 = c_1^*$ (and $c_2 \ne c_2^*)$.
In this case, the key $k_1$ corresponding to $c_1^*$ that is generated by $\cke$  will be used for the decryption of $c_2^*$. Finally, $\mathrm{G}^{2\text{-}b}_D$ only differs from $\mathrm{G}^{1\text{-}b}_D$ in using a uniformly sampled key instead of the key generated by cKEM for encryption and answering encryption and decryption queries. We bound $Adv^{ind\text{-}cca}_{\ike_{\ck,\dem},D} (\lambda)$
using the defined games: For a given sample sam = $(r_A, r_B, r_E)$ generated by $\ckg$, we define $BK_{sam}$, a set of bad keys k, generated by $\cke$, as $BK_{sam} = \{k: \ckd( r_B, c) \ne k\}$. According to the correctness of $\ck$, for $k \leftarrow^\$ \{0,1\}^{l(\lambda)}$ we have $\pr[k \in  BK_{sam}] \leq \epsilon$. The two events $T_0$ and $T_1$ are
only different when the event $[\ckd(r_B, c_1^*) \in BK_{sam}]$ happens. Using Lemma $6.2$ of \cite{Cramer2003DesignAA}, we have

\begin{equation}
\label{eq11:thm2}
    |\pr[T_0] - \pr[T_1]| \leq \pr[k \in BK_{sam}] \leq \epsilon
\end{equation}
We now consider the game $\mathrm{G}^{2\text{-}b}_D$ and $\mathrm{G}^{1\text{-}b}_D$. The game $\mathrm{G}^{2\text{-}b}_D$ is same as $\mathrm{G}^{1\text{-}b}_D$ except that $\mathrm{G}^{2\text{-}b}_D$ uses a uniformly sampled key instead of the key generated by cKEM for encryption and decryption queries. Since the KEM's key is $\sigma$-IND-CCA secure, there exists an adversary $D'$ such that
\begin{equation}
\label{eq12:thm2}
|\pr[T_1] - \pr[T_2]| = Adv^{kind\text{-}cca}_{\ck,D'} (\lambda) \leq \sigma
\end{equation}
In the above case, the adversary $D'$ just runs the adversary $D$. Specifically, $D'$ is playing an attack game against KEM in which $k_b$ is equal to $k^*$ in game $\mathrm{G}^{1\text{-}b}_D$, whereas $k_b$ is a uniformly sampled random value in the game $\mathrm{G}^{2\text{-}b}_D$.

Lastly, note that in game $\mathrm{G}^{2\text{-}b}_D$, a new random key is sampled for each encryption/decryption query. Thus in this game, the adversary $D$ is just executing a chosen ciphertext attack against SE. Therefore, there exists an adversary $D''$ such that    
\begin{equation}
\label{eq16:thm2}
    |\pr[T_2] - 1/2 | = \frac{1}{2}Adv^{ind\text{-}otcca}_{SE,D''} (\lambda) \leq \frac{\sigma'}{2}
\end{equation}

Since $Adv^{ind\text{-}cca}_{\ike_{\ck,\dem},D} (\lambda) = 2|\pr[T_0] - 1/2 |$, using inequalities \ref{eq11:thm2}, \ref{eq12:thm2}, \ref{eq16:thm2} we have 
{
\begin{align} \nonumber
Adv^{ind\text{-}cca}_{\ike_{\ck,\dem},D} (\lambda) &= 2|\pr[T_0] - 1/2 | \\\nonumber
&= 2|\pr[T_0] - \pr[T_1] + \pr[T_1] - \pr[T_2] + \pr[T_2] - 1/2 | \\\nonumber
&\leq 2|\pr[T_0] - \pr[T_1]| + 2|\pr[T_1] - \pr[T_2]| + 2|\pr[T_2] - 1/2| \\\nonumber
&\leq 2\epsilon + 2\sigma + \sigma'.
\end{align}
}
For the proof of the first part, we note that $\mathrm{G}^{0\text{-}b}_D$ and $\mathrm{G}^{1\text{-}b}_D$ are identical because no decryption query is issued. Therefore, $|\pr[T_1] - \pr[T_0]| = 0$. Also since there is no decryption query and the KEM is $\sigma$-IND-CEA secure we have,
\[
|\pr[T_1] - \pr[T_2]| \leq Adv^{kind\text{-}q_e\text{-}cea}_{\ck,D} (\lambda) \leq \sigma;\]
and since the DEM is $\sigma'$-IND-OT secure, we have,
\[
|\pr[T_2] - 1/2| = (Adv^{ind\text{-}ot}_{SE,D} (\lambda))/2 \leq \sigma'/2\]
and finally,
\[
Adv^{ind\text{-}cea}_{\ike_{\ck,\dem},D} (\lambda) = 2|\pr[T_0] - 1/2| \leq 2\sigma + \sigma'. 
\]
$\qed$

\section{{Proof of Theorem~\ref{theo:xorcomb}}}\label{pf:xorcomb}
In the  $q_e$-CEA distinguishing game of  $\ikem$,  the distinguisher $\msa'$ receives $r_{E}$,  $\mathbf{v}^{q_e\text{-}cea} =( {v_1}^{cea}, \cdots, {v_{q_e}}^{cea})$, where ${v_i}^{cea}$ is the encapsulation oracle's output to the $i$th encapsulation query, and the
pair of challenge ciphertext and key $(c_1^*,k^*_{1_b})$, and is supposed to distinguish if $k^*_{1_b}$ is generated by $\ikem$ or is sampled uniformly. 
 ${\msa}'$ uses the KEM $\kem$  to generate $(pk,sk)$ and produces  $({c'}_{2}^*,{k'}_{2}^*)\stackrel{\$}\gets\ekem(pk)$. Then sends $\mb c^*$ and $k^*$ to $\msa$, where $\mathbf{c}^*=(c_1^*,{c'}_2^*)$ and ${k}^*=k_{1_b}^*\oplus{k'}_2^*$ to $\msa$.
Finally, $\msa'$ outputs $b'$ equal to $\msa$'s output. The advantages of $\msa$ and is upper bounded by $\msa'$  because $k^*$ is a sample from the uniform distribution only if $k_{1_b}$ is a sample from the uniform distribution. Since we assumed that $\msa$ breaks the IND-$q_e$-CEA security of the combined key, then $\msa'$ can break the IND-$q_e$-CEA security of the iKEM which is a contradiction. $\ \qed$

\section{{Proof of Theorem~\ref{thmprfxorcomb}}}\label{pf:thmprfxorcomb}
    The proof for a computationally bounded adversary will be based on the proof of Theorem 3 in \cite{giacon2018kem}, and noting that the iKEM will loose its security after a fixed number of repeated queries. We shall prove part $(b)$ of the theorem \ref{thmprfxorcomb}.

To prove the part $(b)$ of the Theorem \ref{thmprfxorcomb}, let $D'=(D_1,D_2)$ denote a computationally unbounded adversary attacking the CCA security of the combiner $Comb^{PtX}_{\ikem,\kem}$ by making at most $q_e$ encapsulation and $q_d$ ciphertext (decapsulation) queries in the CCA distinguishing game $\pind^{(q_e;q_d)\text{-}cca\text{-}b}_{Comb,D'}$, and $b$ be uniform over \{0,1\}.

\begin{figure}[!ht]
\centering
\begin{minipage}{0.46\textwidth}
\vspace{-8.8em}
$\pind^{(q_e;q_d)\text{-}cca\text{-}b}_{Comb,D'} := \mathrm{G}^{0\text{-}b}_{Comb,D'}$
\vspace{.4em}
\hrule 
\vspace{.4em}
1: $Ret[.] \leftarrow \perp$ \\
2: $(r_A, r_B, r_E) \xleftarrow{\$} \ikemg(1^{\lambda},P_{\X \Y \Z })$ \\ 
3: $(s_k, p_k) \xleftarrow{\$} \gkem(1^{\lambda})$ \\
4: $st_1 \xleftarrow{\$} D^{O_1}_{1}(r_E, p_k)$ \\
5: $(k^*_1, c^*_1) \xleftarrow{\$} \ikeme(r_A)$ \\
   \%  $\mathrm{G}^{1\text{-}b}_{D'}-\mathrm{G}^{3\text{-}b}_{D'}:$ $k^*_1 \xleftarrow{\$} \{0,1\}^{\ikeml(\lambda)}$ \\
6: $(k^*_2, c^*_2) \xleftarrow{\$} \ekem( p_k)$ \\
7: $c^*  \leftarrow (c^*_1, c^*_2)$ \\
8: $y^*_1  \leftarrow  F_1(k^*_1, c^*_2)$  \\
\%  $\mathrm{G}^{2\text{-}b}_{D'}-\mathrm{G}^{4\text{-}b}_{D'}$ : $y^*_1 \xleftarrow{\$} \mathcal{K}$ \\
9: $k^*  \leftarrow  y^*_1 \oplus F_2(k^*_2, c^*_1)$ \\
10: $k'_0 \leftarrow k^*; k'_1 \xleftarrow{\$} \mathcal{K}$ \\
11: $b' \xleftarrow{\$} D^{O_2}_{2}(st_1,c^*,k'_b)$ \\ 
12: Return $b'$ \\
\end{minipage}
\hfill
\begin{minipage}{0.46\textwidth}
{\bf Oracle} $\eckem(r_A,pk)$
\vspace{.4em}
\hrule
\vspace{.4em}
1: $(k_{11}, c_{11}) \xleftarrow{\$} \ikeme(r_A)$ \\
2: $(k_{21}, c_{21}) \xleftarrow{\$} \ekem(p_k)$ \\
3: $k  \leftarrow  F_1(k_{11}, c_{21}) \oplus F_2(k_{21}, c_{11})$ \\
4: Return $(k,c_{11}, c_{21})$ \\


{\bf Oracle} $\dckem(r_B,sk,c)$
\vspace{.4em}
\hrule
\vspace{.4em}
1: If $c = c^*$: Abort \\
2: If $Ret[c] \ne \perp:$ Return $Ret[c]$ \\
3: $c_1, c_2  \leftarrow c$ \\
4: If $c_1 = c^*_1$: \\
5: \hspace{.6cm} $k_1  \leftarrow k^*_1$ \\
6: \hspace{.6cm}   $y_1  \leftarrow F_1(k_1, c_2)$ \\
   \% \hspace{.6cm} $\mathrm{G}^{2\text{-}b}_{D'}:$ $ y_1  \xleftarrow{\$} \mathcal{K}$ \\
   \% \hspace{.6cm} $\mathrm{G}^{3\text{-}b}_{D'}:$  $y_1 \xleftarrow{\$}  F_1(k_1, c_2)$ \\ 
7: else \\
8: \hspace{.6cm} $k_1 \xleftarrow{\$}  \ikemd( r_B, c_1)$ \\
9: \hspace{.6cm} If $k_1 =\perp$: Return $\epsilon$ \\
10: \hspace{.6cm} $y_1  \leftarrow F_1(k_1, c_2)$ \\
11: $k_2 \leftarrow  \dkem(s_k, c_2)$ \\
12: If $k_2 =\perp$: Return $\epsilon$ \\
13: $Ret[c]  \leftarrow y_1 \oplus F_2(k_2, c_1)$ \\
14: Return $Ret[c]$ \\


\end{minipage}
\caption{Games $\mathrm{G}^{0\text{-}b}_{Comb,D'}$ to $\mathrm{G}^{4\text{-}b}_{Comb,D'}$ to prove security of the PRF-then-XOR combiner}
\label{figikemcombcca2}
\end{figure}

The proof uses a sequence of five games. We define five games $\mathrm{G}^{0\text{-}b}_{Comb,D'}$ to $\mathrm{G}^{4\text{-}b}_{Comb,D'}$ for a uniform $b$ over \{0,1\}, played by the adversary $D'=(D_1,D_2)$. Figure \ref{figikemcombcca2} depicts these games. In
each game, $D'$ outputs $b' \in \{0,1\}$. Note that, if the adversary has already queried the oracle for the same input, the oracle returns the same output.

Adversary $D'=(D_1,D_2)$ can call two oracles,  $\eckem(r_A,  pk 
)$ and $\dckem(r_B,sk,\cdot)$ that correspond to the encapsulation and decapsulation algorithms of the combiner, and have access to the associated keys of the component KEMs. 
We use $O_1$ and $O_2$ to refer to oracle calls of $D'$ before and after seeing the challenge ciphertext.  

$\mathrm{G}^{0\text{-}b}_{Comb,D'}$ is the CCA distinguishing game of the combiner $Comb^{PtX}_{\ikem,\kem}$ with the distinguisher $D'$ making at most $q_e$ encapsulation and $q_d$ decapsulation queries. That is, $\pind^{(q_e;q_d)\text{-}cca\text{-}b}_{Comb,D'}=\mathrm{G}^{0\text{-}b}_{Comb,D'}$.    Note that according to the PRF-then-XOR construction of the combiner (figure \ref{figikemcombcca2}),  the decapsulation oracle outputs ``$\perp$'' when the ciphertext of at least one of the components decapsulates to ``$\perp$''. 

\begin{equation}
\label{eqnthmcombptx1}
    \pr[\pind^{(q_e;q_d)\text{-}cca\text{-}0}_{Comb,D'}(\lambda)=1]=\pr[\mathrm{G}^{0\text{-}0}_{Comb,D'}(\lambda)=1]
\end{equation}

In $\mathrm{G}^{1\text{-}b}_{Comb,D'}$, the iKEM key $k^*_1$ is replaced by a uniform random key (this replacement is also reflected in the decapsulation oracle Line 5 using $k_1 \leftarrow  k^*_1$).

\begin{clam}
\label{claimptx1}
There exists a computationally unbounded adversary $U_1$ whose advantage in the CCA distinguishing game of iKEM $\ikem$ with at most $q_e$ encapsulation and $q_d$ decapsulation queries is $Adv^{pkind\text{-}(q_e;q_d)\text{-}cca}_{\ikem,U_1}$ such that 
\remove{
{\footnotesize
\begin{eqnarray*}
\label{eqnthmcombptx2}
   &&Adv^{pkind\text{-}(q_e;q_d)\text{-}cca}_{\ikem,U_1} \geq \\
   && \hspace{1.5cm} |\pr[\mathrm{G}^{0\text{-}0}_{Comb,D'}(\lambda)=1] - \pr[\mathrm{G}^{1\text{-}0}_{Comb,D'}(\lambda)=1]|
\end{eqnarray*}
}}
{
{
\begin{eqnarray}\label{eqnthmcombptx2}
    &&|\pr[\mathrm{G}^{0\text{-}0}_{Comb,D'}(\lambda)=1] - \pr[\mathrm{G}^{1\text{-}0}_{Comb,D'}(\lambda)=1]|  
    \leq Adv^{pkind\text{-}(q_e;q_d)\text{-}cca}_{\ikem,U_1}
\end{eqnarray}
}}
\end{clam}
\begin{proof}(claim \ref{claimptx1}) We construct the adversary $U_1=(U_{11},U_{12})$ for the CCA distinguishing game of iKEM as given in figure \ref{figikemcombcca22}. 

Adversary $U_{11}$ takes 
$r_E$ as input. The adversary $U_{12}$ runs on the challenge $(c^*_1,k^*_1)$. At the end, $U_{12}$ relays whatever $D_2$ outputs.       

\begin{figure}[!ht]
\centering
\begin{minipage}{0.46\textwidth}
\vspace{-7.6em}
{\bf Adversary} $U^{O_1}_{11}(r_E)$
 \vspace{0.4em}
\hrule
\vspace{.4em}
1: $(s_k, p_k) \xleftarrow{\$} \gkem()$ \\
2: $st_1 \xleftarrow{\$}  D^{O_1}_{1}(r_E, p_k) $ \\
3: Return $st_1$ \\
\\
 {\bf Adversary} $U^{O_2}_{12}(st_1,c^*_1,k^*_1)$
 \vspace{0.4em}
\hrule
\vspace{.4em}
1. $(k^*_2, c^*_2) \xleftarrow{\$} \ekem( p_k)$ \\
2: $c^*  \leftarrow (c^*_1, c^*_2)$ \\ 
3. $y^*_1  \leftarrow  F_1(k^*_1, c^*_2)$ \\
4: $k^*  \leftarrow  y^*_1 \oplus F_2(k^*_2, c^*_1)$ \\
5: $b' \xleftarrow{\$} D^{O_2}_{2}(st_1,c^*,k^*)$ \\
6: Return $b'$ \\

\end{minipage}
\hfill
\begin{minipage}{0.46\textwidth}
{\bf Oracle} $\eckem(r_A,pk)$
\vspace{0.4em}
\hrule
\vspace{.4em}
1: $(k_{11}, c_{11}) \xleftarrow{\$} \ikeme(r_A)$ \\
2: $(k_{21}, c_{21}) \xleftarrow{\$} \ekem(p_k)$ \\
3: $k  \leftarrow  F_1(k_{11}, c_{21}) \oplus F_2(k_{21}, c_{11})$ \\
4: Return $(k,c_{11}, c_{21})$ \\


{\bf Oracle} $\dckem(r_B,sk,c)$
\vspace{.4em}
\hrule
\vspace{.4em}
1: If $c = c^*$: Abort \\
2: $c_1, c_2  \leftarrow c$ \\
3: If $c_1 = c^*_1$: \\
4: \hspace{.6cm} $k_1  \leftarrow k^*_1$ \\
5: else \\
6: \hspace{.6cm} $k_1 \xleftarrow{\$}  \ikemd( r_B, c_1)$ \\
7: \hspace{.6cm} If $k_1 =\perp$: Return $\epsilon$ \\
8: $y_1  \leftarrow F_1(k_1, c_2)$ \\
9: $k_2 \leftarrow  \dkem(s_k, c_2)$ \\
10: If $k_2 =\perp$: Return $\epsilon$ \\
11: $k  \leftarrow y_1 \oplus F_2(k_2, c_1)$ \\
12: Return $k$ \\


\end{minipage}
\caption{Adversary   $U_1=(U_{11},U_{12})$ is in  
CCA key indistinguishing game of iKEM $\ikem$, and 
$D'=(D_1,D_2)$ is 
the adversary in CCA key indistinguishing game of the combiner}. 
\label{figikemcombcca22}
\end{figure}

In this construction, $U_1$ issues at most as many queries as $D'$.
Now if $U_1$ is run by the game $\pind^{(q_e;q_d)\text{-}cca\text{-}0}_{\ikem,U_1}$, and thus, $k^*_1$ is the actual key output of $\ikeme()$, then $U_1$ simulates the game $\mathrm{G}^{0\text{-}0}_{Comb,D'}$. On the other hand, if $U_1$ is run by the game $\pind^{(q_e;q_d)\text{-}cca\text{-}1}_{\ikem,U_1}$, that is, $k^*_1$ is uniformly sampled, then $U_1$ perfectly simulates the game $\mathrm{G}^{1\text{-}0}_{Comb,D'}$. Therefore, $\pr[\mathrm{G}^{0\text{-}0}_{Comb,D'}(\lambda)=1]=\pr[\pind^{(q_e;q_d)\text{-}cca\text{-}0}_{\ikem,U_1}(\lambda)=1]$ and  $\pr[\mathrm{G}^{1\text{-}0}_{Comb,D'}(\lambda)=1]=\pr[\pind^{(q_e;q_d)\text{-}cca\text{-}1}_{\ikem,U_1}(\lambda)=1]$.

Hence, 
\begin{align}\nonumber
&|\pr[\mathrm{G}^{0\text{-}0}_{Comb,D'}(\lambda)=1] - \pr[\mathrm{G}^{1\text{-}0}_{Comb,D'}(\lambda)=1]| 
\\\nonumber
&= |\pr[\pind^{(q_e;q_d)\text{-}cca\text{-}0}_{\ikem,U_1}(\lambda)=1] 
- \pr[\pind^{(q_e;q_d)\text{-}cca\text{-}1}_{\ikem,U_1}(\lambda)=1]| \\\nonumber
&\leq Adv^{pkind\text{-}(q_e;q_d)\text{-}cca}_{\ikem,U_1}.
\end{align}
\end{proof}

In $\mathrm{G}^{2\text{-}b}_{Comb,D'}$, the output of PRF $F_1$ is replaced by a uniform sample from the output set of the PRF (line 8). This change is also applied to the decapsulation oracle (line 6). 

\begin{clam}
\label{claimptx2}
There exists a computationally unbounded adversary $U_2$ whose advantage, after making at most $q_d+1$ $Eval$ queries, in distinguishing the output of PRF $F_1$ from a uniform sample is $Adv^{(q_d+1)\text{-}PRF}_{F_1,U_2}$ such that 
\begin{equation}
\label{eqnthmcombptx3}
    |\pr[\mathrm{G}^{1\text{-}0}_{Comb,D'}(\lambda)=1] - \pr[\mathrm{G}^{2\text{-}0}_{Comb,D'}(\lambda)=1]| \leq Adv^{(q_d + 1)\text{-}PRF}_{F_1,U_2}.
\end{equation}
\end{clam}
\begin{proof} (claim \ref{claimptx2}) We construct the adversary $U_2$ as given in figure \ref{figikemcombcca23}. From line 1 and 2 of the decapsulation oracle, we ensure that the input to $Eval$ is always different. 

\begin{figure}[!ht]
\centering
\begin{minipage}{0.46\textwidth}
\vspace{-9.5em}
{\bf Adversary} $U_2^{Eval}$
\vspace{.4em}
\hrule
\vspace{.4em}
1: $Ret[.] \leftarrow \perp$ \\
2: $(r_A, r_B, r_E) \xleftarrow{\$} \ikemg(P_{\X \Y \Z })$ \\
3: $(s_k, p_k) \xleftarrow{\$} \gkem()$ \\
4: $st_1 \xleftarrow{\$} D^{O_1}_{1}(r_E, p_k)$ \\
5: $(k^*_1, c^*_1) \xleftarrow{\$} \ikeme( r_A)$ \\
6: $(k^*_2, c^*_2) \xleftarrow{\$} \ekem(p_k)$ \\
7: $c^*  \leftarrow (c^*_1, c^*_2)$ \\
8: $y^*_1  \leftarrow  Eval(c^*_2)$  \\
9: $k^*  \leftarrow  y^*_1 \oplus F_2(k^*_2, c^*_1)$ \\
11: $b' \xleftarrow{\$} D^{O_2}_{2}(st_1,c^*,k^*)$ \\ 
12: Return $b'$ \\
\end{minipage}
\hfill
\begin{minipage}{0.46\textwidth}
{\bf Oracle} $\eckem(r_A,pk)$
\vspace{.4em}
\hrule
\vspace{.4em}
1: $(k_{11}, c_{11}) \xleftarrow{\$} \ikeme(r_A)$ \\
2: $(k_{21}, c_{21}) \xleftarrow{\$} \ekem( p_k)$ \\
3: $k  \leftarrow  F_1(k_{11}, c_{21}) \oplus F_2(k_{21}, c_{11})$ \\
4: Return $(k,c_{11}, c_{21})$ \\


{\bf Oracle} $\dckem(r_B,sk,c)$
\vspace{.4em}
\hrule
\vspace{.4em}
1: If $c = c^*$: Abort \\
2: If $Ret[c] \ne \perp:$ Return $Ret[c]$ \\
3: $c_1, c_2  \leftarrow c$ \\
4: If $c_1 = c^*_1$: \\
5: \hspace{.6cm}   $y_1  \leftarrow Eval(c_2)$ \\
::     else \\
6: \hspace{.6cm} $k_1 \xleftarrow{\$}  \ikemd(r_B, c_1)$ \\
7: \hspace{.6cm} If $k_1 =\perp$: Return $\epsilon$ \\
8: \hspace{.6cm} $y_1  \leftarrow F_1(k_1, c_2)$ \\
9: $k_2 \leftarrow  \dkem(s_k, c_2)$ \\
10: If $k_2 =\perp$: Return $\epsilon$ \\
11: $Ret[c]  \leftarrow y_1 \oplus F_2(k_2, c_1)$ \\
12: Return $Ret[c]$ \\


\end{minipage}
\caption{Adversary $U_2$ against distinguishing output of the PRF $F_1$ from a uniform sample.   Adversary $D'=(D_1,D_2)$ is 
the adversary in CCA key indistinguishing game of the combiner. }
\label{figikemcombcca23}
\end{figure}

From the construction of the adversary $U_2$, we observe that $Eval$ is called only once by $U_2$ during generation of the challenge. In addition, for each query to the decapsulation oracle by $D'$, $Eval$ is called at most once by $U_2$. Hence, $U_2$ queries $Eval$ at most $(q_d+1)$ times. Now when $U_2$ is run by the game $\mathrm{PRI}^{(q_d+1)\text{-}IND\text{-}0}_{F_1,U_2}$, $k^*_1$ is the key  generated by the game $\mathrm{PRI}^{(q_d+1)\text{-}IND\text{-}0}_{F_1,U_2}$.
Thus, $U_2$  emulates the game $\mathrm{G}^{1\text{-}0}_{Comb,D'}$. On the other hand, when $U_2$ is run by the game $\mathrm{PRI}^{(q_d+1)\text{-}IND\text{-}1}_{F_1,U_2}$, $Eval$ outputs uniformly sampled value, that is, $y^*_1$ in line 8 of $U_2^{Eval}$ algorithm (and line 5 of decapsulation Oracle queries) is uniformly generated. Hence $U_2$ perfectly simulates the game $\mathrm{G}^{2\text{-}0}_{Comb,D'}$. Therefore, 
\[
\pr[\mathrm{G}^{1\text{-}0}_{Comb,D'}(\lambda)=1]=\pr[\mathrm{PRI}^{(q_d+1)\text{-}IND\text{-}0}_{F_1,U_2}(\lambda)=1]\]
and \[\pr[\mathrm{G}^{2\text{-}0}_{Comb,D'}(\lambda)=1]=\pr[\mathrm{PRI}^{(q_d+1)\text{-}IND\text{-}1}_{F_1,U_2}(\lambda)=1].
\]

Thus, 

{
\begin{eqnarray*}
&&|\pr[\mathrm{G}^{1\text{-}0}_{Comb,D'}(\lambda)=1] - \pr[\mathrm{G}^{2\text{-}0}_{Comb,D'}(\lambda)=1]| \\\nonumber
&& =|\pr[\mathrm{PRI}^{(q_d+1)\text{-}IND\text{-}0}_{F_1,U_2}(\lambda)=1] 
- \pr[\mathrm{PRI}^{(q_d+1)\text{-}IND\text{-}1}_{F_1,U_2}(\lambda)=1]| \\\nonumber
&& 
\leq Adv^{(q_d + 1)\text{-}PRF}_{F_1,U_2}. 
\end{eqnarray*}}

\remove{\begin{align}\nonumber
&|\pr[\mathrm{G}^{1\text{-}0}_{Comb,D'}(\lambda)=1] - \pr[\mathrm{G}^{2\text{-}0}_{Comb,D'}(\lambda)=1]| \\\nonumber
&= |\pr[\mathrm{PRI}^{(q_d+1)\text{-}IND\text{-}0}_{F_1,U_2}(\lambda)=1] \\\nonumber
&\qquad - \pr[\mathrm{PRI}^{(q_d+1)\text{-}IND\text{-}1}_{F_1,U_2}(\lambda)=1]| \\\nonumber
&\leq Adv^{(q_d + 1)\text{-}PRF}_{F_1,U_2}. 
\end{align}}

\end{proof}

In $\mathrm{G}^{3\text{-}b}_{Comb,D'}$, we reverse the modifications of the decapsulation oracle that we introduced in game $\mathrm{G}^{2\text{-}0}_{Comb,D'}$.   Consequently, if an adversary queries its decapsulation oracle on a ciphertext $c$ whose first component is $c_1$, the oracle computes $y_1$ by invoking the function $F_1$ instead of returning a uniformly random value. Then, there exists an adversary $U'_2$ whose advantage in distinguishing the output of PRF $F_1$ from a uniform sample is $Adv^{q_d\text{-}PRF}_{F_1,U'_2}$ such that,
\begin{eqnarray}
\label{eqnthmcombptx4}
    &&|\pr[\mathrm{G}^{2\text{-}0}_{Comb,D'}(\lambda)=1] - \pr[\mathrm{G}^{3\text{-}0}_{Comb,D'}(\lambda)=1]| 
    \leq Adv^{q_d\text{-}PRF}_{F_1,U'_2},
\end{eqnarray}
and $U'_2$ issues at most $q_d$ $Eval$ queries. We can construct such adversary $U'_2$ by replacing line 8 of the adversary $U_2$ in figure \ref{figikemcombcca23} with uniform value ($y^*_1  \leftarrow  \mathcal{K}$). The proof is same as claim \ref{claimptx2}. In this case, as $y^*_1$ is uniform, $U'_2$ calls $Eval$ at most $q_d$ times.

In $\mathrm{G}^{4\text{-}b}_{Comb,D'}$, we reverse the modifications added in the game $\mathrm{G}^{1\text{-}b}_{Comb,D'}$ by replacing the uniform key $k^*_1$ in line 5 of the game in figure \ref{figikemcombcca2} with an actual key output of $\ikeme()$. Then, there exists a computationally unbounded adversary $U'_1$ whose advantage in the CCA distinguishing game of iKEM $\ikem$ with $q_e$ encapsulation and $q_d$ decapsulation queries is $Adv^{pkind\text{-}(q_e;q_d)\text{-}cca}_{\ikem,U'_1}$ such that  
\begin{eqnarray} 
\label{eqnthmcombptx5}
   && |\pr[\mathrm{G}^{3\text{-}0}_{Comb,D'}(\lambda)=1] - \pr[\mathrm{G}^{4\text{-}0}_{Comb,D'}(\lambda)=1]|  
   \leq Adv^{pkind\text{-}(q_e;q_d)\text{-}cca}_{\ikem,U'_1}
\end{eqnarray}

To construct such adversary $U'_1$, we replace line 3 of $U^{O_2}_{12}(st_1,c^*_1,k^*_1)$ algorithm of the adversary $U_1$ in figure~\ref{figikemcombcca22} with $(y^*_1 \xleftarrow{\$} \mathcal{K})$. The proof is similar to claim \ref{claimptx1}.

Finally, we note that since in $\mathrm{G}^{4\text{-}b}_{Comb,D'}$, $y^*_1$ is sampled from uniform distribution, then $k^*$ is uniformly distributed and 
\begin{equation} 
\label{eqnthmcombptx6}
    \pr[\pind^{(q_e;q_d)\text{-}cca\text{-}1}_{Comb,D'}(\lambda)=1]=\pr[\mathrm{G}^{4\text{-}0}_{Comb,D'}(\lambda)=1]
\end{equation}

Now using triangular inequality on inequalities \ref{eqnthmcombptx1} to \ref{eqnthmcombptx6}, we have

{
\begin{eqnarray*}
 Adv^{pkind\text{-}(q_e;q_d)\text{-}cca}_{Comb^{PtX}_{\ikem,\kem},D'} (\lambda)  
 &&=|\pr[\pind^{(q_e;q_d)\text{-}cca\text{-}0}_{Comb,D'}(\lambda)=1] - \pr[\pind^{(q_e;q_d)\text{-}cca\text{-}1}_{Comb,D'}(\lambda)=1]| 
\\
&&\leq 2(Adv^{pkind\text{-}(q_e;q_d)\text{-}cca}_{
\ikem,U_1}(\lambda) + Adv^{(q_d+1)\text{-}PRF}_{F_1,U_2}(\lambda)). 
\end{eqnarray*}} $\qed$


\section{Proof of Theorem~\ref{thm:ceaconst1}}\label{pf:thm:ceaconst1}
We need to prove that the construction~\ref{ikem:cea} satisfies definition~\ref{def:gkem} for chosen encapsulation attack (CEA) security. In response to an encapsulation query, the oracle returns a key $k$ and a ciphertext $c$ to the adversary. Let, after $q_e$ queries, the adversary's received responses be the vector ${\mathbf{w}}^{q_e\text{-}cea}=(w_1^{cea},\cdots,w_{q_e}^{cea})$, where $w_i^{cea}=(k_i,c_i), \forall i \in \{1,\cdots,q_e\}$. The remaining entropy about $\X $ that can be used to extract the secret key is $\tilde{H}_\infty(\X |\Z , {\mathbf{W}}^{q_e\text{-}cea}={\mathbf{w}}^{q_e\text{-}cea})$, where $\Z $ corresponds to $\z $, the attacker's initial information. Now the $i$-th query's response to the adversary is $w_i^{cea}=(k_i,c_i)$, where $c_i=\big(h(\x ,s),s'_i\big)$ and $k_i=h'(\x ,s'_i)$. For the $i$-th response, the RVs $K_i$ and $C_i$ are distributed over $\{0,1\}^\ell$ and $\{0,1\}^t$ respectively. Now using \cite[Lemma 2.2(b)]{DodisORS08}, for RVs $K_i$ and $C_i$, we have $\tilde{H}_\infty(\X |\Z ,\mathbf{W}^{cea}_i)=\tilde{H}_\infty(\X |\Z ,K_i,C_i) \ge \tilde{H}_\infty(\X |\Z ) - \ell - t$. Since $h(\x ,s)$ remains the same in all $q_e$ responses and the challenge, after $q_e$ encapsulation queries, from \cite[Lemma 2.2(b)]{DodisORS08}, we have
{
\begin{align}\label{equn:ceaconst1}\nonumber
    \tilde{H}_\infty(\X |\Z ,\mathbf{W}^{q_e\text{-}cea}) &= \tilde{H}_\infty\left(\X |\Z ,(\mathbf{W}^{cea}_1,\cdots,\mathbf{W}^{cea}_{q_e})\right) \\
    &\ge \tilde{H}_\infty(\X |\Z ) - t - q_e \cdot \ell
\end{align}     
}     

Now since $\tilde{H}_{\infty}(\X |Z^*,h\left(\X ,( S', S)\right)) \ge \tilde{H}_{\infty}(\X |Z^*) - t$, from Lemma~\ref{glhl}, we have 
{
\begin{align}\nonumber
&\Delta\Big(h'(\X , S'), h\left(\X ,( S', S)\right), S', S, Z^*; U_\ell, h\left(\X ,(S', S)\right), S', S, Z^*\Big) \\
&\le \frac{1}{2}\sqrt{2^{-\tilde{H}_{\infty}(\X |Z^{*},h\left(\X ,( S', S)\right))}\cdot 2^
{\ell}} \le \frac{1}{2}\sqrt{2^{-\tilde{H}_{\infty}(\X |Z^*)}\cdot 2^
{\ell + t}} \label{equn4:ceaconst14}. 
\end{align}     
}

Therefore, from equation~\ref{equn:ceaconst1} and noting that $h(\x ,s)$ remains the same in all $q_e$ responses and the challenge, putting $Z^*=(\Z ,\mathbf{W}^{q_e\text{-}cea})$ in inequality~\ref{equn4:ceaconst14}, we have 
{
\begin{align}\nonumber
&\Delta\Big(h'(\X , S'), h(\X , S), S', S, \Z , \mathbf{W}^{q_e\text{-}cea} ;
U_\ell, h(\X , S), S', S, \Z , \mathbf{W}^{q_e\text{-}cea}\Big) \\\nonumber 
&\le \frac{1}{2}\sqrt{2^{(q_e + 1) \ell + t - \tilde{H}_{\infty}(\X |\Z )}} \\\label{eqn:ceaproof} 
&=\frac{1}{2}\sqrt{2^{(q_e + 1) \ell + t - n\tilde{H}_{\infty}(X |Z )}} \\\label{eqn:ceaproofsecurity} 
&\le \sigma
\end{align}     
}

In above, step~\ref{eqn:ceaproof} uses $\tilde{H}_\infty({\X }|{\Z })=n\tilde{H}_\infty(X|Z)$ in Lemma~\ref{lemma:minentropy},  
and the last step~\ref{eqn:ceaproofsecurity}  holds 
if $\ell \le \frac{n\tilde{H}_{\infty}(X |Z ) + 2\log(\sigma) + 2 - t}{q_e+1}$. 
\remove{Therefore, for $\ell \le \frac{n\tilde{H}_{\infty}(X |Z ) + 2\log(\sigma) + 2 - t}{q_e+1}$,
}
 To complete the proof, we use 
\cite[Lemma 1]{sharifian2021information} and ~\cite[Proposition 5.9]{Bellare2005}
,  that relates the statistical distance to
$ Adv^{pkind\text{-}cea}_{\pk,\msa}(\lambda)   $,  
 concluding that the extracted key is $2\sigma$-indistinguishable from random, and we have $2\sigma$-IND-$q_e$-CEA security. $\qed$

\section{Proof of Theorem~\ref{Thm:ikemotsecurity}}\label{pf:Thm:ikemotsecurity}
{\it Correctness (reliability).} We first determine the value of $\nu$ and $t$, and then  compute the extracted secret key length $\ell$. In the decapsulation algorithm $\ikemd(\cdot)$, Bob searches the set $\mathcal{R}$ for $\hat{\x }$ whose hash value matches with the received hash value $v$ and checks whether a unique such $\hat{\x }$ is found. It declares success if a unique $\hat{\x }$ is found in the set $\mathcal{R}$ with such required property. Therefore, the algorithm fails if one of these two events occurs: $(i)$ there is no element $\x $ in the set  $\mathcal{R}$ such that its hash value matches with the received hash value i.e. $\x $ is not in the set $R$, $(ii)$ there are more than one element in the set $R$, whose hash values are equal to the received hash value $v$. Hence, the probability that Bob fails to recover the correct key is upper bounded by the sum of the probabilities of these two events. These two cases corresponds to the events: 
\begin{align} \nonumber
 & \mathcal{E}_{1} = \{\x : \x  \notin \mathcal{R}\}=\{\x : -\log(P_{\X |\Y }(\x |\y )) > \nu\} \text{ and } \\\nonumber
 &\mathcal{E}_2 = \{\x  \in \mathcal{R}: \exists \text{ } \hat{\x } \in \mathcal{R} \text{ s.t. }  h(\x , (s',s)) = h(\hat{\x }, (s',s)\}.
 \end{align}
 
 For any $\epsilon > 0$, choose $\epsilon_1 >0$ and $\epsilon_2 >0 $ such that $\epsilon_1 + \epsilon_2 \le \epsilon$. Let $\epsilon_1 = 2^{\frac{-n{\delta_1}^2}{2\log^2(|\mathcal{X}|+3)}}$
 and $\nu = H(\X |\Y ) + n\delta_1$.
 Now, $\mathsf{Pr}(\mathcal{E}_{1})={\mathsf{Pr}}\Big(-\log(P_{\X |\Y }(\x |\y )) > H(\X |\Y ) + n\delta_1\Big) \le \epsilon_1$ (from~\cite{Holenstein11}, Theorem 2). 
 To bound $\mathsf{Pr}(\mathcal{E}_{2})$, note that since $h$ is a universal  hash family   with input space $\mathcal{X}^n$ and seed space $(\mathcal{S'} \times \mathcal{S})$, 
 for any $\x ,\hat{\x } \in \mathcal{R}$, $\x  \ne \hat{\x }$, $s' \in \mathcal{S'}$ and randomly chosen $s \in \mathcal{S}$, we have $\mathsf{Pr}\left(h(\x ,( s',s))=h(\hat{\x },( s',s))\right) \le 2^{-t}$, where probability is over the random choices $(s',s)$ from $(\mathcal{S'}\times \mathcal{S})$. Thus, $\mathsf{Pr}(\mathcal{E}_2) \le 
 |\mathcal{R}|\cdot 2^{-t}$. 
Equation~\ref{reconset} implies that the probability of each element of $\mathcal{R}$ is lower bounded by $2^{-\nu}$.  Therefore, using equation ~\ref{reconset} and noting that the sum of probability of elements of $\mathcal{R}$ is less than or equal to 1, we have $\frac{|\mathcal{R}|}{2^\nu} \le \mathsf{Pr}(\mathcal{R}) \le 1 \Rightarrow |\mathcal{R}| \le 2^\nu$. Thus, $\mathsf{Pr}(\mathcal{E}_2) \le 
|\mathcal{R}|\cdot 2^{-t} \le 2^{\nu-t}.$ 
Let $t=\nu -\log(\epsilon_2)$, then we have  $\mathsf{Pr}(\mathcal{E}_2) \le \epsilon_2$. Therefore, for $t=H(\X |\Y ) + n\delta_1 -\log(\epsilon_2)$, the probability that Bob fails to recover the correct key is less than or equal to $\mathsf{Pr}(\mathcal{E}_1) + \mathsf{Pr}(\mathcal{E}_2) \le \epsilon_1 + \epsilon_2 = \epsilon$. 
 Moreover, since $\X , \Y $ are generated due to $n$ independent and identical experiments $P_{X_i Y_i Z_i}=P_{X Y Z}$ for all $ i \in \{1, \cdots,n \}$, we have $H(\X |\Y )=nH(X |Y )$.
 Finally, by choosing $\epsilon_1=(\sqrt{n}-1)\epsilon/\sqrt{n}$ and $\epsilon_2=\epsilon/\sqrt{n}$, we conclude that if $\nu = nH(X |Y ) + \sqrt{n} \log (|\mathcal{X}|+3) \sqrt{\log (\frac{\sqrt{n}}{(\sqrt{n}-1)\epsilon})}$ and \\$t \ge nH(X |Y ) + \sqrt{n} \log (|\mathcal{X}|+3) \sqrt{\log (\frac{\sqrt{n}}{(\sqrt{n}-1)\epsilon})} + \log (\frac{\sqrt{n}}{\epsilon})$, then $\mathsf{Pr}(\mathcal{E}_1) + \mathsf{Pr}(\mathcal{E}_2) \le \epsilon$. Thus, the construction~\ref{ikem:cca} is $\epsilon$-correct, and the reliability condition is satisfied.

{\it Security.} To prove chosen encapsulation attack (CEA) security, we need to prove that the construction~\ref{ikem:cca} satisfies definition~\ref{def:gkem}. In response to an encapsulation query, the encapsulation oracle returns a pair of key and ciphertext $(k,c)$ to the adversary. Let the adversary's received responses to its $q_e$ encapsulation queries be the vector ${\mathbf{w}}^{q_e\text{-}cea}=(w_1^{cea},\cdots,w_{q_e}^{cea})$, where $w_i^{cea}=(k_i,c_i), \forall i \in \{1,\cdots,q_e\}$. The remaining entropy about $\X $ is $\tilde{H}_\infty(\X |\Z , {\mathbf{W}}^{q_e\text{-}cea}={\mathbf{w}}^{q_e\text{-}cea})$, where $\Z $ corresponds to $\z $, the attacker's initial information. This remaining entropy about $\X $ is used to extract the key. Now consider the $i$-th query's response $w_i^{cea}=(k_i,c_i)$, where $c_i=\Big(h\big(\x , ( s'_i,s_i)\big),s'_i,s_i\Big)$ and $k_i=h'(\x ,s'_i)$. For the $i$-th response, the RVs $K_i$ and $C_i$ are distributed over $\{0,1\}^\ell$ and $\{0,1\}^t$ respectively. Using \cite[Lemma 2.2(b)]{DodisORS08}, for RVs $K_i$ and $C_i$ and noting that $s'_i,s_i$ are randomly chosen and independent of RV $\X $, we have $\tilde{H}_\infty(\X |\Z ,\mathbf{W}^{cea}_i)=\tilde{H}_\infty(\X |\Z ,K_i,C_i) =\tilde{H}_\infty\big(\X |\Z ,K_i,h\big(\X , ( S'_i,S_i)\big)\big)\ge \tilde{H}_\infty(\X |\Z ) - \ell - t$. Therefore, after $q_e$ encapsulation queries, from \cite[Lemma 2.2(b)]{DodisORS08}, we have
{
\begin{align}\label{equn:ceaconst2}\nonumber
    \tilde{H}_\infty(\X |\Z ,\mathbf{W}^{q_e\text{-}cea}) &= \tilde{H}_\infty\left(\X |\Z ,(\mathbf{W}^{cea}_1,\cdots,\mathbf{W}^{cea}_{q_e})\right) \\
    &\ge \tilde{H}_\infty(\X |\Z ) - q_e (t + \ell)
\end{align}     
}

Now since $\tilde{H}_{\infty}(\X |Z^*,h\left(\X ,( S', S)\right)) \ge \tilde{H}_{\infty}(\X |Z^*) - t$, from Lemma~\ref{glhl}, we have 
{
\begin{align}\nonumber
&\Delta\Big(h'(\X , S'), h\left(\X ,( S', S)\right), S', S, Z^*; U_\ell, h\left(\X ,(S', S)\right), S', S, Z^*\Big) \\
&\le \frac{1}{2}\sqrt{2^{-\tilde{H}_{\infty}(\X |Z^{*},h\left(\X ,( S', S)\right))}\cdot 2^
{\ell}} \le \frac{1}{2}\sqrt{2^{-\tilde{H}_{\infty}(\X |Z^*)}\cdot 2^
{\ell + t}} \label{equn3:ceaconst23}. 
\end{align}     
}
Therefore, from inequality~\ref{equn:ceaconst2} and putting $Z^*=(\Z ,\mathbf{W}^{q_e\text{-}cea})$ in inequality~\ref{equn3:ceaconst23}, we have 
{
\begin{align}\nonumber
&\Delta\Big(h'(\X , S'), h\left(\X , (S', S)\right), S', S, \Z , \mathbf{W}^{q_e\text{-}cea} ; 
U_\ell, h\left(\X ,(S', S)\right), S', S, \Z , \mathbf{W}^{q_e\text{-}cea}\Big) \\\nonumber 
&\le \frac{1}{2}\sqrt{2^{-\left(\tilde{H}_{\infty}(\X |\Z )- q_e (t + \ell)\right)}\cdot 2^{\ell + t}} \\\nonumber
&=  \frac{1}{2}\sqrt{2^{(q_e + 1)(t + \ell)-\tilde{H}_{\infty}(\X |\Z )}} \\\label{eqn:ccaproof}
&=\frac{1}{2}\sqrt{2^{(q_e + 1)(t + \ell)-n\tilde{H}_{\infty}(X |Z )}} \\\label{eqn:ccaproofsec}
& \le \sigma
\end{align}     
}

The equality~\ref{eqn:ccaproof} follows from Lemma~\ref{lemma:minentropy} that proves $\tilde{H}_\infty({\X }|{\Z })=n\tilde{H}_\infty(X|Z)$. The  inequality~\ref{eqn:ccaproofsec} holds if \\ $\ell \le \frac{n\tilde{H}_{\infty}(X |Z ) + 2\log(\sigma) + 2}{q_e+1} - t$. To complete the proof, we use 
\cite[Lemma 1]{sharifian2021information} and ~\cite[Proposition 5.9]{Bellare2005}
,  that relates the statistical distance to
$ Adv^{pkind\text{-}cea}_{\pk,\msa}(\lambda)   $, concluding that the extracted key is $2\sigma$-indistinguishable from random, and  we have $2\sigma$-IND-$q_e$-CEA security. $\qed$

\section{Proof of Lemma~\ref{lemma:uhf}.}\label{appn:uhfproof}
\begin{proof}
 We show that $h$ satisfies Definition~\ref{defn:uhf}. 
    Let $\x$ and $\y$ be  
    such 
    that $\x \ne \y$.  We need to show that \\$\Pr[h(\x,(S',S))=h(\y,(S',S))] \le \frac{1}{2^t}$, where the probability is over the uniformly random choices of $(\mathcal{S'} \times \mathcal{S})$, $\mathcal{S'}=GF(2^w)$ and $\mathcal{S}=GF(2^{n-t})\times GF(2^t)$. Note that $s=(s_2,s_1)$ with $s_2 \in GF(2^{n-t})$ and $s_1 \in GF(2^t)$.

    Since $\x \ne \y$, we have $ (\x_2 \parallel \x_1) \ne (\y_2 \parallel \y_1)$.

    {\bf Case 1.} Let $\x_1 \ne \y_1$. 
     For fixed values of $s'=(s'_1,\cdots,s'_r) \in (GF(2^{n-t}))^r$ and $s_2 \in GF(2^{n-t})$,  there is a unique value of $s_1 $ for which we have,
\begin{align} \nonumber
 &\big[(\x_{2})^{r+3} + {\sum}_{i=1}^r s'_i(\x_{2})^{i+1} +s_2 \x_{2}\big]_{1\cdots t}+ (\x_{1})^{3} + s_1 \x_{1} 
 =\big[(\y_{2})^{r+3} + {\sum}_{i=1}^r s'_i(\y_{2})^{i+1} +s_2 \y_{2}\big]_{1\cdots t}+ (\y_{1})^{3} + s_1 \y_{1} \\\label{eqn:uhfp}
 &\Leftrightarrow s_1 (\x_{1} - \y_1) = \big[(\y_{2})^{r+3} + {\sum}_{i=1}^r s'_i(\y_{2})^{i+1} +s_2 \y_{2}\big]_{1\cdots t}
+ (\y_{1})^{3} - 
\big[(\x_{2})^{r+3} + {\sum}_{i=1}^r s'_i(\x_{2})^{i+1} +s_2 \x_{2}\big]_{1\cdots t} 
- (\x_{1})^{3} 
 \end{align} 
 \remove{This implies that there is exactly one value of $s_1$ for which the above equation~\ref{eqn:uhfp} holds true, namely 
\begin{align}\nonumber
    &s_1  = \Bigg[\big[(\y_{2})^{r+3} + {\sum}_{i=1}^r s'_i(\y_{2})^{i+1} +s_2 \y_{2}\big]_{1\cdots t} + (\y_{1})^{3} - \\\nonumber
&\qquad\quad \big[(\x_{2})^{r+3} + {\sum}_{i=1}^r s'_i(\x_{2})^{i+1} +s_2 \x_{2}\big]_{1\cdots t}  \\\label{eqn:uhfp1}
&\qquad\quad- (\x_{1})^{3}\Bigg](\x_{1} - \y_1)^{-1}
\end{align}
 
}


Therefore, for a random choice of $(s',s)$, we have that  $\Pr[h(\x,(S',S))=h(\y,(S',S))]$   is given by $\frac{1}{2^t}$.\\

 {\bf Case 2.} Let $\x_2 \ne \y_2$. 
 
 For fixed values of $s'=(s'_1,\cdots,s'_r) \in (GF(2^{n-t}))^r$ and $s_1 \in GF(2^{t})$,  there is a unique value of $\big[s_2 (\x_{2} - \y_2)\big]_{1\cdots t} $ for which we have,
 %
\begin{align} \nonumber
 &\big[(\x_{2})^{r+3} + {\sum}_{i=1}^r s'_i(\x_{2})^{i+1} +s_2 \x_{2}\big]_{1\cdots t}+ (\x_{1})^{3} + s_1 \x_{1} 
 =\big[(\y_{2})^{r+3} + {\sum}_{i=1}^r s'_i(\y_{2})^{i+1} +s_2 \y_{2}\big]_{1\cdots t}+ (\y_{1})^{3} + s_1 \y_{1} \\\nonumber
 &\Leftrightarrow \big[s_2 (\x_{2} - \y_2)\big]_{1\cdots t} = \big[(\y_{2})^{r+3} + {\sum}_{i=1}^r s'_i(\y_{2})^{i+1} \big]_{1\cdots t}
+ (\y_{1})^{3} + s_1 \y_{1} -  
\\\label{eqn:uhfp2}
&\qquad\qquad\qquad\qquad\qquad   \big[(\x_{2})^{r+3} + {\sum}_{i=1}^r s'_i(\x_{2})^{i+1} \big]_{1\cdots t} 
- (\x_{1})^{3} -  s_1 \x_{1}
 \end{align}

 For every $\big[s_2 (\x_{2} - \y_2)\big]_{1\cdots t} $, there are $2^{n-2t}$ values of $s_2 (\x_{2} - \y_2)$, where each,  for  fixed $(\x_{2} - \y_2)$, determines a single value for $s_2$. 
Thus there are  exactly $2^{n-2t}$ values of $s_2$ for which the above equation~\ref{eqn:uhfp2} holds true. 

 Thus, for a random choice of $(s',s)$ the probability of collision in this case is, 
 exactly $\frac{2^{n-2t}}{2^{n-t}}=\frac{1}{2^t}$. 


 Therefore, $h$ is a universal hash family. 
\end{proof}


\section{B\'{e}zout's  theorem~\cite{Coolidge1959,BEZOUTtheorem}}
\label{appn:bezout}
{\bf B\'{e}zout's  Theorem~\cite{Coolidge1959,BEZOUTtheorem}.} In general, two algebraic curves of degree $m$ and $n$ can intersect in $m\cdot n$ points and cannot meet in more than $m \cdot n$ points unless they have a common factor (i.e. the two equations have a common factor).

Moreover, $N$ polynomial equations of degrees $n_1,n_2,\cdots,n_N$ in $N$ variables have in general  $n_1n_2\cdots n_N$ common solutions.

\section{CEA secure iKEM protocol of Sharifian et al.~\cite{sharifian2021information}}\label{app:setalgotheo}

\begin{definition}[strongly universal  hash family]\label{defn:suhash}
A family of hash functions $h: \mathcal{X} \times \mathcal{S} \rightarrow \mathcal{Y}$ is called a strongly universal  hash family if for all $x \ne y$, and any $a,b \in \mathcal{Y}$, $\pr[h(x,S)=a \wedge h(y,S)=b]=\frac{1}{|\mathcal{Y}|^2}$, where the probability is  over the uniform choices over $\mathcal{S}$.  
\end{definition}

We briefly recall the construction of CEA secure iKEM protocol due to Sharifian et al.~\cite{sharifian2021information}. .  \begin{construction}
 The iKEM $iK_{OWSWA}$'s three algorithms $(Gen,Encap,Decap)$ are as follows:    
 The protocol is designed for preprocessing model in which Alice, Bob and Eve have $n$ components of the source $(\X ,\Y , \Z )$ respectively according to a distribution $P_{\X \Y \Z}$. 
 The protocol uses two strongly universal  hash families: $h: \mathcal{X}^n \times \mathcal{S} \to \{0,1\}^t$ and $h': \mathcal{X}^n \times \mathcal{S'} \to \{0,1\}^\ell$. $\mathcal{C}=\{0,1\}^t \times \mathcal{S'} \times{S}$ and $\mathcal{K}=\{0,1\}^{\ell}$ denote the ciphertext space and key space respectively. 
 \begin{enumerate}
     \item  $Gen(P_{\X \Y \Z})$. A trusted sample samples the distribution $P_{\X \Y \Z}$  independently $n$ times and gives $\x$, $\y$ and $\z$ privately to Alice, Bob and Eve respectively.
     \item $Encap(\x)$. The encapsulation algorithm takes Alice's private input $\x$, randomly sample the seeds $s' \xleftarrow{\$}\mathcal{S'}$ and $s \xleftarrow{\$}\mathcal{S}$ for two strongly universal  hash families $h'$ and $h$ respectively. It generates the key $k=h'(\x,s')$ and the ciphertext $c=(h(\x,s),s',s)$.
     \item $Decap(\y,c)$. The decapsulation algorithm takes Bob's private key $\y$ and the ciphertext $c$. It parses $c$ as $(g,s',s)$, where $g$ is a $t$-bit string. It defines a set $\mathcal{T(\X|\y)}=\{\x :-\log(P_{\X |\Y }(\x |\y )) \le \nu \} $, and for each vector $\hat{\x} \in \mathcal{T(\X|\y)}$ checks whether $g=h(\hat{\x},s)$. The decapsulation algorithm outputs the key $h'(\hat{\x},s')$ if there is a unique $\hat{\x}$  that satisfies $g=h(\hat{\x},s)$; otherwise,  it outputs $\perp$.
 \end{enumerate}
\end{construction}

\end{document}